\documentclass{article}
\usepackage{arxiv}
\usepackage{graphicx} 

\usepackage{amsfonts, amsmath, amsthm}
\usepackage{hyperref}
\usepackage{mathtools}
\usepackage{xcolor}
\usepackage{subcaption}

\usepackage{placeins}

\usepackage{pgfplots}
\usepackage{pgfplotstable}
\pgfplotsset{compat=1.18}
\usepgfplotslibrary{groupplots}

\definecolor{ferngreen}{HTML}{789024}
\definecolor{winered}{HTML}{A63852}
\definecolor{perfumepurple}{HTML}{c0affb}
\definecolor{apricotorange}{HTML}{e6a176}
\definecolor{orientblue}{HTML}{00678a}
\definecolor{downygreen}{HTML}{5eccab}

\usepackage{cleveref} 

\newtheorem{remark}{Remark}
\newtheorem{definition}{Definition}
\newtheorem{corollary}{Corollary}
\newtheorem{assumption}{Assumption}

\crefname{remark}{Remark}{Remarks}
\crefname{definition}{Definition}{Definitions}
\crefname{corollary}{Corollary}{Corollaries}
\crefname{assumption}{Assumption}{Assumptions}

\usepackage[backend=biber,style=numeric-comp,url=false,isbn=false]{biblatex}
\AtEveryBibitem{\clearfield{month}}
\AtEveryBibitem{\clearfield{day}}

\DeclareMathOperator{\Tr}{Tr}

\newtheorem{example}{Example}

\DeclarePairedDelimiterX{\infdivx}[2]{(}{)}{%
  #1\;\delimsize\|\;#2%
}
\DeclareMathOperator*{\argmax}{arg\,max}
\DeclareMathOperator*{\argmin}{arg\,min}
\newcommand{\KL}{D_\mathrm{KL}\infdivx*}
\DeclarePairedDelimiter{\norm}{\lVert}{\rVert}

\title{Likelihood-informed dimension reduction across tempered Bayesian posteriors}

\author{%
    \href{https://orcid.org/0000-0002-2745-1982}{\includegraphics[scale=0.06]{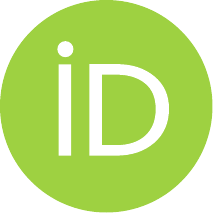}\hspace{1mm}Arne Bouillon}\\
    Department of Computer Science (NUMA)\\
    KU Leuven\\
    Celestijnenlaan 200A, 3001 Leuven, Belgium\\
    \texttt{arne.bouillon@kuleuven.be}\\
    \And
    \href{https://orcid.org/0000-0001-7374-0382}{\includegraphics[scale=0.06]{orcid.pdf}\hspace{1mm}Oliver R. A. Dunbar} \\
	Division of Geological and Planetary Sciences\\
	California Institute of Technology\\
	Pasadena, CA, United States \\
	\texttt{odunbar@caltech.edu}
}

\begin{document}



\maketitle

\begin{abstract}
    Scientific computer simulations cannot represent all scales in realistic applications. To bridge this model-data gap, parameters are injected into models and constrained with noisy data using Bayesian inversion. To reduce the number of simulator evaluations, which can be $10^5$ or more, modern approaches employ dimension reduction in conjunction with emulation of the forward map (that contains the simulator). Due to scarcity of model evaluations and data, this dimension reduction becomes very important for posterior sampling performance. Recent work on likelihood-informed subspaces (LIS) truncates to informative directions by optimizing bounds on information loss, and though mathematically well-adapted to sampling, they are often restrictive in practice.

    In this work, we provably generalize this methodology to facilitate application to  \emph{$\alpha$-tempered} (i.e., \emph{annealed}, \emph{power-posterior}) distributions for $\alpha\in[0,1]$. We provide theory to build partially-informed spaces termed $\alpha$-LIS. We show how $\alpha < 1$ can often produce near-optimal spaces. In addition, we focus on applying $\alpha$-LIS to practical cases, where the available data is severely limited and noisy. We propose and test extensions for utilizing data from the entire sequence of distributions $\alpha_0< \dots < \alpha_k$, and use simple approximations of model gradients so that our approach can be used for emulation of forward maps for chaotic or stochastic systems where derivatives are unavailable or uninformative due to noise. In experiments, our accumulated approach is much more robust to these challenging circumstances than the theoretically optimal $\alpha=1$.
\end{abstract}

\section{Introduction}
The success of machine learning and over-parameterization has allowed new closures and data-driven representations of unresolved processes in scientific computer simulators. However, with these high-dimensional parameters, the curse of dimensionality makes calibration (finding optimal parameters to fit data) and uncertainty quantification (UQ, sampling a joint distribution for likely parameters) challenging. We must use algorithms that scale well to a high-dimensional setting.

Calibration and UQ algorithms require evaluations of the simulators embedded within their loss or likelihood functions. Scientific simulators are frequently expensive black-box software, with a limited number of evaluations available. In these situations, robust dimension reduction techniques (as a way of pre-processing) play an important yet underrepresented role in algorithm performance.

Posterior sampling for such models commonly requires in excess of $10^5$ model evaluations with, e.g., Markov chain Monte Carlo (MCMC)~\cite{Gey11:book}. To offset this, models can be replaced with surrogates that have more amenable properties: they are cheaper, lower-dimensional, or smoother. This includes multi-fidelity or multi-level approaches~\cite{BiHigLee02,ChrFox05,HoaSchStu13} and data-driven emulators of either the model~\cite{Hig_etal04,Jin08,BloFraNarWuk10,CleGarLanSchStu21} (a focus of our investigation) or the loss/likelihood function~\cite{MarXiu09,TecStu18,BegHerLou20}. In data-driven emulation, a design question arises: how does one select input-output pairs to train an emulator, without falling victim to the curse of dimensionality? Goal-independent strategies such as space-filling designs, Latin-hypercube sampling, and sparse grids~\cite{CheLinSim01,CroDhaLae11} may suffer in high-dimensional problems. Meanwhile, other methods show more success by addressing a goal-oriented design problem with active or sequential approaches~\cite{NowSin17, CleGarLanSchStu21, PluSurWil24,CatChaHolKra26_pre}, and some of these are themselves enhanced with emulators (e.g., Bayesian history matching~\cite{DroNotPag21, KinManShe24}).

In the recent \emph{calibrate, emulate, sample} approach~\cite{CleGarLanSchStu21,Dun_etal24}, the authors propose to address the training point design for UQ with variants of ensemble Kalman inversion (EKI)~\cite{CheOli12,IglLawStu13,SchStu17,ChaCheSan20,HuaHuaReiStu22,CalReiStu25, Dun_etal22}. EKI is an optimization and sampling method that scales to high dimensions with relatively few evaluations and can handle complex loss functions~\cite{DunGjiMorSch25_pre,KovStu19}. The strict assumptions (linearity and Gaussianity) underlying the updates of EKI make it unsuitable for accurate posterior sampling, but it can efficiently place input-output training pairs in regions of high posterior mass. The heuristic is that resulting emulators have high accuracy in these regions, where they most often make predictions within a sampling algorithm. Despite this tailored sampling, training the high-dimensional emulator is still challenging: only about $10^2$--$10^3$ training pairs are realistically available, and emulators such as Gaussian processes (GPs) are expensive in high dimensions.

Motivated by this final challenge, we will investigate the most efficient use of samples generated from a sequence of approximating distributions for the purpose of dimension reduction in sampling. We consider the task of posterior sampling and, later, emulator-accelerated posterior sampling, and measure performance in terms of distances between the reduced and full posterior distributions. Our algorithms, based on optimization of information-theoretic bounds on the posterior, are known as likelihood-informed subspace (LIS) methods, with data-based (DB) and data-free (DF) variants~\cite{CuiMarMarSolSpa14, BapMarZah22_pre, CuiZah21, CuiTon22}. The information-theoretic approach is quite general, and is directly linked to active subspaces for inference~\cite{Con15:book,10.1137/15M1042127,NguWittBui22} and prior-preconditioned Hessian approaches~\cite{Fla_etal11, BurGarMarWil12}. It naturally generalizes canonical correlation analysis (CCA)---see, e.g.,~\cite{BapMarZah22_pre}. A conceptually related approach for dimension reduction of input and output variables---however, in linear dynamical systems instead of Bayesian inverse problems---is known as \emph{balanced truncation}~\cite{1102568,1084254}. It was specialized to the case of Bayesian inversion in~\cite {qian2022model}. New work also directly investigates the differences in prior-to-posterior covariance~\cite{PolMaiSocGes26_pre}, which is suitable for intractable likelihoods.

Motivated by sequential ensemble methodologies, we construct a flexible generalization of current LIS methods to any of the 1-parameter family of \emph{$\alpha$-tempered} posteriors ($\alpha\in [0,1]$). We prove that this generalizes DB-LIS ($\alpha=1$, linked to the posterior distribution) and DF-LIS ($\alpha=0$, linked to the prior), and term this generalization $\alpha$-LIS. It can be defined for any $\alpha$ as a linear map from the full to a reduced space, optimizing an upper bound on the Kullback--Leibler divergence between the reduced and full posterior. We demonstrate its interpolative properties over a series of numerical studies to illustrate algorithm performance, when one has draws from distributions that are neither the prior ($\alpha=0$) nor the posterior ($\alpha=1$). The numerical studies consider different data sample sizes, for linear and nonlinear problems, with direct sampling or by first training an emulator. The nonlinear problems include a PDE simulator (solving Darcy flow inversion) and a chaotic Lorenz `96 model, solving an inversion for a spatially-dependent forcing. 

Beyond proof-of-theory numerics, we focus on applicability to practical, highly constrained situations. We propose practical extensions for some key limitations in current LIS methodologies. \begin{itemize}
    \item First, we propose how one can compute gradient-free approximations of the $\alpha$-LIS subspaces, based on available samples. This addresses the fact that many models must be treated without gradients, as these might be unavailable (e.g., in a legacy code) or uninformative (such as for chaotic or stochastic simulators). Our simple approximation is to take a global statistical linearization over the samples (as presented in, e.g.,~\cite{ChaCheSan20,ungaralaIteratedFormsKalman2012,CalReiStu25}).
    
    Other gradient-free approaches include score approximation~\cite{BapBreMar24_pre} for LIS, as well as traditional finite differences (for active subspaces, see~\cite{ConGlei14}). Covariance-informed subspace methods are also naturally gradient-free~\cite{PolMaiSocGes26_pre}.

    \item Second, we propose a way to combine operators that describe $\alpha$-LIS spaces from a sequence of $\alpha$-tempered posteriors to produce a suitable unified subspace. This extension seeks to leverage all possible information from a sequence of distributions, as tempered algorithms are iterative and necessarily generate a full sequence of distributions, not just one.
\end{itemize}

In many applications, practitioners use off-the-shelf approaches for dimension reduction that optimize a reconstruction error (e.g., projecting onto principal components or nonlinear basis functions). This often optimizes the input and output reductions independently, rather than utilizing their combined information. We compare our approaches to PCA-based reduction in all investigations as an affine benchmark.

Sampling from a sequence of approximating distributions is not limited to the ensemble approaches we consider here (e.g.,~\cite{IglParTre18,IglYan21}), and so our dimension reduction approach can apply to other methodologies such as simulated annealing~\cite{MarPar92,DelMonSus18}, annealed importance sampling and sequential Monte Carlo~\cite{Nea96, Nea01,DelDouJas06,ChoPap20}, thermodynamic integration~\cite{LarPhi06,FriPet08}, and parallel tempering~\cite{EarDee05}. One simply needs to generate samples of the different tempered posterior distributions.

\subsection{Our contributions}
In this investigation, we make contributions to theory and practice to improve affine dimension reduction for inverse problems:
\begin{enumerate}
    \item[(C1)] We define a generalized $\alpha$-LIS, a 1-parameter family that bridges between DF-LIS ($\alpha=0$) and DB-LIS ($\alpha=1$) to compute a LIS from any $\alpha$-tempered posterior distribution. We propose extensions of the $\alpha$-LIS framework to handle important practical settings such as non-differentiable likelihoods and accumulating data from a sequence of distributions for $0=\alpha_0<\dots< \alpha_k\leq1$ (as available from a tempering algorithm). 
    \item[(C2)] We perform a broad assessment with numerical experiments. For large sample sizes, we validate the theory and explore the accuracy-computation tradeoff. We show that near-optimal performance can often be seen for $\alpha < 1$, and sometimes significantly smaller. For more moderate sample sizes and in the presence of noise or gradient approximations, we observe that optimal numerical performance is (surprisingly) found for $0<\alpha < 1$, as $\alpha=1$ becomes numerically non-robust. The observed tradeoff for large sample size is replaced with a ``best of both worlds" performance in accuracy and robustness for moderate sample sizes. Especially the accumulated approach performs well.
    \item[(C3)] We implement a concrete application of these algorithms to aid within the \emph{calibrate, emulate, sample} framework for accelerated derivative-free sampling. Here an ensemble Kalman method is used to approximate tempered distributions. Following this, we create likelihood-informed subspaces that aid in training an emulator that can be used within a traditional sampling framework to approximate the full posterior. Such a framework leverages the data-efficiency and robustness of ensemble Kalman methods, and uses this information judiciously to create the best subspace to aid the less robust and scalable problem of supervised learning with limited data. We apply this successfully to sample the posterior of parameters arising in a chaotic Lorenz `96 system.
\end{enumerate}

\section{Bayesian inverse problems and dimension reduction} \label{sec:bip}
We are solving a Bayesian inverse problem formulated as
\begin{equation}\label{eq:bip:ip}
    Y = G(X) + \eta, \qquad \eta \sim \mathcal N(0, \Gamma),
\end{equation}
where $X\in\mathbb R^{d_x}$ are model parameters and $Y\in\mathbb R^{d_y}$ is an observation, both treated as random variables. Here, $\mathcal N(0, \Gamma)$ denotes a multivariate normal distribution with zero mean and symmetric positive definite covariance matrix~$\Gamma$. We will write $\pi_Z$ for the probability density function of a random variable $Z$. Consider a prior distribution~$\pi_X(x)$ and the likelihood
\begin{equation} \label{eq:bip:likelihood}
    \pi_{Y\mid X}(y\mid x) \propto \exp\left(-\frac12\norm{y-G(x)}_\Gamma^2\right)
\end{equation}
implied by \eqref{eq:bip:ip}, where $\norm z_\Gamma\coloneqq\sqrt{z^\top \Gamma^{-1}z}$ denotes the Mahalanobis norm. The broad aim is to draw samples from the posterior distribution $\pi_{X\mid Y}(x\mid y^\dagger)$, where $y^\dagger$ is a known observation. The posterior can be expressed by Bayes' rule in terms of the prior and likelihood as
\begin{equation} \label{eq:bip:post}
    \pi_{X\mid Y}(x\mid y) = \frac{\pi_X(x)\pi_{Y\mid X}(y\mid x)}{\pi_Y(y)} \propto \pi_X(x)\pi_{Y\mid X}(y\mid x).
\end{equation}
There are many algorithms to (approximately) sample from $\pi_{X\mid Y}(x\mid y^\dagger)$. Part of them work by evolving samples from the prior $\pi_X$ to the posterior, taking steps through a family of \emph{tempered} distributions
\begin{equation}
   \pi^{(\alpha)}_{X\mid Y}(x\mid y^\dagger)\propto \pi_X(x)\pi_{Y\mid X}(y^\dagger\mid x)^\alpha
\end{equation}
with $\alpha\in[0,1]$.

\begin{example}[Ensemble Kalman inversion] \label{ex:eki}
    For a given ensemble size $J\geq 3$, and initial pseudo-time index $n=0$, sample $x_0\coloneqq\{x_0^{(j)}\}_{j=1}^J$ from the prior distribution. We define ensemble Kalman inversion (EKI) as the following iterative updates applied to this ensemble for $n=0,\dots, N-1$:
    \begin{equation}\label{eq:eki}
        x^{(j)}_{n+1} = x^{(j)}_{n} + C^{xG}_n(C^{GG}_n+\Delta t_n^{-1}\Gamma)^{-1}(y^{(j)} - G(x^{(j)}_n)), \quad \text{ where} \quad y^{(j)} = y^\dagger + \sqrt{\Delta t_n^{-1}} \gamma^{(j)}, \quad \gamma^{(j)}\sim N(0,\Gamma).
    \end{equation}
    Here, the empirical means and covariances between the ensemble members are given by
    \begin{align*}
     C^{xG}_n &= \frac{1}{J}\sum_{j=1}^J \Big[ (x^{(j)}_n - \overline{x_n})\otimes(G(x^{(j)}_n) - \overline{G(x_n)}) \Big], \qquad &\overline{x}_n &= \frac{1}{J}\sum_{j=1}^J x^{(j)}_n,
        \\
        C^{GG}_n &= \frac{1}{J}\sum_{j=1}^J \left[(G(x^{(j)}_n) - \overline{ G(x_n)})\otimes(G(x^{(j)}_n) - \overline{ G(x_n)}) \right],\qquad &\overline{G(x_n)} &= \frac{1}{J} \sum_{j=1}^J G(x^{(j)}_n).
    \end{align*}
    This is a particle discretization of a statistically optimal update for linear forward maps with Gaussian prior and noise distributions. The pseudo-timesteps $0<\Delta t_n\leq 1$ are introduced to help handle nonlinearity through iteration and the algorithm is terminated at $N$ such that the $\Delta t_n$ satisfy
    \begin{equation}
        \sum_{n=0}^N\Delta t_n = 1.
    \end{equation}
    Prior work~\cite{IglParTre18,IglYan21,HuaHuaReiStu22} illustrates how this algorithm sequentially (and approximately) draws ensembles from the tempered posterior:
    \begin{equation}
        x^{(j)}_k \overset{\mathrm{approx.}}\sim \pi^{(\alpha_k)}_{X\mid Y}(x\mid y^\dagger)\qquad\text{with } \alpha_k \coloneqq \sum_{n=1}^k\Delta t_n.
    \end{equation}
    Effective adaptive methods for selecting the pseudo-timestep are available (e.g.,~\cite{IglYan21,KovStu19}), so that at termination time ($\alpha_N=1$) the final ensemble $x^{(j)}_N$ approximate samples of the posterior. Many variants of this algorithm exist, and for computational reasons, we will complete experiments in this paper using the ensemble transform Kalman inversion (ETKI) algorithm~\cite{DunGjiMorSch25_pre,BisEthSha01,huangIteratedKalmanMethodology2022c}. We do not present its update here; it reformulates the linear algebra of \eqref{eq:eki} to avoid storing or solving the linear systems with $(C^{GG}_n+\Delta t^{-1}_n \Gamma)$, thus providing advantageous scaling in the data space.
\end{example}

In this work, we are interested in approximately sampling from \eqref{eq:bip:post} by applying sampling algorithms in a reduced dimension $(r, s)$ instead of the full dimension $(d_x, d_y)$. Specifically, we investigate linear reductions: we look for orthogonal matrices $[U_r, U_\perp]\in\mathbb R^{d_x\times d_x}$ and $[V_s, V_\perp]\in\mathbb R^{d_y\times d_y}$, where $U_r$ and $V_s$ have $r$ and $s$ columns, respectively. Introducing
\begin{subequations} \label{eq:bip:XrYs}
\begin{align}
    X_r &\coloneqq U_r^\top X, &X_\perp &\coloneqq U_\perp^\top X,\\
    Y_s &\coloneqq V_s^\top Y, &Y_\perp &\coloneqq V_\perp^\top Y,
\end{align}
\end{subequations}
such that $X=U_rX_r+U_\perp X_\perp$, $Y=V_sY_s + V_\perp Y_\perp$, we define the approximate likelihood
\begin{equation} \label{eq:bip:pi*}
    \pi_{Y\mid X}^*(y\mid x) \coloneqq \pi_{Y_s\mid X_r}(y_s\mid x_r),
\end{equation}
where we use the marginal likelihood
\begin{equation} \label{eq:bip:marg-likelihood}
\begin{aligned}
    \pi_{Y_s\mid X_r}(y_s\mid x_r) &= \iint \pi_{Y\mid X}(y\mid x)\pi_{X_\perp\mid X_r}(x_\perp\mid x_r)\,\mathrm dx_\perp\,\mathrm dy_\perp\\
    &\propto \int\exp\left(-\frac12\norm{y_s - V_s^\top G(x)}_{V_s^\top\Gamma V_s}^2\right)\pi_{X_\perp\mid X_r}(x_\perp\mid x_r)\,\mathrm dx_\perp.
\end{aligned}
\end{equation}
Likewise, we can introduce an approximate posterior $\pi_{X\mid Y}^*(x\mid y^\dagger)$ with
\begin{equation} \label{eq:bip:approx-post}
    \pi_{X\mid Y}^*(x\mid y) \coloneqq \frac{\pi_{Y\mid X}^*(y\mid x)\pi_X(x)}{\int \pi_{Y\mid X}^*(y'\mid x)\pi_X(x)\,\mathrm dy'} = \frac{\pi_{Y\mid X}^*(y\mid x)\pi_X(x)}{\pi_{Y_s}(y_s)};
\end{equation}
the last equality is proven in~\cite[p.\ 34]{BapMarZah22_pre}. Sampling $x$ from the approximate posterior $\pi_{X\mid Y}^*(\cdot\,\mid y^\dagger)$ is done by first sampling $x_r$ with density $\propto \pi_{Y_s\mid X_r}(y_s^\dagger\mid\,\cdot)\pi_{X_r}(\cdot)$---an $(r, s)$-dimensional problem---and then conditionally sampling $x_\perp$ from $\pi_{X_\perp\mid X_r}(\cdot\,\mid x_r)$. In practice, the likelihood $\pi_{Y_s\mid X_r}$ from \eqref{eq:bip:marg-likelihood} itself is also approximated with nested Monte Carlo sampling from $\pi_{X_\perp\mid X_r}$, but this typically needs very few samples~\cite{CuiTon22}. This is expected: by design, $X_\perp$ has a limited impact on the likelihood term.

Henceforth, we often relax subscripts of $\pi$ in notation to improve readability.

\begin{remark}
   In \cite{CuiTon22}, two other ways to define an approximate likelihood $\pi^*_{Y\mid X}$ are mentioned: instead of marginalizing over the likelihood, they also propose \begin{itemize}
        \item the square of the marginalized square-root of the likelihood, and
        \item the exponential of the marginalized logarithm of the likelihood.
    \end{itemize}
    These all have different properties. In addition, one could consider using the likelihood evaluated in
    \begin{equation}
        \int G(x)\pi_{X_\perp\mid X_r}(x_\perp\mid x_r)\,\mathrm dx_\perp.
    \end{equation}

    In this paper, we always sample \eqref{eq:bip:marg-likelihood} with only one Monte Carlo sample from $\pi_{X_\perp\mid X_r}(\cdot\,\mid x_r)$, motivated by settings where $G$ is very expensive. This makes all four approaches equivalent.
\end{remark}

\subsection{PCA-based dimension reduction}
A simple and common way to find candidates for $U_r$ and $V_s$ is to look at the priors $\pi_X$ and $\pi_Y$, our beliefs about $X$ and $Y$ before making an observation, and keep only those directions that show the greatest variance. This is a principal-component analysis (PCA) of $X$ and $Y$. Specifically, $U_r$ and~$V_s$ are set to the leading $r$ and $s$ eigenvectors of $\mathrm{Cov}(X)$ and $\mathrm{Cov}(Y)$, respectively. These are easily computed: the covariance of $X$ is usually readily available and that of $Y$ can be estimated by combining Monte Carlo sampling from $G(X)$ with the known covariance $\Gamma$.

A major limitation of this approach, and nonlinear variants thereof, is that it does not consider the wider context of the problem: approximating $\pi(\cdot\mid y^\dagger)$. This is clearest in the input space: $U_r$ in no way depends on $G$ and $\Gamma$. Should $G$ vary sharply along a direction with low prior variance, PCA will nevertheless discard this direction. The relation between $X$ and $Y$ is not considered; instead, their spaces are reduced in isolation.

\subsection{Likelihood-informed subspaces with diagnostic matrices}
The problem of finding $U_r$ and $V_s$ with the goal that $\pi_{X\mid Y}^*$ be close to $\pi_{X\mid Y}$ has received significant attention. Most works beyond PCA focus on reducing only the input space, such that in the output space $s=d_y$ and~$V_s=I$. While existing results are typically formulated for more general likelihoods than the Gaussian \eqref{eq:bip:likelihood}, we will specialize those methods to the Gaussian case in our exposition. Non-Gaussian priors and non-linear forward maps $G$ are still included in our discussion.
\begin{remark}[Whitening] \label{rem:intro:white}
    In this subsection we describe dimension reduction applied to the whitened variables
    \begin{subequations} \label{eq:intro:white}
    \begin{align}
        \bar X &\coloneqq \Gamma_0^{-1/2}X,\\
        \bar Y &\coloneqq \Gamma^{-1/2}Y,
    \end{align}
    \end{subequations}
    where $\Gamma_0$ is the covariance\footnote{Using $\bar X\coloneqq\Gamma_0^{-1/2}X$ is only a heuristic choice and can straightforwardly be replaced by other whitening techniques.} of $\pi_X$. They are related as
    \begin{equation}
        \bar Y = \bar G(\bar X) + \bar\eta \coloneqq \Gamma^{-1/2}G(\Gamma_0^{1/2}\bar X) + \Gamma^{-1/2}\eta,
    \end{equation}
    such that $\bar\eta\sim\mathcal N(0,\bar\Gamma)$ with $\bar\Gamma=I$, and $\bar X$ has covariance $\bar\Gamma_0=I$, although it need not be Gaussian. Applying the dimension reduction algorithm to the problem $\bar Y=\bar G(\bar X)+\bar\eta$ results in matrices $\bar U_r$ and $\bar V_s$ that define reduced variables $\bar X_r \coloneqq \bar U_r^\top \bar X = \bar U_r^\top  \Gamma_0^{-1/2}X$ and $\bar Y_s \coloneqq \bar V_s^\top \bar Y = \bar V_s^\top  \Gamma^{-1/2}Y$. Recalling \eqref{eq:bip:XrYs}, we want
    \begin{subequations}
    \begin{align}
        \mathrm{ker}(U_r^\top ) &= \mathrm{ker}(\bar U_r^\top \Gamma_0^{-1/2}) &\Leftrightarrow &&\mathrm{Col}(U_r) &= \mathrm{Col}(\Gamma_0^{-1/2}\bar U_r),\\
        \mathrm{ker}(V_s^\top ) &= \mathrm{ker}(\bar V_s^\top \Gamma^{-1/2}) &\Leftrightarrow &&\mathrm{Col}(V_s) &= \mathrm{Col}(\Gamma^{-1/2}\bar V_s).
    \end{align}
    \end{subequations}
    The matrices $U_r$ and $V_s$ can then be computed by orthogonalizing $\Gamma_0^{-1/2}\bar U_r$ and $\Gamma^{-1/2}\bar V_s$, respectively.
\end{remark}

\subsubsection{Input space.}
For the input space, there have been a variety of definitions of \emph{diagnostic matrices}, matrices whose leading $r$ eigenvectors can be used as a good subspace $\bar U_r$. For instance,~\cite{BapMarZah22_pre,CuiZah21} propose to use
\begin{equation} \label{eq:intro:HX}
    H_{\bar X} \coloneqq \int \Gamma_0^{1/2}\nabla G(x)^\top \Gamma^{-1}\nabla G(x)\Gamma_0^{1/2}\pi(x)\,\mathrm dx.
\end{equation}
Notice that the matrix does not depend on $y^\dagger$ and thus results in a data-independent reduced space. On the other hand,~\cite{CuiMarMarSolSpa14} proposes a data-dependent diagnostic matrix through integration of the posterior,
\begin{equation} \label{eq:intro:HXdag1}
    \int\Gamma_0^{1/2}\nabla G(x)^\top \Gamma^{-1}\nabla G(x)\Gamma_0^{1/2}\pi(x\mid y^\dagger)\,\mathrm dx,
\end{equation}
while~\cite{10.1137/15M1042127} considers data-dependence explicitly through the integrand, but integrates over $\pi(x)$:
\begin{equation} \label{eq:intro:HXdag2}
    \int \Gamma_0^{1/2}\nabla G(x)^\top \Gamma^{-1}(y^\dagger-G(x))(y^\dagger-G(x))^\top \Gamma^{-1}\nabla G(x)\Gamma_0^{1/2}\pi(x)\,\mathrm dx
\end{equation}
Finally,~\cite{zahm2022certified} considers data-dependence both in the integrand and through integrating over the posterior:
\begin{equation} \label{eq:intro:HXdag3}
    H_{\bar X\mid y^\dagger} \coloneqq \int \Gamma_0^{1/2}\nabla G(x)^\top \Gamma^{-1}(y^\dagger-G(x))(y^\dagger-G(x))^\top \Gamma^{-1}\nabla G(x)\Gamma_0^{1/2}\pi(x\mid y^\dagger)\,\mathrm dx.
\end{equation}
Again, $\bar U_r$ then consists of the $r$ leading eigenvectors of these matrices.

These formulations create two classes. The data-free likelihood-informed method (DF-LIS), using $H_{\bar X}$, is advantageous due to its data-independence; it can be computed once and then used for multiple problem instances with different~$y^\dagger$. It is also naturally simple to apply, as $\pi(x)$ is quick to sample from. On the other hand, the data-based methods (DB-LIS), using \eqref{eq:intro:HXdag1}, \eqref{eq:intro:HXdag2}, or \eqref{eq:intro:HXdag3}, involve information from $y^\dagger$ and may result in lower approximation errors for any one problem. Integrating over the posterior $\pi(x\mid y^\dagger)$, required for \eqref{eq:intro:HXdag1} and \eqref{eq:intro:HXdag3}, is in general challenging, as sampling from that distribution is our goal in the first place. Yet experiments~\cite{CuiTon22} show that $H_{\bar X\mid y^\dagger}$ achieves lower posterior error after reduction than \eqref{eq:intro:HXdag2}. 

\subsubsection{Output space.}
In the output space,~\cite{BapMarZah22_pre} proposes the data-independent diagnostic matrix
\begin{equation} \label{eq:intro:HY}
    H_{\bar Y} \coloneqq \int \Gamma^{-1/2}\nabla G(x)\Gamma_0\nabla G(x)^\top \Gamma^{-1/2}\pi(x)\,\mathrm dx,
\end{equation}
such that $\bar V_s$'s columns are the $s$ leading eigenvectors of $H_{\bar Y}$. In the output space, only a data-free method is known in the LIS literature. We will see later that there is no straightforward equivalent to the DB-LIS diagnostic matrix~\eqref{eq:intro:HXdag3} for output space reduction.

\section{Partially data-informed subspaces from tempered distributions} \label{sec:tempered}
This section contains our main theoretical contribution: $\alpha$-LIS, \emph{partially data-informed} subspaces based on samples from tempered distributions.

\subsection{Assumptions}
The core assumption in this section will be that certain \emph{subspace logarithmic Sobolev inequalities} hold. It is standard in the LIS literature to make such \cite{BapMarZah22_pre,zahm2022certified} or related \cite{CuiTon22} assumptions. We first define these inequalities and then state our assumptions.

\begin{definition}[Logarithmic Sobolev inequality]
    A random variable $Z$ satisfies the logarithmic Sobolev inequality if there exists a constant $C$ such that for all smooth, nonnegative functions $h$ with $\int h(z)\pi_Z(z)\,\mathrm dz = 1$,
    \begin{equation}
        \int h(z)\log h(z)\pi_Z(z)\,\mathrm dz \le \frac C2 \int h(z)\norm{\nabla_z \log h(z)}^2\,\pi_Z(z)\,\mathrm dz.
    \end{equation}
\end{definition}
\begin{definition}[Subspace logarithmic Sobolev inequality]
    A random variable $Z$ satisfies the subspace logarithmic Sobolev inequality if there exists a constant $\overline C$ such that, for any unitary matrix $[W_t, W_\perp]$ and any $z_\perp$, each of compatible dimensions, the conditional random variable $Z_t\mid Z_\perp = z_\perp$ (with $Z_t \coloneqq W_t^\top Z$ and $Z_\perp \coloneqq W_\perp^\top Z$) satisfies the logarithmic Sobolev inequality with $C=\overline C$.
\end{definition}
\begin{corollary} \label{cor:joint}
    Let a random variable $(X, Y)$ satisfy the subspace logarithmic Sobolev inequality with constant~$\overline C$. Then, for any unitary matrix $[V_s, V_\perp]$ and any $y_s$, each of compatible dimensions, the conditional random variable~${X \mid Y_s = y_s}$ (with $Y_s \coloneqq V_s^\top Y$) satisfies the logarithmic Sobolev inequality with $C=\overline C$.
\end{corollary}
\begin{proof}
    Define $Y_\perp \coloneqq V_\perp^\top Y$. With $\int h(x)\pi(x\mid y_s)\,\mathrm dx=1$, we verify that
    \begin{align*}
        \int h(x)\log h(x)\pi(x\mid y_s)\,\mathrm dx &= \iint h(x)\log h(x)\pi(x,y_\perp\mid y_s)\,\mathrm dx\,\mathrm dy_\perp\\
        &\le \int\frac{\overline C}2\int h(x)\norm{\nabla_{(x,y_\perp)}\log h(x)}^2\pi(x, y_\perp\mid y_s)\,\mathrm dx\,\mathrm dy_\perp\\
        &=\frac{\overline C}2 \int h(x)\norm{\nabla_x\log h(x)}^2\pi(x\mid y_s)\,\mathrm dx,
    \end{align*}
    which is exactly what we need to prove. The inequality above follows by the subspace logarithmic Sobolev inequality satisfied by $(X, Y)$, after verifying that
    \begin{equation}
        \iint h(x)\pi(x,y_\perp\mid y_s)\,\mathrm dx\,\mathrm dy_\perp = \int h(x)\pi(x\mid y_s)\,\mathrm dx = 1. \qedhere
    \end{equation}
\end{proof}
\begin{assumption} \label{ass:prior}
    The prior distribution $\pi_X$ satisfies the subspace logarithmic Sobolev inequality.
\end{assumption}
\begin{assumption} \label{ass:joint}
    The joint distribution $\pi_{X,Y}$ satisfies the subspace logarithmic Sobolev inequality.
\end{assumption}
\Cref{ass:prior} is used for deriving the reduced input space; \cref{ass:joint} is needed to find the reduced output space. We note that \cref{ass:joint} is also made in~\cite{BapMarZah22_pre} and implies \cref{ass:prior}. Both assumptions are satisfied for Gaussian prior~$\pi_X$ and linear forward model $G$. For more examples of distributions that satisfy these assumptions, we refer to~\cite{BapMarZah22_pre,zahm2022certified}.

\subsection{Definition of the new generalized subspaces}

Denote by
\begin{equation}
    \KL\mu\nu \coloneqq \int \mu(x)\log\frac{\mu(x)}{\nu(x)}\,\mathrm dx
\end{equation}
the Kullback--Leibler divergence of two distributions. The dimension reduction method of using the diagnostic matrix $H_{\bar X\mid y^\dagger}$ from \eqref{eq:intro:HXdag3} derives from upper-bounding $\KL{\pi(\cdot\mid y^\dagger)}{\pi^*(\cdot\mid y^\dagger)}$ and then finding a $\bar U_r$ that minimizes that upper bound. The data-independent $H_{\bar X}$ and $H_{\bar Y}$, on the other hand, minimize an upper bound on $\mathbb E_{y\sim\pi_Y}\left[\KL{\pi(\cdot\mid y)}{\pi^*(\cdot\mid y)}\right]$. In other words, the data-dependent approach tries to minimize the KL divergence of the true and reduced posterior distributions for the given observation, while the data-independent approach tries to minimize the \emph{average} KL divergence over all likely observations.

We will interpolate between these two objectives. Specifically, we will aim to minimize an upper bound on
\begin{equation} \label{eq:diag:E}
    \mathbb E_{y\sim p_\alpha}\left[\KL{\pi(\cdot\mid y)}{\pi^*(\cdot\mid y)} \right], 
\end{equation}
where we define a family of distributions $p_\alpha$, parameterized by a scalar $\alpha\in[0, 1]$:
\begin{equation}
    p_\alpha(y) \propto \pi(y)f_\mathrm{gauss}\left(y\, ; y^\dagger, \frac{1-\alpha}\alpha\Gamma\right).
\end{equation}
Here $f_\mathrm{gauss}(\,\cdot\,; m, \Sigma)$ denotes the PDF of a Gaussian distribution with mean $m$ and covariance $\Sigma$. The boundaries are defined by the natural extensions of the limit of $\alpha$:
\begin{equation}
\begin{cases}
    p_0(y)\coloneqq\lim_{\alpha\to 0}p_\alpha(y) = \pi(y),\\
    p_1(y)\coloneqq\lim_{\alpha\to 1}p_\alpha(y) = \delta_{y^\dagger}(y).
\end{cases}
\end{equation}
As $\alpha$ increases from $0$ towards $1$, the objective \eqref{eq:diag:E} evolves from an averaged KL divergence over all likely observations towards the KL divergence for one specific observation.

To find good reduced spaces, we will first upper-bound \eqref{eq:diag:E} and then find the spaces that minimize that upper bound. The full derivation is deferred to \cref{sec:tempered-deriv}; here, we only define the resulting spaces.
\begin{definition}[$\alpha$-LIS]
    For a given value of $\alpha\in[0, 1]$, the $\alpha$-LIS method computes the dimension-$r$ reduced input subspace as the $r$ leading eigenvectors of the diagnostic matrix
    \begin{equation} \label{eq:tempered:HX}
        H_{X\mid y^\dagger}^\alpha \coloneqq \int\nabla G(x)^\top \Gamma^{-1}\left((1-\alpha)\Gamma + \alpha^2(y^\dagger-G(x))(y^\dagger-G(x))^\top \right)\Gamma^{-1}\nabla G(x)\pi^{(\alpha)}(x\mid y^\dagger)\,\mathrm dx.
    \end{equation}

    The dimension-$s$ reduced output space is not computed through a diagnostic matrix. Instead, it is defined as the minimizer for $V_s\in\mathbb R^s$ with $V_s^\top V_s=I$ of 
    \begin{equation} \label{eq:tempered:J}
        J(V_s) \coloneqq \Tr\left[\int\nabla G(x)^\top P_s\left((1-\alpha)\Gamma + \alpha^2(y^\dagger-G(x))(y^\dagger-G(x))^\top \right)P_s\nabla G(x)\pi^{(\alpha)}(x\mid y^\dagger)\,\mathrm dx\right]
    \end{equation}
    with $P_s = \Gamma^{-1} - V_s(V_s^\top \Gamma V_s)^{-1}V_s^\top$.
\end{definition}

\begin{remark}[$\alpha$-LIS as generalization of DF-LIS and DB-LIS] \label{rem:tempered:out-α=0-diagnostic}
    In the case $\alpha=0$, the $\alpha$-LIS input diagnostic matrix matches the diagnostic matrix from~\cite{BapMarZah22_pre} when both are applied after whitening (see \cref{rem:intro:white}): $H_{\bar X\mid y^\dagger}^0 = H_{\bar X}$. For $\alpha=1$, we have $H_{\bar X\mid y^\dagger}^1 = H_{\bar X\mid y^\dagger}$. For $0<\alpha<1$, $H_{\bar X\mid y^\dagger}^\alpha$ interpolates non-linearly between these two diagnostic matrices.

    For the output space, consider the case where $\alpha=0$ and whitening has been applied according to \cref{rem:intro:white}, such that $\Gamma=I$. First note that $P_s = V_\perp(V_\perp^\top\Gamma V_\perp)^{-1}V_\perp^\top$ and so, this becomes $P_s=V_\perp V_\perp^\top$, such that $P_sP_s=P_s$. Applying this and the cyclic invariance of the trace (i.e., $\Tr[AB]=\Tr[BA]$ for all $A$ and $B$), $J(V_s)$ becomes
    \begin{equation}
        \frac{\bar C}2\Tr\left[\int\nabla G(x)^\top P_sP_s\nabla G(x)\pi(x)\,\mathrm dx\right] = \frac{\bar C}2\Tr\left[V_\perp^\top\left\{\int\nabla G(x)\nabla G(x)^\top \pi(x)\,\mathrm dx\right\}V_\perp\right];
    \end{equation}
    thus we recover the diagnostic matrix $H_{\bar Y}$, given in~\eqref{eq:intro:HY}, as in~\cite{BapMarZah22_pre}. 
\end{remark}

In the case for $\alpha>0$, there is no such formulation where the optimal subspace is the leading eigenspace of some diagnostic matrix. We will need to find or approximate the minimizer of $J$ in some other way, which is studied in \cref{sec:opt}.

\subsection{Accumulated posterior information} \label{sec:tempered:acc}
When using tempered algorithms, one will not have access to \emph{only} samples $\alpha_k$, when $k>0$. Instead, samples are generated from the entire sequence $0=\alpha_0< \dots<\alpha_k \leq 1$ representing accumulated information over $[\alpha_0,\alpha_k]$. In the infinite-sample limit, we are motivated to only use the samples of the final distribution $\alpha_k$, but with a finite number of samples we propose that it is likely helpful to utilize all available samples.  A simple approach is to minimize an upper bound on
\begin{equation}
    \int_{\alpha_0}^{\alpha_k}\mathbb E_{y\sim p_\alpha}\left[\KL{\pi(\cdot\mid y)}{\pi^*(\cdot\mid y)} \right]\,\mathrm d\alpha
\end{equation}
instead of \eqref{eq:diag:E}. This results in an accumulated diagnostic matrix
\begin{equation}
    H_{X\mid y^\dagger}^{[\alpha_0,\alpha_k]} = \int_{\alpha_0}^{\alpha_k} H_{X\mid y^\dagger}^\alpha \,\mathrm{d}\alpha
\end{equation}
and an objective function for the output space that equals $J$ integrated over $[\alpha_0, \alpha_k]$. In practice, we will approximate these integrals with quadrature based on the $\{\alpha_i\}_{i=0}^k$ available. Note that even for linear and Gaussian problems, $H_{X\mid y^\dagger}^{[0,1]} \neq H_{X\mid y^\dagger}^{0.5}$ in general.

There are possibly several benefits of working with accumulated information. As mentioned, it can use all the available sample information in a situation where the number of samples is heavily limited. Also, accumulation widens the support of the information, and when using an emulator for sampling, it has been shown that wider-than-posterior distributions may be preferred as training data~\cite[Theorem 3]{helin2023introductiongaussianprocessregression}. Such behavior might benefit situations using the posterior in other ways than to compute integrals. For example, extremal regions could be explored for classifying the tails and may be lost under only a near-posterior analysis.

\subsection{Statistical linearization} \label{sec:tempered:gradfree}
In contrast to PCA- and CCA-based dimension reduction, (partially) data-informed techniques require gradients $\nabla G$ of the forward map. In many practical settings one either does not have access to these derivatives, or they are too corrupted by noise to be useful.

To address this challenge, we use the available set of model evaluations $\{ G(x_j)\}_{j=1}^J$, sampled from the $\alpha$-tempered posterior. Consider the empirical covariances
\begin{align*}
    C^{xx} &= \frac{1}{J}\sum_{j=1}^J \left[(x_j - \overline{x})\otimes(x_j - \overline{x}) \right],\qquad &\overline{x} &= \frac{1}{J}\sum_{j=1}^J x_j,\\
    C^{xG} &= \frac{1}{J}\sum_{j=1}^J \Big[ (x_j - \overline{x})\otimes(G(x_j) - \overline{G(x)}) \Big], \qquad&\overline{G(x)} &= \frac{1}{J} \sum_{j=1}^J G(x_j) ,
\end{align*}
similarly to those used in \cref{ex:eki}. Gradient approximations can then be obtained by the \emph{statistical linearization} (SL) approach \cite{ChaCheSan20,ungaralaIteratedFormsKalman2012,CalReiStu25}:
\begin{equation}
\nabla G(\cdot)\approx \bigl[(C^{xx})^\dagger C^{xG}\bigr]^\top.
\end{equation}
The SL form provides a single constant derivative (the linear approximation of $\nabla G(x)$ over the set of evaluations). This also makes $H_{X\mid y^\dagger}^\alpha$ and $J$ cheaper to compute.

Other derivative-free dimension reduction methods exist: for example, local finite-difference approximation~\cite{ConGlei14}, denoising score functions~\cite{BapBreMar24_pre}, and related approaches such as covariance-informed subspaces~\cite{PolMaiSocGes26_pre} are natively derivative-free. However, our extensions offer simple nonlocal derivative approximations that are not yet explored by this literature.

\section{Results} \label{sec:results}
Here we perform numerical experiments with the dimension reduction techniques discussed in this paper. We first give an overview of the test problems we consider in this section. In all our experiments, we solve the optimization problem for the reduced output space by the incremental strategy of \cref{sec:opt:incr}, as motivated by the experiments of \cref{sec:opt:comp}.

\subsection{Overview of the models}
\paragraph{Ablations in a linear problem: testing the theory.}
We start by considering problems where $G(x) = Ax$, with $A \in \mathbb R^{d_y\times d_x}$ and~$d_x\ge d_y$. Specifically, we choose
\begin{equation}
    A = U\Lambda V^\top ,
\end{equation}
where $U\in\mathbb R^{d_y\times d_y}$ and $V\in\mathbb R^{d_x\times d_y}$ are random orthogonal matrices and $\Lambda = 100\cdot\operatorname{diag}\left(\{1^{-1}, 2^{-1}, \ldots, d_y^{-1}\}\right)$. The prior distribution is $\pi_X = \mathcal N(0, \Gamma_0)$ with $\Gamma_0 = \gamma_0\cdot\operatorname{diag}\left(\{1^{-2}, 2^{-2}, \ldots, d_x^{-2}\}\right)$ and $\gamma_0=4$, and the noise covariance equals $\Gamma=I$. Unless mentioned otherwise, we use an observation
\begin{equation} \label{eq:res:linear:ydagger}
    y^\dagger = G(x^\dagger) + \eta^\dagger,
\end{equation}
where $x^\dagger \sim \mathcal N(0, \Gamma_0)$ and $\eta^\dagger \sim \mathcal N(0, \Gamma)$ are randomly sampled.

The main advantage of studying linear models first is that the resulting posteriors, even after dimension reduction, are available in closed form. For the full posterior, this is straightforward:
\begin{equation}
    \pi_{X\mid Y}(\cdot \mid y^\dagger) = \mathcal N(\mu, \Sigma) \qquad \text{with} \qquad \Sigma = (\Gamma_0^{-1} + A^\top \Gamma^{-1}A)^{-1} \qquad \text{and} \qquad \mu = \Sigma A^\top  \Gamma^{-1} y^\dagger.
\end{equation}
For the reduced posterior, consider $(U_r, V_s)$ computed by some dimension reduction method (based on samples of $\pi(x)$ or $\pi(x \mid y^\dagger)$, for instance). We can first compute the likelihood \eqref{eq:bip:pi*} as a convolution of two multivariate Gaussians:
\begin{equation}
\begin{aligned}
    \pi^*_{Y\mid X}(y\mid x) &\propto \int \exp\left(-\frac12\norm{y_s - V_s^\top A(U_rx_r + U_\perp x_\perp)}_{V_s^\top\Gamma V_s}^2\right)\exp\left(-\frac12\norm{x_r - \mu_\perp}_S^2\right)\,\mathrm dx_\perp\\
    &\propto \int f_\mathrm{gauss}\left(y_s - V_s^\top AU_rx_r; V_s^\top AU_\perp x_\perp, V_s^\top\Gamma V_s\right) f_\mathrm{gauss}\left(x_\perp; \mu_\perp, S\right)\,\mathrm dx_\perp\\
    &\propto f_\mathrm{gauss}(y_s; Bx_r, C)
\end{aligned}
\end{equation}
with
\begin{equation}
    B = V_s^\top A[U_r + U_\perp (U_\perp^\top\Gamma_0U_r)(U_r^\top\Gamma_0U_r)^{-1}] \qquad \text{and} \qquad C = V_s^\top\Gamma V_s + (V_s^\top A U_\perp)S(U_\perp^\top A^\top V_s),
\end{equation}
where we used that $\bigl[\begin{smallmatrix}
    X_r\\X_\perp
\end{smallmatrix}\bigr]\sim\mathcal N\bigl(0, \bigl[\begin{smallmatrix}
    U_r^\top \Gamma_0U_r & U_r^\top \Gamma_0U_\perp\\
    U_\perp^\top \Gamma_0U_r & U_\perp^\top \Gamma_0U_\perp
\end{smallmatrix}\bigr]\bigr)$ and, thus, $X_\perp \mid X_r = x_r \sim \mathcal N(\mu_\perp, S)$ with
\begin{equation}
    \mu_\perp = (U_\perp^\top \Gamma_0U_r)(U_r^\top \Gamma_0U_r)^{-1}x_r \qquad \text{and} \qquad S = U_\perp^\top \Gamma_0U_\perp - (U_\perp^\top \Gamma_0U_r)(U_r^\top \Gamma_0U_r)^{-1}(U_r^\top \Gamma_0U_\perp).
\end{equation}
Combining this with the prior, we find
\begin{equation}
    \pi^*_{X\mid Y}(\cdot\mid y^\dagger) = \mathcal N(\mu^*, \Sigma^*)
\end{equation}
with
\begin{equation}
    \Sigma^* = (\Gamma_0^{-1} + U_rB^\top C^{-1}BU_r^\top)^{-1} \qquad \text{and} \qquad \mu^* = \Sigma^*(U_rB^\top C^{-1}V_s^\top y).
\end{equation}
We can then use a closed-form expression for the Wasserstein-2 distance between $\pi$ and $\pi^*$:
\begin{equation}
    W_2(\pi, \pi^*) = \norm{\mu-\mu^*}^2 + \Tr\Bigl[\Sigma+\Sigma^* - 2\sqrt{\sqrt\Sigma\,\Sigma^*\,\sqrt\Sigma}\,\Bigr].
\end{equation}

\paragraph{Linear-exponential model.}
We also consider a nonlinear problem: the linear-exponential set-up from~\cite{CuiTon22}. We have
\begin{equation}
    G(x) = A \exp(x),
\end{equation}
with the exponential applied elementwise and where $A$, the prior distribution $\Gamma_0 = \gamma_0\cdot\operatorname{diag}\left(\{1^{-2}, 2^{-2}, \ldots, d_x^{-2}\}\right)$ with (by default) $\gamma_0=4$, and the noise distribution are formed in the same way as in the linear problem. The observation, similarly, is $y^\dagger = G(x^\dagger) + \eta^\dagger$ with $x^\dagger\sim\mathcal N(0,\Gamma_0)$ and $\eta^\dagger\sim\mathcal N(0,\Gamma)$.

For this problem, we have no closed-form formula for the Wasserstein-2 distance. In addition, it is very hard to approximate in high dimensions. Instead, we will consider the squared Hellinger distance
\begin{equation}
    \operatorname{hell}^2(\pi, \pi^*) \coloneqq 1 - \int\sqrt{\pi(x)\pi^*(x)}\,\mathrm dx,
\end{equation}
which measures distribution overlap and therefore heavily punishes even small deviations between distributions. However, it is more tractable to approximate. Noting that we only have access to unnormalized distributions $\widetilde\pi = Z\pi$ and $\widetilde\pi^* = Z^*\pi^*$, we draw samples $\{x^{*,i}\}_{i=1}^n$ from $\pi^*$ with Markov chain Monte Carlo and approximate, using \emph{self-normalized importance sampling},
\begin{equation}
\begin{aligned}
    \int\sqrt{\pi(x)\pi^*(x)}\,\mathrm dx &= \frac1{\sqrt{ZZ^*}}\int\sqrt{\widetilde\pi(x)\widetilde\pi^*(x)}\,\mathrm dx = \sqrt{\frac{Z^*}{Z}}\int\sqrt{\frac{\widetilde\pi(x)}{\widetilde\pi^*(x)}}\,\pi^*(x)\,\mathrm dx\\
    &= \frac{\int\sqrt{\frac{\widetilde\pi(x)}{\widetilde\pi^*(x)}}\,\pi^*(x)\,\mathrm dx}{\sqrt{\int\frac{\widetilde\pi(x)}{\widetilde\pi^*(x)}\,\pi^*(x)\,\mathrm dx}} \approx \frac{\frac1n\sum_{i=1}^n\sqrt{\frac{\widetilde\pi(x^{*,i})}{\widetilde\pi^*(x^{*,i})}}}{\sqrt{\frac1n\sum_{i=1}^n\frac{\widetilde\pi(x^{*,i})}{\widetilde\pi^*(x^{*,i})}}}.
\end{aligned}
\end{equation}

\paragraph{Darcy flow with spatially varying permeability.}
Flow through a porous medium is a standard test case among Bayesian inverse problems. We consider the spatial domain $\Omega=[0,1]\times[0,1]$ and the PDE
\begin{subequations} \label{eq:res:darcy:pde}
\begin{align}
    -\nabla\cdot\bigl(a(\xi,u)\nabla p(\xi)\bigr) &= c, &&\xi\in\Omega,\\
    p(\xi) &= 0, &&\xi\in\partial\Omega,
\end{align}
\end{subequations}
which models the pressure field $p$ for some constant fluid source $c$. This field depends on the permeability field~$a$, which we will parameterize by an unknown vector $u\in \mathbb R^{d_x}$, which will correspond to the vector $x$ to be learned in the inverse problem.

Based on $d_y$ measurements of $p$, equispaced throughout $\Omega$, we want to find out~$a$. We model
\begin{equation} \label{eq:res:darcy:a}
    a(\xi, u) = \exp\left(\sum\nolimits_{k\in K} u_k\sqrt{\lambda_k}\,\phi_k(\xi)\right),
\end{equation}
where $\lambda_k = (\pi^2\norm k_2^2+\tau^2)^{-s}$ with $(\tau,s)=(3,2)$, $K$ contains the $d_x$ elements of $\mathbb N^2\backslash\{(0,0)\}$ with largest $\lambda_k$, and $\phi_k(\xi) = c_k\cos(\pi k_1\xi_1)\cos(\pi k_2\xi_2)$ with $c_k=\sqrt2$ if either $k_1=0$ or $k_2=0$, and $c_k=2$ otherwise. The prior is $u\sim\mathcal N(0,I)$. This corresponds to modeling $a$ as a log-normal random field and taking a truncated Karhunen--Lo\`eve expansion (see, e.g.,~\cite{huangIteratedKalmanMethodology2022c}). Combined, we solve a Bayesian inverse problem where $G$ first maps $u$ to $a$ through \eqref{eq:res:darcy:a}, then solves \eqref{eq:res:darcy:pde} with central finite differences on a grid with stepsize $h=2^{-5}$, and finally measures $p$. We model the observational noise covariance as $\Gamma=10^{-4}I$. An observation $y^\dagger$ is constructed by sampling $u$ from the prior, solving the PDE with finer stepsize $2^{-7}$, and adding a sample from the noise distribution.

Similarly to the linear-exponential model above, we measure the difference between true and reduced posteriors by estimating the squared Hellinger distance.

\subsection{Reference setting}
The first set of experiments uses our algorithms as detailed in this paper to study their theoretical performance. In the next subsection, we will approximate various aspects of the dimension reduction methods to make them more suitable for realistic applications.

Integrals over a tempered or posterior distribution are approximated with Markov chain Monte Carlo methods performed in the full space.

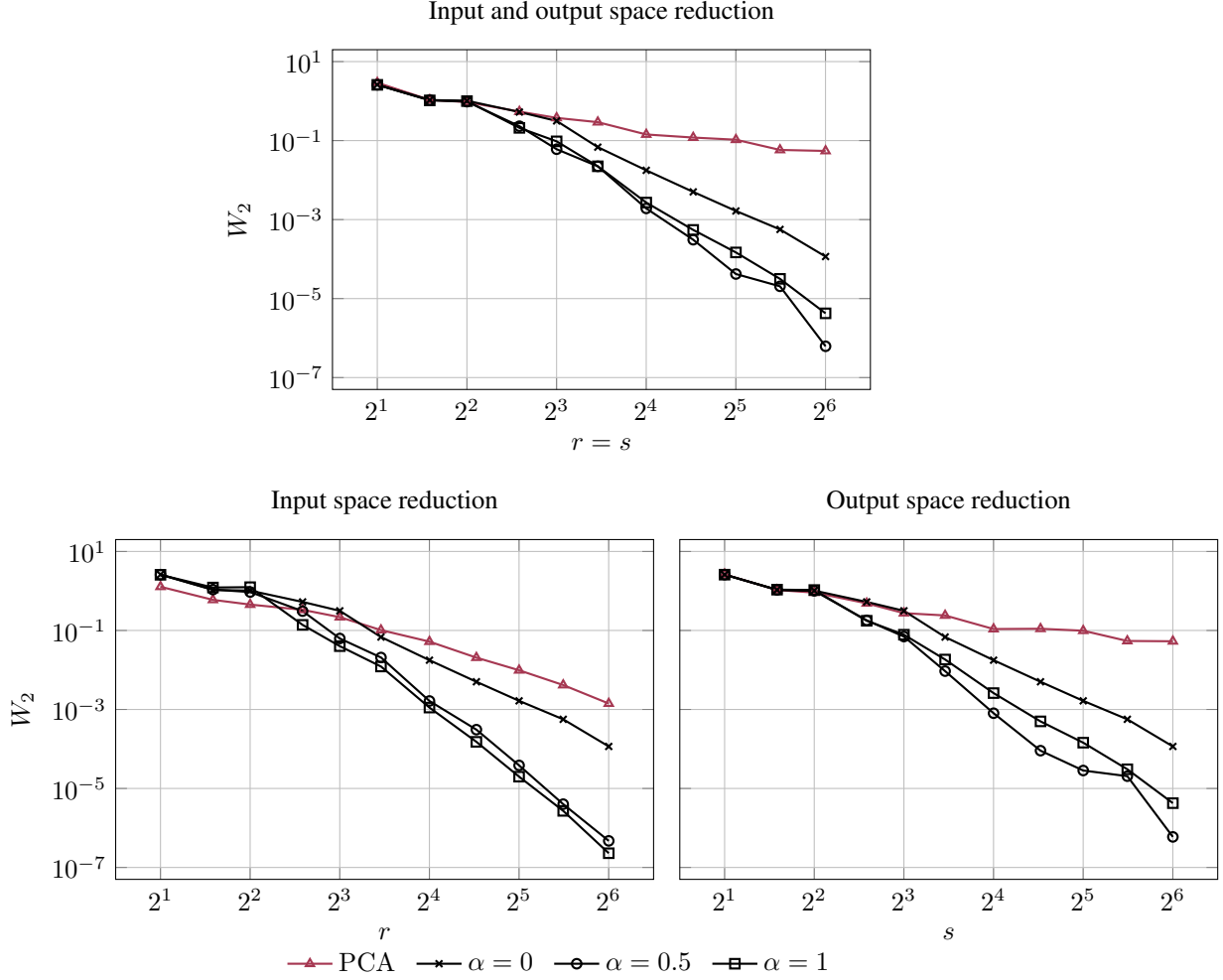
\begin{figure}
    \centering
        \begin{tikzpicture}
        \begin{groupplot}[group style={{group size=2 by 2, horizontal sep=2cm, vertical sep=2cm}},name=t2, height=0.37\textwidth, width=0.53\textwidth,legend style={transpose legend,legend columns=0,draw=none}]
            \nextgroupplot[xshift=0.325\textwidth, xlabel={$r=s$}, ylabel={$W_2$}, legend to name=grouplegend1, xmode=log, log basis x=2, ymode=log, grid=major, ymin=0.5e-7, ymax=20, title={Input and output space reduction}, every axis plot/.append style={thick, mark options={solid, scale=0.9}}]
            \addplot[winered, mark=triangle] table [x=dim, y=err_PCA_PCA, col sep=comma] {data/linear-simple/W2-100-100-true-true.csv};
            \addlegendentry{$\mathrm{PCA}\;\;$}
            \addplot[black, mark=x] table [x=dim, y={err_α=0.0_α=0.0}, col sep=comma] {data/linear-simple/W2-100-100-true-true.csv};
            \addlegendentry{$\alpha=0\;\;$}
            \addplot[black, mark=o] table [x=dim, y={err_α=0.5_α=0.5 (incr.)}, col sep=comma] {data/linear-simple/W2-100-100-true-true.csv};
            \addlegendentry{$\alpha=0.5\;\;$}
            \addplot[black, mark=square] table [x=dim, y={err_α=1.0_α=1.0 (incr.)}, col sep=comma] {data/linear-simple/W2-100-100-true-true.csv};
            \addlegendentry{$\alpha=1\;\;$}

            \nextgroupplot[group/empty plot]
    
            \nextgroupplot[xshift=-0cm, xshift=-0.175\textwidth, xlabel={$r$}, ylabel={$W_2$}, legend to name=grouplegend1, xmode=log, log basis x=2, ymode=log, grid=major, ymin=0.5e-7, ymax=20, title={Input space reduction}, every axis plot/.append style={thick, mark options={solid, scale=0.9}}]
            \addplot[winered, mark=triangle] table [x=dim, y=err_PCA_PCA, col sep=comma] {data/linear-simple/W2-100-100-true-false.csv};
            \addlegendentry{$\mathrm{PCA}\;\;$}
            \addplot[black, mark=x] table [x=dim, y={err_α=0.0_α=0.0}, col sep=comma] {data/linear-simple/W2-100-100-true-false.csv};
            \addlegendentry{$\alpha=0\;\;$}
            \addplot[black, mark=o] table [x=dim, y={err_α=0.5_α=0.5 (incr.)}, col sep=comma] {data/linear-simple/W2-100-100-true-false.csv};
            \addlegendentry{$\alpha=0.5\;\;$}
            \addplot[black, mark=square] table [x=dim, y={err_α=1.0_α=1.0 (incr.)}, col sep=comma] {data/linear-simple/W2-100-100-true-false.csv};
            \addlegendentry{$\alpha=1\;\;$}
        
            \nextgroupplot[xshift=-0cm, xshift=-0.275\textwidth, xlabel={$s$}, legend to name=grouplegend1, xmode=log, log basis x=2, ymode=log, grid=major, legend to name=grouplegend2, ymin=0.5e-7, ymax=20, title={Output space reduction}, every axis plot/.append style={thick, mark options={solid, scale=0.9}}, yticklabels={}]
            \addplot[winered, mark=triangle] table [x=dim, y=err_PCA_PCA, col sep=comma] {data/linear-simple/W2-100-100-false-true.csv};
            \addlegendentry{$\mathrm{PCA}\;\;$}
            \addplot[black, mark=x] table [x=dim, y={err_α=0.0_α=0.0}, col sep=comma] {data/linear-simple/W2-100-100-false-true.csv};
            \addlegendentry{$\alpha=0\;\;$}
            \addplot[black, mark=o] table [x=dim, y={err_α=0.5_α=0.5 (incr.)}, col sep=comma] {data/linear-simple/W2-100-100-false-true.csv};
            \addlegendentry{$\alpha=0.5\;\;$}
            \addplot[black, mark=square] table [x=dim, y={err_α=1.0_α=1.0 (incr.)}, col sep=comma] {data/linear-simple/W2-100-100-false-true.csv};
            \addlegendentry{$\alpha=1\;\;$}
        \end{groupplot}
    
        \node at (8.5,0)[below, yshift=-20\pgfkeysvalueof{/pgfplots/every axis title shift}]{\ref*{grouplegend2}};
    \end{tikzpicture}
    \caption{Wasserstein-2 posterior errors for different algorithms on the linear test problem}
    \label{fig:res:linear-simple}
\end{figure}

\paragraph{General comparison: linear problem.} We first perform a general comparison of various dimension reduction methods: we compare PCA to likelihood-informed reduction with $\alpha\in\{0,0.5,1\}$ on the linear problem with~${d_x=d_y=100}$.

The result is shown in \cref{fig:res:linear-simple}. PCA performs worst, which can be expected since the spectrum of $\Lambda$ is relatively spread out. Data-free dimension reduction with $\alpha=0$ performs somewhat better, but $\alpha=0.5$ and $\alpha=1$ heavily outperform the other methods. Interestingly, $\alpha=0.5$ is roughly as good as $\alpha=1$. This is valuable in cases where tempered samples with $\alpha=0.5$ are cheaper to obtain than full posterior samples, or where the tempered samples are already available (such as in \cref{sec:ces} later).

\begin{figure}
    \centering
        \begin{tikzpicture}
        \begin{groupplot}[group style={{group size=2 by 1, horizontal sep=2cm, vertical sep=2cm}},name=t2,, height=0.37\textwidth, width=0.53\textwidth,legend style={transpose legend,legend columns=0,draw=none}]
            \nextgroupplot[xlabel={$\alpha$}, ylabel={$W_2$}, legend to name=grouplegend1, ymode=log, grid=major, legend to name=grouplegend, unbounded coords=discard, every axis plot/.append style={thick, mark options={solid, scale=0.9}}, title={Input space reduction ($r=d$, $s=d_y$)}, ymin = 4e-7, ymax = 2e-1]
            \addplot[color=ferngreen, solid, mark=o] table [x=alpha, y=err_incr_500_in, col sep=comma]{data/linear-alpha/W2-10_summary_50.csv};
            \addlegendentry{$d=10\;\;$}
            \addplot[color=winered, solid, mark=o] table [x=alpha, y=err_incr_500_in, col sep=comma]{data/linear-alpha/W2-20_summary_50.csv};
            \addlegendentry{$d=20\;\;$}
            \addplot[color=orientblue, solid, mark=o] table [x=alpha, y=err_incr_500_in, col sep=comma]{data/linear-alpha/W2-40_summary_50.csv};
            \addlegendentry{$d=40$}

            \nextgroupplot[xshift=-.1\textwidth, xlabel={$\alpha$}, legend to name=grouplegend1, ymode=log, grid=major, legend to name=grouplegend, unbounded coords=discard, every axis plot/.append style={thick, mark options={solid, scale=0.9}}, title={Output space reduction ($r=d_x$, $s=d$)}, ymin = 4e-7, ymax = 2e-1, yticklabels={}]
            \addplot[color=ferngreen, solid, mark=o] table [x=alpha, y=err_incr_500_out, col sep=comma]{data/linear-alpha/W2-10_summary_50.csv};
            \addlegendentry{$d=10\;\;$}
            \addplot[color=winered, solid, mark=o] table [x=alpha, y=err_incr_500_out, col sep=comma]{data/linear-alpha/W2-20_summary_50.csv};
            \addlegendentry{$d=20\;\;$}
            \addplot[color=orientblue, solid, mark=o] table [x=alpha, y=err_incr_500_out, col sep=comma]{data/linear-alpha/W2-40_summary_50.csv};
            \addlegendentry{$d=40$}
        \end{groupplot}
    
        \node at (6,0)[below, yshift=-2\pgfkeysvalueof{/pgfplots/every axis title shift}]{\ref*{grouplegend}};
    \end{tikzpicture}
    \caption{Wasserstein-2 posterior errors for the linear test problem in function of $\alpha$}
    \vspace{1cm}
    \label{fig:res:linear-alpha}
\end{figure}
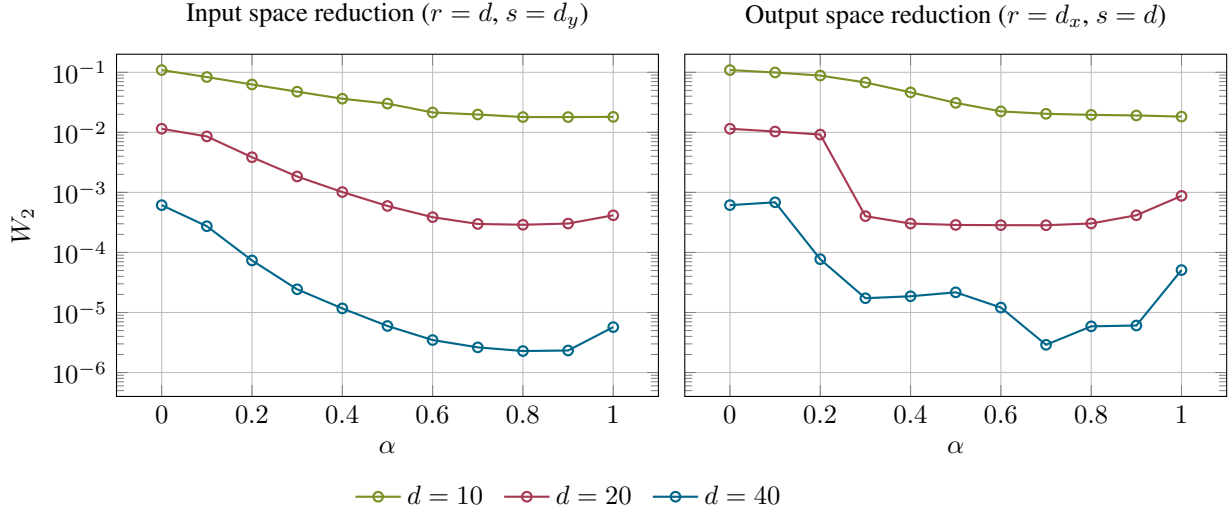

\paragraph{Study of the effect of $\alpha$.}
\Cref{fig:res:linear-alpha} expands on the earlier observation that $\alpha=0.5$ can rival $\alpha=1$. Keeping the parameters of the previous experiments, for each of $r=s\in\{10, 20, 40\}$, we plot the posterior error as a function of $\alpha$. To deal with fluctuations due to the manifold optimization, we plot the median error over 32 different realizations of the problem. While a larger $\alpha$ mostly gives better errors, choosing $\alpha$ (slightly) below~$1$ seems beneficial, perhaps due to increased robustness in the presence of Monte Carlo errors when approximating $H_{X\mid y^\dagger}^\alpha$ and $J$. This effect seems to grow more pronounced as the reduced dimension $r=s$ increases.

\begin{figure}
    \centering
    \begin{tikzpicture}
        \begin{groupplot}[group style={group size=1 by 1, horizontal sep=2cm, vertical sep=2cm}, height=0.37\textwidth, width=0.53\textwidth, legend style={transpose legend, legend columns=3, draw=none}]
            \nextgroupplot[name=mainplot, xmode=log, log basis x=2, ymode=log, grid=major, xlabel={$r=s$}, ylabel={$W_2$}, legend to name=grouplegend, unbounded coords=discard, every axis plot/.append style={thick, mark options={solid, scale=0.9}}]
            \addplot[color=black, mark=x] table [x=dim, y={err_α=0.0_α=0.0}, col sep=comma]{data/linear-ood/real-W2.csv};
            \addlegendentry{$\alpha=0\;\;\;\;$}
            \addplot[color=black, mark=o] table [x=dim, y={err_α=0.5_α=0.5 (incr.)}, col sep=comma]{data/linear-ood/real-W2.csv};
            \addlegendentry{$\alpha=0.5\;\;\;\;$}
            \addplot[color=black, mark=square] table [x=dim, y={err_α=1.0_α=1.0 (incr.)}, col sep=comma]{data/linear-ood/real-W2.csv};
            \addlegendentry{$\alpha=1.0\;\;\;\;$}
            \addplot[color=winered, mark=x] table [x=dim, y={err_α=0.0_α=0.0}, col sep=comma]{data/linear-ood/ood-W2.csv};
            \addlegendentry{$\alpha=0\mathrm{\;(OOD)}$}
            \addplot[color=winered, mark=o] table [x=dim, y={err_α=0.5_α=0.5 (incr.)}, col sep=comma]{data/linear-ood/ood-W2.csv};
            \addlegendentry{$\alpha=0.5\mathrm{\;(OOD)}$}
            \addplot[color=winered, mark=square] table [x=dim, y={err_α=1.0_α=1.0 (incr.)}, col sep=comma]{data/linear-ood/ood-W2.csv};
            \addlegendentry{$\alpha=1.0\mathrm{\;(OOD)}$}
        \end{groupplot}
    
        \node at (mainplot.south) [below, yshift=-3\pgfkeysvalueof{/pgfplots/every axis title shift}] {\ref*{grouplegend}};
    \end{tikzpicture}
    \caption{Wasserstein-2 posterior errors for the linear test problem in function of the reduced input and output dimensions $r=s$, with either an out-of-distribution (``OOD'') or a typical observation}
    \label{fig:res:linear-ood}
\end{figure}
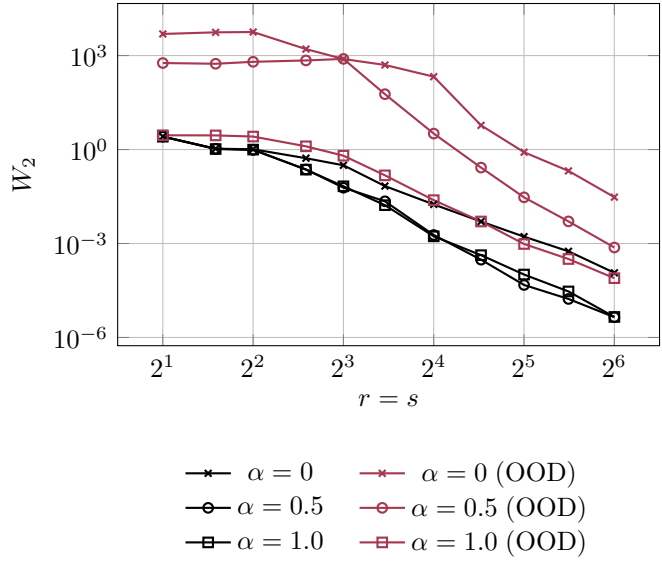

\paragraph{``Typical'' versus ``out-of-distribution'' observations.} Here we compare the observation \eqref{eq:res:linear:ydagger}, which tends to have a large probability density under the prior, to the observation
\begin{equation}
    y^\dagger = G\left(\begin{bmatrix}
        10\\\vdots\\10
    \end{bmatrix}\right) + 10\begin{bmatrix}
        \eta_1\\\vdots\\\eta_{d_y}
    \end{bmatrix},
\end{equation}
where $\eta_i\sim\mathcal N(0, 1)$, which is highly improbable under the prior.

Since $\alpha=0$ targets the expected KL divergence over the evidence $\pi_Y$, while $\alpha=1$ fully specializes on the observation $y^\dagger$, we expect larger $\alpha$s to outperform smaller ones more significantly on the out-of-distribution observation than on the typical one. \Cref{fig:res:linear-ood} confirms this, on the linear test problem using $d_x=d_y=100$. For the out-of-distribution observation, $\alpha=0$ and $\alpha=0.5$ both perform poorly.

\paragraph{General comparison: linear-exponential problem.} We repeat the earlier experiment comparing different dimension reduction methods for the linear-exponential problem with $d_x=d_y=100$.

\Cref{fig:res:linexp-simple} paints a similar picture to the linear test problem: the choice $\alpha\in\{0.5, 1\}$ performs significantly better than $\alpha=0$, which in turn is better than PCA. We note two differences: \begin{itemize}
    \item PCA appears much worse than before in the input space, achieving a Hellinger distance of around $1$ for any $r$ considered. This is primarily due to the error metric chosen: the Hellinger distance is much more punishing for relatively large posterior differences than the Wasserstein-2 distance we use for the linear problem.
    \item In the output space, $\alpha=0.5$ outperforms $\alpha=1$, instead of being comparable as in the linear problem.
\end{itemize}

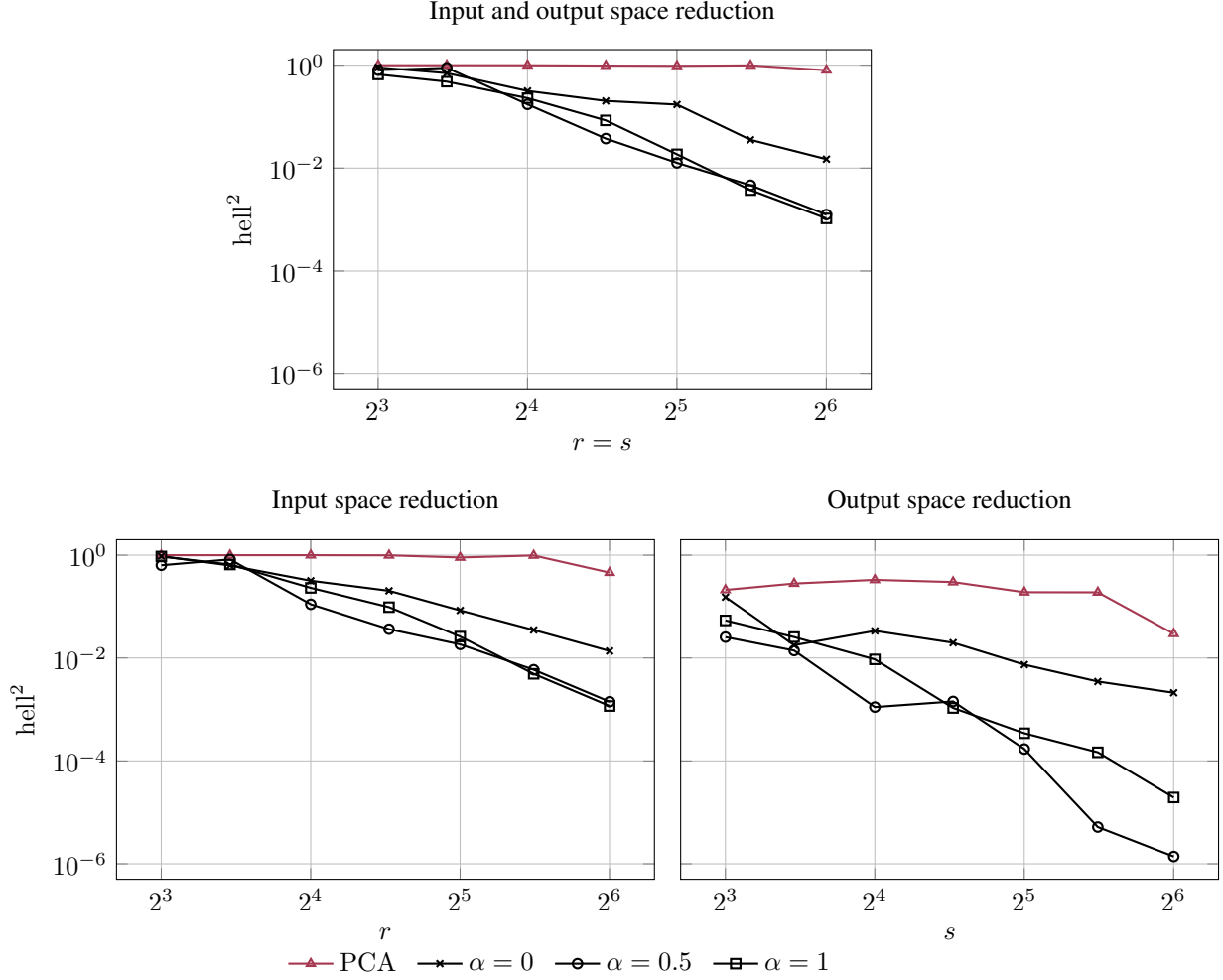
\begin{figure}
    \centering
        \begin{tikzpicture}
        \begin{groupplot}[group style={{group size=2 by 2, horizontal sep=2cm, vertical sep=2cm}},name=t2, height=0.37\textwidth, width=0.53\textwidth,legend style={transpose legend,legend columns=0,draw=none}]
            \nextgroupplot[xshift=0.325\textwidth, xlabel={$r=s$}, ylabel={$\operatorname{hell}^2$}, legend to name=grouplegend1, xmode=log, log basis x=2, ymode=log, grid=major, ymin=5e-7, ymax=2, title={Input and output space reduction}, every axis plot/.append style={thick, mark options={solid, scale=0.9}}]

            \addplot[winered, mark=triangle] table [x=dim, y=err_PCA_PCA, col sep=comma] {data/linexp-simple/errors-true-true.csv};
            \addlegendentry{$\mathrm{PCA}\;\;$}
            \addplot[black, mark=x] table [x=dim, y={err_α=0.0_α=0.0}, col sep=comma] {data/linexp-simple/errors-true-true.csv};
            \addlegendentry{$\alpha=0\;\;$}
            \addplot[black, mark=o] table [x=dim, y={err_α=0.5_α=0.5}, col sep=comma] {data/linexp-simple/errors-true-true.csv};
            \addlegendentry{$\alpha=0.5\;\;$}
            \addplot[black, mark=square] table [x=dim, y={err_α=1.0_α=1.0}, col sep=comma] {data/linexp-simple/errors-true-true.csv};
            \addlegendentry{$\alpha=1\;\;$}

            \nextgroupplot[group/empty plot]
    
            \nextgroupplot[xshift=-0cm, xshift=-0.175\textwidth, xlabel={$r$}, ylabel={$\operatorname{hell}^2$}, legend to name=grouplegend1, xmode=log, log basis x=2, ymode=log, grid=major, ymin=5e-7, ymax=2, title={Input space reduction}, every axis plot/.append style={thick, mark options={solid, scale=0.9}}]
            \addplot[winered, mark=triangle] table [x=dim, y=err_PCA_PCA, col sep=comma] {data/linexp-simple/errors-true-false.csv};
            \addlegendentry{$\mathrm{PCA}\;\;$}
            \addplot[black, mark=x] table [x=dim, y={err_α=0.0_α=0.0}, col sep=comma] {data/linexp-simple/errors-true-false.csv};
            \addlegendentry{$\alpha=0\;\;$}
            \addplot[black, mark=o] table [x=dim, y={err_α=0.5_α=0.5}, col sep=comma] {data/linexp-simple/errors-true-false.csv};
            \addlegendentry{$\alpha=0.5\;\;$}
            \addplot[black, mark=square] table [x=dim, y={err_α=1.0_α=1.0}, col sep=comma] {data/linexp-simple/errors-true-false.csv};
            \addlegendentry{$\alpha=1\;\;$}
        
            \nextgroupplot[xshift=-0cm, xshift=-0.275\textwidth, xlabel={$s$}, legend to name=grouplegend1, xmode=log, log basis x=2, ymode=log, grid=major, legend to name=grouplegend2, ymin=5e-7, ymax=2, title={Output space reduction}, every axis plot/.append style={thick, mark options={solid, scale=0.9}}, yticklabels={}]
            \addplot[winered, mark=triangle] table [x=dim, y=err_PCA_PCA, col sep=comma] {data/linexp-simple/errors-false-true.csv};
            \addlegendentry{$\mathrm{PCA}\;\;$}
            \addplot[black, mark=x] table [x=dim, y={err_α=0.0_α=0.0}, col sep=comma] {data/linexp-simple/errors-false-true.csv};
            \addlegendentry{$\alpha=0\;\;$}
            \addplot[black, mark=o] table [x=dim, y={err_α=0.5_α=0.5}, col sep=comma] {data/linexp-simple/errors-false-true.csv};
            \addlegendentry{$\alpha=0.5\;\;$}
            \addplot[black, mark=square] table [x=dim, y={err_α=1.0_α=1.0}, col sep=comma] {data/linexp-simple/errors-false-true.csv};
            \addlegendentry{$\alpha=1\;\;$}
        \end{groupplot}
    
        \node at (8.5,0)[below, yshift=-20\pgfkeysvalueof{/pgfplots/every axis title shift}]{\ref*{grouplegend2}};
    \end{tikzpicture}
    \caption{Hellinger posterior errors for different algorithms on the linear-exponential test problem}
    \label{fig:res:linexp-simple}
\end{figure}

\subsection{Practical setting}
We now deal with four practical problems that may arise when trying to apply our dimension reduction methods to real problems. We start by testing approximations of our algorithms suited for those situations one by one on the linear or linear-exponential test problems. We then combine all four approximations on a slightly more complex Darcy model. This prepares us for the next section, where likelihood-informed reduction will be intractable without such approximations.

\paragraph{The effect of noisy function evaluations.} If the model is stochastic, then to build $(U_r, V_s)$ we might only be able to evaluate $G(\cdot) + \eta$ with $\eta\sim\mathcal N(0, \Gamma)$. \Cref{fig:res:linear-noise} compares the two cases for the linear model with~$d_x = d_y = 100$. As it turns out, performance is very comparable for $\alpha=0.5$ but suffers for $\alpha=1$. This strengthens our case for intermediate (tempered) $\alpha$ values being more robust than $\alpha=1$. In the case where~$\alpha=0$, no evaluations of $G$ are required---we only need $\nabla G$, which we still evaluate exactly here---so the results are identical with and without noise.

\begin{figure}
    \centering
    \begin{tikzpicture}
    \begin{groupplot}[group style={group size=1 by 1, horizontal sep=2cm, vertical sep=2cm}, height=0.37\textwidth, width=0.53\textwidth, legend style={transpose legend, legend columns=3, draw=none}]

        \nextgroupplot[name=mainplot, xmode=log, log basis x=2, ymode=log, grid=major, xlabel={$r=s$}, ylabel={$W_2$}, legend to name=grouplegend, unbounded coords=discard, every axis plot/.append style={thick, mark options={solid, scale=0.9}}]
            \addplot[color=black, mark=x] table [x=dim, y={err_α=0.0_α=0.0}, col sep=comma]{data/linear-noise/no_noise-W2.csv};
            \addlegendentry{$\alpha=0\;\;\;\;$}
            \addplot[color=black, mark=o] table [x=dim, y={err_α=0.5_α=0.5 (incr.)}, col sep=comma]{data/linear-noise/no_noise-W2.csv};
            \addlegendentry{$\alpha=0.5\;\;\;\;$}
            \addplot[color=black, mark=square] table [x=dim, y={err_α=1.0_α=1.0 (incr.)}, col sep=comma]{data/linear-noise/no_noise-W2.csv};
            \addlegendentry{$\alpha=1.0\;\;\;\;$}
            \addplot[color=winered, dashed, mark=x] table [x=dim, y={err_α=0.0_α=0.0}, col sep=comma]{data/linear-noise/noise-W2.csv};
            \addlegendentry{$\alpha=0\mathrm{\;(noise)}$}
            \addplot[color=winered, dashed, mark=o] table [x=dim, y={err_α=0.5_α=0.5 (incr.)}, col sep=comma]{data/linear-noise/noise-W2.csv};
            \addlegendentry{$\alpha=0.5\mathrm{\;(noise)}$}
            \addplot[color=winered, dashed, mark=square] table [x=dim, y={err_α=1.0_α=1.0 (incr.)}, col sep=comma]{data/linear-noise/noise-W2.csv};
            \addlegendentry{$\alpha=1.0\mathrm{\;(noise)}$}
        \end{groupplot}

        \node at (mainplot.south) [below, yshift=-3\pgfkeysvalueof{/pgfplots/every axis title shift}] {\ref*{grouplegend}};
    \end{tikzpicture}
    \caption{Wasserstein-2 posterior errors for the linear test problem in function of the reduced input and output dimensions $r=s$, with either noisy (``noise'') or noise-free function evaluations}
    \label{fig:res:linear-noise}
\end{figure}
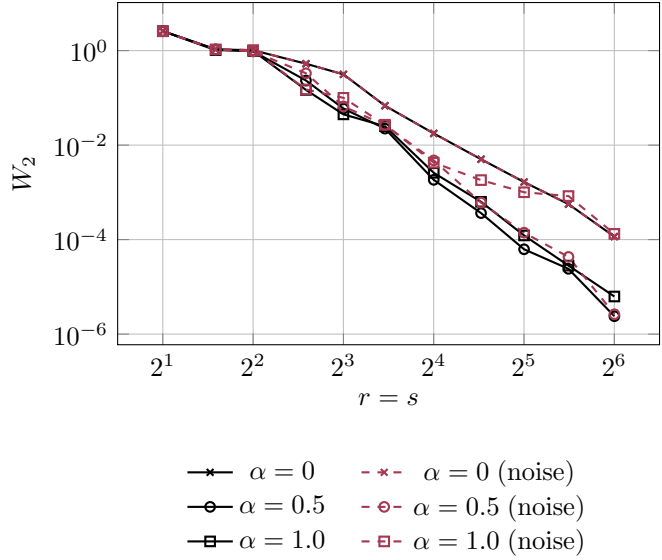

\paragraph{Using accumulated posterior information.} We now compare $H_{X\mid y^\dagger}^0$, $H_{X\mid y^\dagger}^{0.5}$, and $H_{X\mid y^\dagger}^1$, each with either 10 or 200 (sufficiently subsampled) MCMC samples per $\alpha$, to the accumulated tempered information described in \cref{sec:tempered:acc}. Accumulated posterior information is summed by taking the samples at $\{\alpha_i = 0.1i\}$ for $i$ going from $0$ to the specified upper bound, all with equal quadrature weights.

\begin{figure}
    \centering
        \begin{tikzpicture}
        \begin{groupplot}[group style={{group size=2 by 1, horizontal sep=2cm, vertical sep=2cm}},name=t2, height=0.37\textwidth, width=0.53\textwidth,legend style={transpose legend,legend columns=0,draw=none}]
        \nextgroupplot[xlabel={$r=s$}, ylabel={$W_2$}, legend to name=grouplegend1, xmode=log, log basis x=2, ymode=log, grid=major, legend to name=grouplegend2, ymin=0.5e-7, ymax=20, title={10 samples per $\alpha$ value}, every axis plot/.append style={thick, mark options={solid, scale=0.9}}]
            \addplot[black, mark=x] table [x=dim, y={err_α=0.0_α=0.0}, col sep=comma] {data/linear-combined-H/W2-10-100-100-true-true.csv};
            \addlegendentry{$\alpha=0\;\;$}
            \addplot[black, mark=o] table [x=dim, y={err_α=0.5_α=0.5 (incr.)}, col sep=comma] {data/linear-combined-H/W2-10-100-100-true-true.csv};
            \addlegendentry{$\alpha=0.5\;\;$}
            \addplot[black, mark=square] table [x=dim, y={err_α=1.0_α=1.0 (incr.)}, col sep=comma] {data/linear-combined-H/W2-10-100-100-true-true.csv};
            \addlegendentry{$\alpha=1\;\;$}
            \addplot[ferngreen, mark=o] table [x=dim, y={"err_α∈[0}, col sep=comma] {data/linear-combined-H/W2-10-100-100-true-true.csv};
            \addlegendentry{$H^{[0,0.5]}\;\;$}
            \addplot[ferngreen, mark=square] table [x=dim, y={err_all_all}, col sep=comma] {data/linear-combined-H/W2-10-100-100-true-true.csv};
            \addlegendentry{$H^{[0,1]}\;\;$}
        
        \nextgroupplot[xshift=-.1\textwidth, xlabel={$r=s$}, legend to name=grouplegend1, xmode=log, log basis x=2, ymode=log, grid=major, legend to name=grouplegend2, ymin=0.5e-7, ymax=20, title={200 samples per $\alpha$ value}, every axis plot/.append style={thick, mark options={solid, scale=0.9}}, yticklabels={}]
            \addplot[black, mark=x] table [x=dim, y={err_α=0.0_α=0.0}, col sep=comma] {data/linear-combined-H/W2-200-100-100-true-true.csv};
            \addlegendentry{$\alpha=0\;\;$}
            \addplot[black, mark=o] table [x=dim, y={err_α=0.5_α=0.5 (incr.)}, col sep=comma] {data/linear-combined-H/W2-200-100-100-true-true.csv};
            \addlegendentry{$\alpha=0.5\;\;$}
            \addplot[black, mark=square] table [x=dim, y={err_α=1.0_α=1.0 (incr.)}, col sep=comma] {data/linear-combined-H/W2-200-100-100-true-true.csv};
            \addlegendentry{$\alpha=1\;\;$}
            \addplot[ferngreen, mark=o] table [x=dim, y={"err_α∈[0}, col sep=comma] {data/linear-combined-H/W2-200-100-100-true-true.csv};
            \addlegendentry{$H^{[0,0.5]}\;\;$}
            \addplot[ferngreen, mark=square] table [x=dim, y={err_all_all}, col sep=comma] {data/linear-combined-H/W2-200-100-100-true-true.csv};
            \addlegendentry{$H^{[0,1]}\;\;$}
        \end{groupplot}
    
        \node at (6.2,0)[below, yshift=-2\pgfkeysvalueof{/pgfplots/every axis title shift}]{\ref*{grouplegend2}};
    \end{tikzpicture}
    \caption{Wasserstein-2 posterior errors for the linear test problem in function of the reduced input and output dimensions $r=s$, comparing different $\alpha$ values to accumulated information (following \cref{sec:tempered:acc})}
    \label{fig:res:linear-combined-H}
\end{figure}
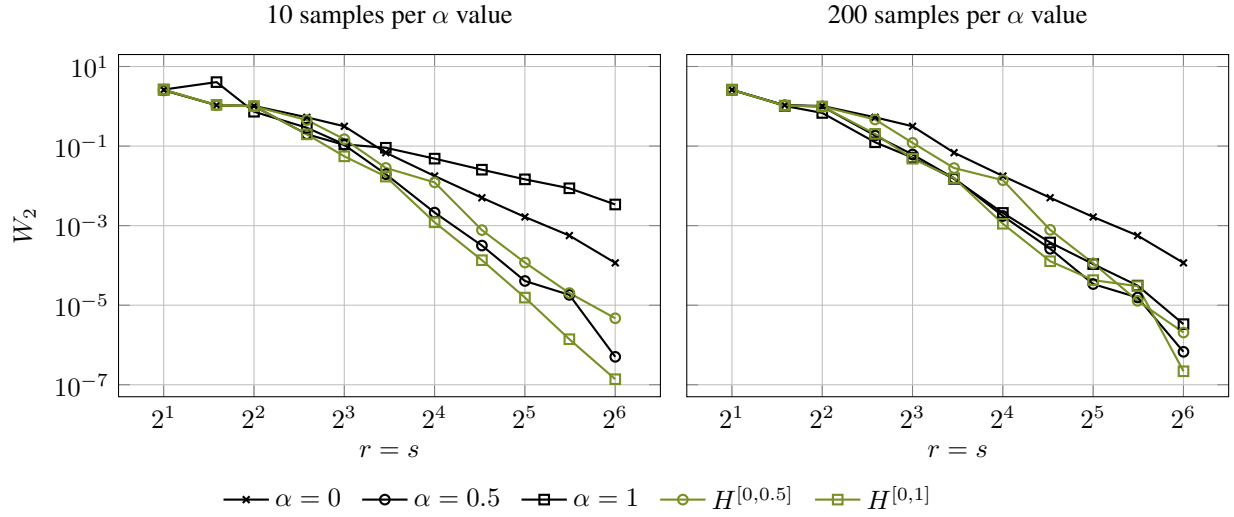

\Cref{fig:res:linear-combined-H} shows the result for the linear test problem with $d_x=d_y=100$. When using only few samples,~${\alpha=0.5}$ vastly outperforms other $\alpha$ values, showing again the robustness of using tempered distributions. Especially~$\alpha=1$ performs very poorly here. The aggregated data up to $\alpha=1$ performs slightly better than~$\alpha=0.5$, which in turn is slightly better than the aggregated data up to $\alpha=0.5$. When using more samples per $\alpha$, the familiar trends from before reemerge: $\alpha=0.5$ and $\alpha=1$ are comparable while $\alpha=0$ lags behind. Both versions with accumulated information perform similarly to the optimal $\alpha=0.5$; even though $H^{[0,1]}$ contains samples from the $\alpha=1$ distribution, performance is not diminished.

\paragraph{Using approximate samples to build the diagnostic matrices.}
When $G$ is expensive, drawing even a limited number of MCMC samples to build $H_{X\mid y^\dagger}^\alpha$ may not be feasible. Here, we suggest using ensemble Kalman methods---e.g., ensemble Kalman inversion (EKI), as described in \cref{ex:eki}---to instead draw approximate samples.

Concretely, throughout this section, ensemble Kalman-based sampling will use ETKI (mentioned in \cref{ex:eki}) with the adaptive timestepper from~\cite{KovStu19}, enhanced to stop additionally at the $\alpha$ values we need for each specific experiment. We now compare this approach to MCMC-based sampling in \cref{fig:res:linexp-mcmc-vs-ekp}, or the linear-exponential problem with $d_x=d_y=100$. The error of reduction with EKI-based samples seems slightly more erratic than that using MCMC-based samples, but otherwise they show comparable performance.

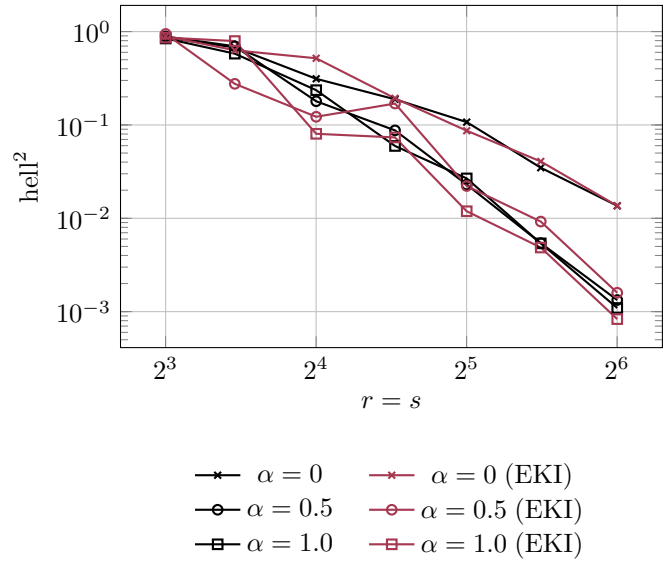
\begin{figure}
    \centering
    \begin{tikzpicture}
    \begin{groupplot}[group style={group size=1 by 1, horizontal sep=2cm, vertical sep=2cm}, height=0.37\textwidth, width=0.53\textwidth, legend style={transpose legend, legend columns=3, draw=none}]
        \nextgroupplot[name=mainplot, xmode=log, log basis x=2, ymode=log, grid=major, xlabel={$r=s$}, ylabel={$\operatorname{hell}^2$}, legend to name=grouplegend, unbounded coords=discard, every axis plot/.append style={thick, mark options={solid, scale=0.9}}]
            \addplot[color=black, mark=x] table [x=dim, y={err_0.0_mcmc}, col sep=comma]{data/linexp-mcmc-vs-ekp/errors.csv};
            \addlegendentry{$\alpha=0\;\;\;\;$}
            \addplot[color=black, mark=o] table [x=dim, y={err_0.5_mcmc}, col sep=comma]{data/linexp-mcmc-vs-ekp/errors.csv};
            \addlegendentry{$\alpha=0.5\;\;\;\;$}
            \addplot[color=black, mark=square] table [x=dim, y={err_1.0_mcmc}, col sep=comma]{data/linexp-mcmc-vs-ekp/errors.csv};
            \addlegendentry{$\alpha=1.0\;\;\;\;$}
            \addplot[color=winered, mark=x] table [x=dim, y={err_0.0_ekp}, col sep=comma]{data/linexp-mcmc-vs-ekp/errors.csv};
            \addlegendentry{$\alpha=0\mathrm{\;(EKI)}$}
            \addplot[color=winered, mark=o] table [x=dim, y={err_0.5_ekp}, col sep=comma]{data/linexp-mcmc-vs-ekp/errors.csv};
            \addlegendentry{$\alpha=0.5\mathrm{\;(EKI)}$}
            \addplot[color=winered, mark=square] table [x=dim, y={err_1.0_ekp}, col sep=comma]{data/linexp-mcmc-vs-ekp/errors.csv};
            \addlegendentry{$\alpha=1.0\mathrm{\;(EKI)}$}
        \end{groupplot}

        \node at (mainplot.south) [below, yshift=-3\pgfkeysvalueof{/pgfplots/every axis title shift}] {\ref*{grouplegend}};
    \end{tikzpicture}
    \caption{Hellinger posterior errors for the linear-exponential test problem in function of the reduced input and output dimensions $r=s$, with either EKI-based (``EKI'') or MCMC-based samples to compute the subspaces}
    \label{fig:res:linexp-mcmc-vs-ekp}
\end{figure}

\clearpage
\begin{figure}
    \centering
    \begin{tikzpicture}
    \begin{groupplot}[group style={group size=1 by 1, horizontal sep=2cm, vertical sep=2cm}, height=0.37\textwidth, width=0.53\textwidth, legend style={transpose legend, legend columns=3, draw=none}]

        \nextgroupplot[name=mainplot, xmode=log, ymode=log, grid=major, xlabel={$\gamma_0$}, ylabel={$\operatorname{hell}^2$}, legend to name=grouplegend, unbounded coords=discard, every axis plot/.append style={thick, mark options={solid, scale=0.9}}]
            \addplot[color=black, mark=x] table [x=gamma0, y={err_0.0_perfect}, col sep=comma]{data/linexp-comp-grads/by-gamma0.csv};
            \addlegendentry{$\alpha=0\;\;$}
            \addplot[color=black, mark=o] table [x=gamma0, y={err_0.5_perfect}, col sep=comma]{data/linexp-comp-grads/by-gamma0.csv};
            \addlegendentry{$\alpha=0.5\;\;$}
            \addplot[color=black, mark=square] table [x=gamma0, y={err_1.0_perfect}, col sep=comma]{data/linexp-comp-grads/by-gamma0.csv};
            \addlegendentry{$\alpha=1\;\;$}
            \addplot[color=ferngreen, mark=x] table [x=gamma0, y={err_0.0_linreg}, col sep=comma]{data/linexp-comp-grads/by-gamma0.csv};
            \addlegendentry{$\alpha=0\mathrm{\;(SL)}\;\;$}
            \addplot[color=ferngreen, mark=o] table [x=gamma0, y={err_0.5_linreg}, col sep=comma]{data/linexp-comp-grads/by-gamma0.csv};
            \addlegendentry{$\alpha=0.5\mathrm{\;(SL)}\;\;$}
            \addplot[color=ferngreen, mark=square] table [x=gamma0, y={err_1.0_linreg}, col sep=comma]{data/linexp-comp-grads/by-gamma0.csv};
            \addlegendentry{$\alpha=1\mathrm{\;(SL)}\;\;$}
        \end{groupplot}

        \node at (mainplot.south) [below, yshift=-3\pgfkeysvalueof{/pgfplots/every axis title shift}] {\ref*{grouplegend}};
    \end{tikzpicture}
    \caption{Hellinger posterior errors for the linear-exponential test problem in function of $\gamma_0$, with gradients that are exact or based on statistical linearization (``SL''), for $r=s=32$}
    \label{fig:res:linexp-comp-grads}
\end{figure}
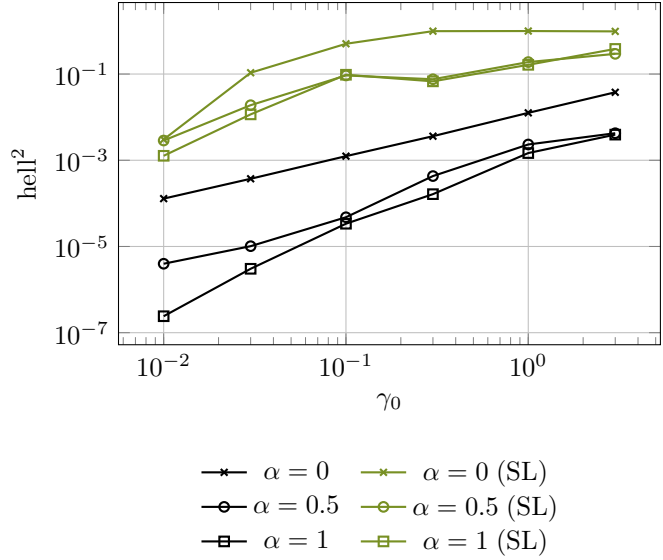

\paragraph{Approximating derivatives.} In some problems, gradients are not available. They are either too expensive or too noisy, or the model is simply non-differentiable or stochastic. In those cases, we can use the gradient approximations from \cref{sec:tempered:gradfree} instead. \Cref{fig:res:linexp-comp-grads} compares using the true gradient to statistical linearization (\cref{sec:tempered:gradfree}), as a derivative-free alternative, for the linear-exponential test problem with $d_x=d_y=100$ and $r=s=32$. Depending on $\gamma_0$, which determines whether the prior and posterior are restricted to a near-linear region or not, linearization-based gradient approximations are either relatively effective or do not work at all.

We note that an exponential factor is challenging to capture with a linearization; in addition, this experiment uses only a few hundred MCMC samples. \Cref{fig:res:darcy-num-samples} later explores gradient approximations in another nonlinear problem, in function of the number of samples.

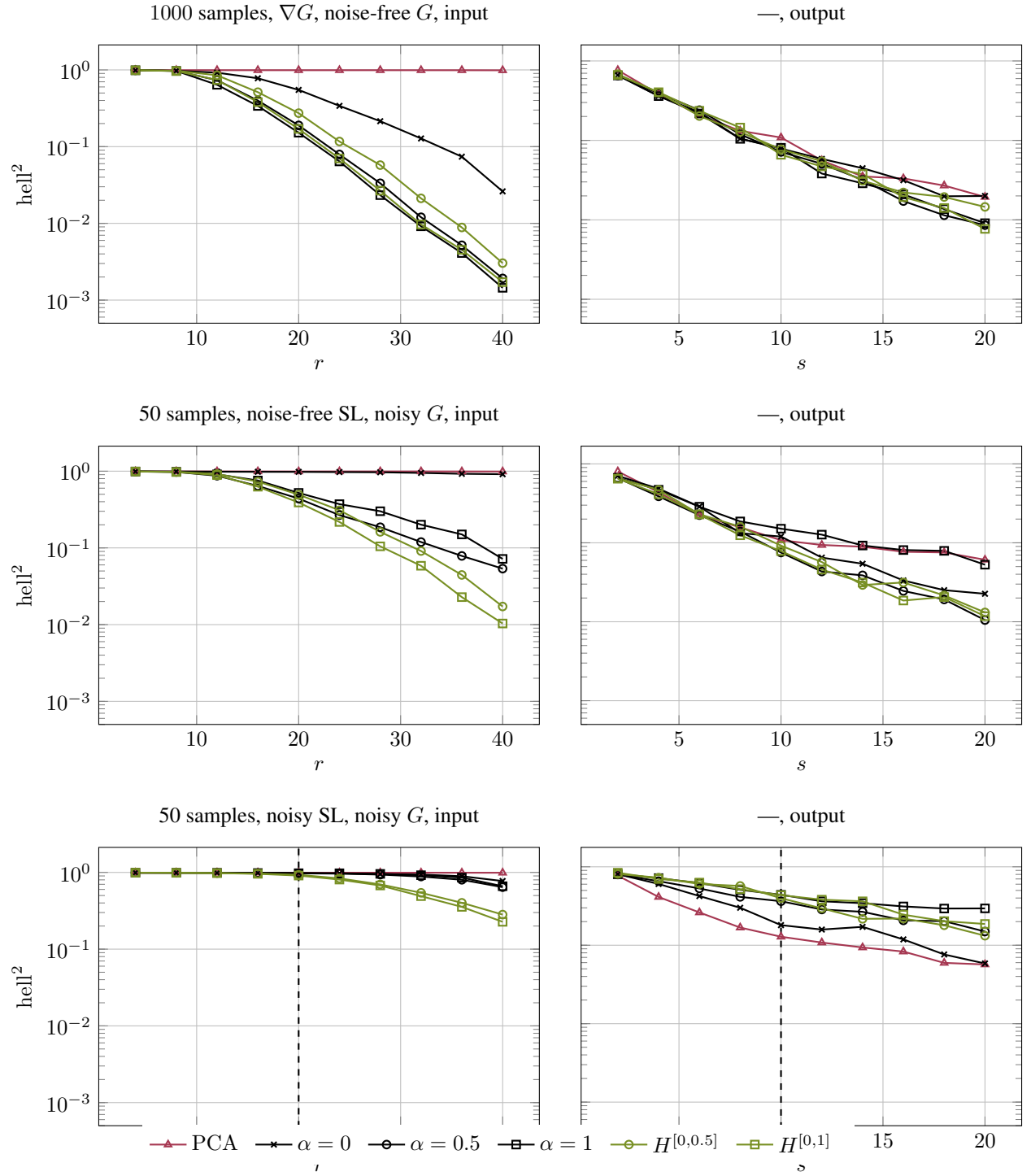
\begin{figure}
    \centering
        \begin{tikzpicture}
        \begin{groupplot}[group style={{group size=2 by 3, horizontal sep=2cm, vertical sep=2cm}},name=t2, height=0.37\textwidth, width=0.53\textwidth,legend style={transpose legend,legend columns=0,draw=none}]
        \nextgroupplot[xlabel={$r$}, ylabel={$\mathrm{hell}^2$}, ymode=log, ymin = 5e-4, grid=major, legend to name=grouplegend2, title={$1000$ samples, $\nabla G$, noise-free $G$, input}, every axis plot/.append style={thick, mark options={solid, scale=0.9}}]
            \addplot[winered, mark=triangle] table [x={in_r}, y={err_pca_u_pca_g}, col sep=comma] {data/darcy/errors-true-false-1000-none-perfect_summary_50.csv};
            \addlegendentry{$\mathrm{PCA}\;\;$}
            \addplot[black, mark=x] table [x={in_r}, y={err_Hu_0.0_ekp_perfect_Hg_0.0_ekp_perfect}, col sep=comma] {data/darcy/errors-true-false-1000-none-perfect_summary_50.csv};
            \addlegendentry{$\alpha=0\;\;$}
            \addplot[black, mark=o] table [x={in_r}, y={err_Hu_0.5_ekp_perfect_Hg_0.5_ekp_perfect}, col sep=comma] {data/darcy/errors-true-false-1000-none-perfect_summary_50.csv};
            \addlegendentry{$\alpha=0.5\;\;$}
            \addplot[black, mark=square] table [x={in_r}, y={err_Hu_1.0_ekp_perfect_Hg_1.0_ekp_perfect}, col sep=comma] {data/darcy/errors-true-false-1000-none-perfect_summary_50.csv};
            \addlegendentry{$\alpha=1\;\;$}
            \addplot[ferngreen, mark=o] table [x={in_r}, y={err_Hu_half_ekp_perfect_Hg_half_ekp_perfect}, col sep=comma] {data/darcy/errors-true-false-1000-none-perfect_summary_50.csv};
            \addlegendentry{$H^{[0,0.5]}\;\;$}
            \addplot[ferngreen, mark=square] table [x={in_r}, y={err_Hu_all_ekp_perfect_Hg_all_ekp_perfect}, col sep=comma] {data/darcy/errors-true-false-1000-none-perfect_summary_50.csv};
            \addlegendentry{$H^{[0,1]}\;\;$}
        
        \nextgroupplot[xshift=-.08\textwidth, xlabel={$s$}, ymode=log, ymin = 5e-4, grid=major, legend to name=grouplegend2, title={---, output}, every axis plot/.append style={thick, mark options={solid, scale=0.9}}, yticklabels={}]
            \addplot[winered, mark=triangle] table [x={out_r}, y={err_pca_u_pca_g}, col sep=comma] {data/darcy/errors-false-true-1000-none-perfect_summary_50.csv};
            \addlegendentry{$\mathrm{PCA}\;\;$}
            \addplot[black, mark=x] table [x={out_r}, y={err_Hu_0.0_ekp_perfect_Hg_0.0_ekp_perfect}, col sep=comma] {data/darcy/errors-false-true-1000-none-perfect_summary_50.csv};
            \addlegendentry{$\alpha=0\;\;$}
            \addplot[black, mark=o] table [x={out_r}, y={err_Hu_0.5_ekp_perfect_Hg_0.5_ekp_perfect}, col sep=comma] {data/darcy/errors-false-true-1000-none-perfect_summary_50.csv};
            \addlegendentry{$\alpha=0.5\;\;$}
            \addplot[black, mark=square] table [x={out_r}, y={err_Hu_1.0_ekp_perfect_Hg_1.0_ekp_perfect}, col sep=comma] {data/darcy/errors-false-true-1000-none-perfect_summary_50.csv};
            \addlegendentry{$\alpha=1\;\;$}
            \addplot[ferngreen, mark=o] table [x={out_r}, y={err_Hu_half_ekp_perfect_Hg_half_ekp_perfect}, col sep=comma] {data/darcy/errors-false-true-1000-none-perfect_summary_50.csv};
            \addlegendentry{$H^{[0,0.5]}\;\;$}
            \addplot[ferngreen, mark=square] table [x={out_r}, y={err_Hu_all_ekp_perfect_Hg_all_ekp_perfect}, col sep=comma] {data/darcy/errors-false-true-1000-none-perfect_summary_50.csv};
            \addlegendentry{$H^{[0,1]}\;\;$}
        
        \nextgroupplot[xlabel={$r$}, ylabel={$\mathrm{hell}^2$}, ymode=log, ymin = 5e-4, grid=major, legend to name=grouplegend2, title={50 samples, noise-free SL, noisy $G$, input}, every axis plot/.append style={thick, mark options={solid, scale=0.9}}]
            \addplot[winered, mark=triangle] table [x={in_r}, y={err_pca_u_pca_g}, col sep=comma] {data/darcy/errors-true-false-50-res-linreg_summary_50.csv};
            \addlegendentry{$\mathrm{PCA}\;\;$}
            \addplot[black, mark=x] table [x={in_r}, y={err_Hu_0.0_ekp_linreg_Hg_0.0_ekp_linreg}, col sep=comma] {data/darcy/errors-true-false-50-res-linreg_summary_50.csv};
            \addlegendentry{$\alpha=0\;\;$}
            \addplot[black, mark=o] table [x={in_r}, y={err_Hu_0.5_ekp_linreg_Hg_0.5_ekp_linreg}, col sep=comma] {data/darcy/errors-true-false-50-res-linreg_summary_50.csv};
            \addlegendentry{$\alpha=0.5\;\;$}
            \addplot[black, mark=square] table [x={in_r}, y={err_Hu_1.0_ekp_linreg_Hg_1.0_ekp_linreg}, col sep=comma] {data/darcy/errors-true-false-50-res-linreg_summary_50.csv};
            \addlegendentry{$\alpha=1\;\;$}
            \addplot[ferngreen, mark=o] table [x={in_r}, y={err_Hu_half_ekp_linreg_Hg_half_ekp_linreg}, col sep=comma] {data/darcy/errors-true-false-50-res-linreg_summary_50.csv};
            \addlegendentry{$H^{[0,0.5]}\;\;$}
            \addplot[ferngreen, mark=square] table [x={in_r}, y={err_Hu_all_ekp_linreg_Hg_all_ekp_linreg}, col sep=comma] {data/darcy/errors-true-false-50-res-linreg_summary_50.csv};
            \addlegendentry{$H^{[0,1]}\;\;$}
        
        \nextgroupplot[xlabel={$s$}, ymode=log, ymin = 5e-4, grid=major, legend to name=grouplegend2, title={---, output}, every axis plot/.append style={thick, mark options={solid, scale=0.9}}, yticklabels={}]
            \addplot[winered, mark=triangle] table [x={out_r}, y={err_pca_u_pca_g}, col sep=comma] {data/darcy/errors-false-true-50-res-linreg_summary_50.csv};
            \addlegendentry{$\mathrm{PCA}\;\;$}
            \addplot[black, mark=x] table [x={out_r}, y={err_Hu_0.0_ekp_linreg_Hg_0.0_ekp_linreg}, col sep=comma] {data/darcy/errors-false-true-50-res-linreg_summary_50.csv};
            \addlegendentry{$\alpha=0\;\;$}
            \addplot[black, mark=o] table [x={out_r}, y={err_Hu_0.5_ekp_linreg_Hg_0.5_ekp_linreg}, col sep=comma] {data/darcy/errors-false-true-50-res-linreg_summary_50.csv};
            \addlegendentry{$\alpha=0.5\;\;$}
            \addplot[black, mark=square] table [x={out_r}, y={err_Hu_1.0_ekp_linreg_Hg_1.0_ekp_linreg}, col sep=comma] {data/darcy/errors-false-true-50-res-linreg_summary_50.csv};
            \addlegendentry{$\alpha=1\;\;$}
            \addplot[ferngreen, mark=o] table [x={out_r}, y={err_Hu_half_ekp_linreg_Hg_half_ekp_linreg}, col sep=comma] {data/darcy/errors-false-true-50-res-linreg_summary_50.csv};
            \addlegendentry{$H^{[0,0.5]}\;\;$}
            \addplot[ferngreen, mark=square] table [x={out_r}, y={err_Hu_all_ekp_linreg_Hg_all_ekp_linreg}, col sep=comma] {data/darcy/errors-false-true-50-res-linreg_summary_50.csv};
            \addlegendentry{$H^{[0,1]}\;\;$}
        
        \nextgroupplot[xlabel={$r$}, ylabel={$\mathrm{hell}^2$}, ymode=log, ymin = 5e-4, grid=major, legend to name=grouplegend2, title={50 samples, noisy SL, noisy $G$, input}, every axis plot/.append style={thick, mark options={solid, scale=0.9}}]
            \addplot[winered, mark=triangle] table [x={in_r}, y={err_pca_u_pca_g}, col sep=comma] {data/darcy/errors-true-false-50-all-linreg_summary_50.csv};
            \addlegendentry{$\mathrm{PCA}\;\;$}
            \addplot[black, mark=x] table [x={in_r}, y={err_Hu_0.0_ekp_linreg_Hg_0.0_ekp_linreg}, col sep=comma] {data/darcy/errors-true-false-50-all-linreg_summary_50.csv};
            \addlegendentry{$\alpha=0\;\;$}
            \addplot[black, mark=o] table [x={in_r}, y={err_Hu_0.5_ekp_linreg_Hg_0.5_ekp_linreg}, col sep=comma] {data/darcy/errors-true-false-50-all-linreg_summary_50.csv};
            \addlegendentry{$\alpha=0.5\;\;$}
            \addplot[black, mark=square] table [x={in_r}, y={err_Hu_1.0_ekp_linreg_Hg_1.0_ekp_linreg}, col sep=comma] {data/darcy/errors-true-false-50-all-linreg_summary_50.csv};
            \addlegendentry{$\alpha=1\;\;$}
            \addplot[ferngreen, mark=o] table [x={in_r}, y={err_Hu_half_ekp_linreg_Hg_half_ekp_linreg}, col sep=comma] {data/darcy/errors-true-false-50-all-linreg_summary_50.csv};
            \addlegendentry{$H^{[0,0.5]}\;\;$}
            \addplot[ferngreen, mark=square] table [x={in_r}, y={err_Hu_all_ekp_linreg_Hg_all_ekp_linreg}, col sep=comma] {data/darcy/errors-true-false-50-all-linreg_summary_50.csv};
            \addlegendentry{$H^{[0,1]}\;\;$}
            \draw[black, dashed, thick] (20, 1e-4) -- (20, 1e1);
        
        \nextgroupplot[xlabel={$s$}, legend to name=grouplegend1, ymode=log, ymin = 5e-4, grid=major, title={---, output}, every axis plot/.append style={thick, mark options={solid, scale=0.9}}, yticklabels={}]
            \addplot[winered, mark=triangle] table [x={out_r}, y={err_pca_u_pca_g}, col sep=comma] {data/darcy/errors-false-true-50-all-linreg_summary_50.csv};
            \addlegendentry{$\mathrm{PCA}\;\;$}
            \addplot[black, mark=x] table [x={out_r}, y={err_Hu_0.0_ekp_linreg_Hg_0.0_ekp_linreg}, col sep=comma] {data/darcy/errors-false-true-50-all-linreg_summary_50.csv};
            \addlegendentry{$\alpha=0\;\;$}
            \addplot[black, mark=o] table [x={out_r}, y={err_Hu_0.5_ekp_linreg_Hg_0.5_ekp_linreg}, col sep=comma] {data/darcy/errors-false-true-50-all-linreg_summary_50.csv};
            \addlegendentry{$\alpha=0.5\;\;$}
            \addplot[black, mark=square] table [x={out_r}, y={err_Hu_1.0_ekp_linreg_Hg_1.0_ekp_linreg}, col sep=comma] {data/darcy/errors-false-true-50-all-linreg_summary_50.csv};
            \addlegendentry{$\alpha=1\;\;$}
            \addplot[ferngreen, mark=o] table [x={out_r}, y={err_Hu_half_ekp_linreg_Hg_half_ekp_linreg}, col sep=comma] {data/darcy/errors-false-true-50-all-linreg_summary_50.csv};
            \addlegendentry{$H^{[0,0.5]}\;\;$}
            \addplot[ferngreen, mark=square] table [x={out_r}, y={err_Hu_all_ekp_linreg_Hg_all_ekp_linreg}, col sep=comma] {data/darcy/errors-false-true-50-all-linreg_summary_50.csv};
            \addlegendentry{$H^{[0,1]}\;\;$}
            \draw[black, dashed, thick] (10, 1e-4) -- (10, 1e1);
        \end{groupplot}
    
        \node at (6.5,-12)[below, yshift=-2\pgfkeysvalueof{/pgfplots/every axis title shift}]{\ref*{grouplegend1}};
    \end{tikzpicture}
    \caption{Hellinger posterior errors for the Darcy test problem in function of $r$ and $s$}
    \label{fig:res:darcy}
\end{figure}

\paragraph{Combining everything in a Darcy flow problem.} We now apply the various dimension reduction methods to a slightly more complex problem: two-dimensional Darcy flow. We use approximate EKI-based sampling and will vary the other approximations. This experiment is nonlinear and the output samples a grid in the spatial field. We expect PCA to do well here in the output space. To deal with random error fluctuations due to randomly generated observations and EKI samples, each plot for this experiment shows the median error out of $32$ runs. In addition, in order to robustly combine noisy samples with statistical linearization, we add a regularization nugget $\frac{\Tr\Gamma}JI$ to $C^{xx}$ before inversion whenever SL is applied to noisy data. Again, accumulated posterior information is collected at $\alpha$s that are multiples of $0.1$. We consider various settings.

\begin{enumerate}
    \item First, as a best-case scenario, we use $J=1000$ EKI samples, exact gradients, and no noise. The first row in \cref{fig:res:darcy} shows robust dimension reduction in the input space, with all methods performing almost identically except for $\alpha=0$ and especially PCA (which makes sense due to the uninformative prior). In the output space, we get a similar picture, but with $\alpha=0$ and PCA performing almost as well as the other techniques.
    \item We then add the complicating factors from this subsection. We reduce the number of samples to $J=50$ and add noise to them. We also use statistical linearization instead of exact gradients, but (somewhat artificially) use the samples \emph{without} noise for the SL computation. The error is computed using the noise-free model; the noise is added only to find the reduced spaces. The results are shown in the second row of \cref{fig:res:darcy}. Most configurations suffer heavily in the input space, with only the accumulated diagnostic matrix still effectively reducing the dimensions. In the output space, the difference with the best-case scenario is smaller; mainly $\alpha=1$ is impacted.
    \item If we then, more realistically, use the noisy samples also for the SL computation, the last row of \cref{fig:res:darcy} shows the picture worsening. In the input space even the leading performer, the accumulated diagnostic matrix, becomes significantly worse. This time, the output reduced space's quality is also heavily impacted, with PCA---which does not use gradients---now the best method.
\end{enumerate}
In summary, $\alpha$-LIS shows some robustness, particularly when used in conjunction with accumulation, in the face of approximations and noisy samples---more than DB-LIS, for instance---but struggles from a certain point onward. \Cref{fig:res:darcy-num-samples} keeps the most complex setting from the last row of \cref{fig:res:darcy}, and fixes $r=20$ and $s=10$ in the input and output spaces, respectively, while varying the number of EKI samples $J$. As $J$ increases, the quality of the reduced spaces improves significantly to performance comparable with the first row of \cref{fig:res:darcy}, indicating that the challenging assumptions in this setting can be overcome by using more ensemble members. In fact, we can almost recover the performance without noise for very large $J$.

\begin{figure}
    \centering
        \begin{tikzpicture}
        \begin{groupplot}[group style={{group size=2 by 1, horizontal sep=2cm, vertical sep=2cm}},name=t2, height=0.37\textwidth, width=0.53\textwidth,legend style={transpose legend,legend columns=0,draw=none}]
        \nextgroupplot[xlabel={$J$}, ylabel={$\mathrm{hell}^2$}, legend to name=grouplegend1, xmode=log, ymode=log, grid=major, legend to name=grouplegend2, ymin=5e-2, ymax=1.3, title={Input space reduction}, every axis plot/.append style={thick, mark options={solid, scale=0.9}}]
            \addplot[winered, mark=triangle] table [x={num_ekp_samps}, y={err_pca_u_pca_g}, col sep=comma] {data/darcy-num-samples/errors-true-false-all-linreg_summary_50.csv};
            \addlegendentry{$\mathrm{PCA}\;\;$}
            \addplot[black, mark=x] table [x={num_ekp_samps}, y={err_Hu_0.0_ekp_linreg_Hg_0.0_ekp_linreg}, col sep=comma] {data/darcy-num-samples/errors-true-false-all-linreg_summary_50.csv};
            \addlegendentry{$\alpha=0\;\;$}
            \addplot[black, mark=o] table [x={num_ekp_samps}, y={err_Hu_0.5_ekp_linreg_Hg_0.5_ekp_linreg}, col sep=comma] {data/darcy-num-samples/errors-true-false-all-linreg_summary_50.csv};
            \addlegendentry{$\alpha=0.5\;\;$}
            \addplot[black, mark=square] table [x={num_ekp_samps}, y={err_Hu_1.0_ekp_linreg_Hg_1.0_ekp_linreg}, col sep=comma] {data/darcy-num-samples/errors-true-false-all-linreg_summary_50.csv};
            \addlegendentry{$\alpha=1\;\;$}
            \addplot[ferngreen, mark=o] table [x={num_ekp_samps}, y={err_Hu_half_ekp_linreg_Hg_half_ekp_linreg}, col sep=comma] {data/darcy-num-samples/errors-true-false-all-linreg_summary_50.csv};
            \addlegendentry{$H^{[0,0.5]}\;\;$}
            \addplot[ferngreen, mark=square] table [x={num_ekp_samps}, y={err_Hu_all_ekp_linreg_Hg_all_ekp_linreg}, col sep=comma] {data/darcy-num-samples/errors-true-false-all-linreg_summary_50.csv};
            \addlegendentry{$H^{[0,1]}\;\;$}

        \nextgroupplot[xshift=-0.08\textwidth, xlabel={$J$}, legend to name=grouplegend1, xmode=log, ymode=log, grid=major, legend to name=grouplegend2, ymin=5e-2, ymax=1.3, title={Output space reduction}, every axis plot/.append style={thick, mark options={solid, scale=0.9}}, yticklabels={}]
            \addplot[winered, mark=triangle] table [x={num_ekp_samps}, y={err_pca_u_pca_g}, col sep=comma] {data/darcy-num-samples/errors-false-true-all-linreg_summary_50.csv};
            \addlegendentry{$\mathrm{PCA}\;\;$}
            \addplot[black, mark=x] table [x={num_ekp_samps}, y={err_Hu_0.0_ekp_linreg_Hg_0.0_ekp_linreg}, col sep=comma] {data/darcy-num-samples/errors-false-true-all-linreg_summary_50.csv};
            \addlegendentry{$\alpha=0\;\;$}
            \addplot[black, mark=o] table [x={num_ekp_samps}, y={err_Hu_0.5_ekp_linreg_Hg_0.5_ekp_linreg}, col sep=comma] {data/darcy-num-samples/errors-false-true-all-linreg_summary_50.csv};
            \addlegendentry{$\alpha=0.5\;\;$}
            \addplot[black, mark=square] table [x={num_ekp_samps}, y={err_Hu_1.0_ekp_linreg_Hg_1.0_ekp_linreg}, col sep=comma] {data/darcy-num-samples/errors-false-true-all-linreg_summary_50.csv};
            \addlegendentry{$\alpha=1\;\;$}
            \addplot[ferngreen, mark=o] table [x={num_ekp_samps}, y={err_Hu_half_ekp_linreg_Hg_half_ekp_linreg}, col sep=comma] {data/darcy-num-samples/errors-false-true-all-linreg_summary_50.csv};
            \addlegendentry{$H^{[0,0.5]}\;\;$}
            \addplot[ferngreen, mark=square] table [x={num_ekp_samps}, y={err_Hu_all_ekp_linreg_Hg_all_ekp_linreg}, col sep=comma] {data/darcy-num-samples/errors-false-true-all-linreg_summary_50.csv};
            \addlegendentry{$H^{[0,1]}\;\;$}
        \end{groupplot}
    
        \node at (6.5,0)[below, yshift=-2\pgfkeysvalueof{/pgfplots/every axis title shift}]{\ref*{grouplegend2}};
    \end{tikzpicture}
    \caption{Hellinger posterior errors for the Darcy test problem in function of $J$, for the same settings as the last row in \cref{fig:res:darcy} and with $r=20$ (input space reduction) and $s=10$ (output space reduction), indicated in \cref{fig:res:darcy} by vertical lines}
    \label{fig:res:darcy-num-samples}
\end{figure}
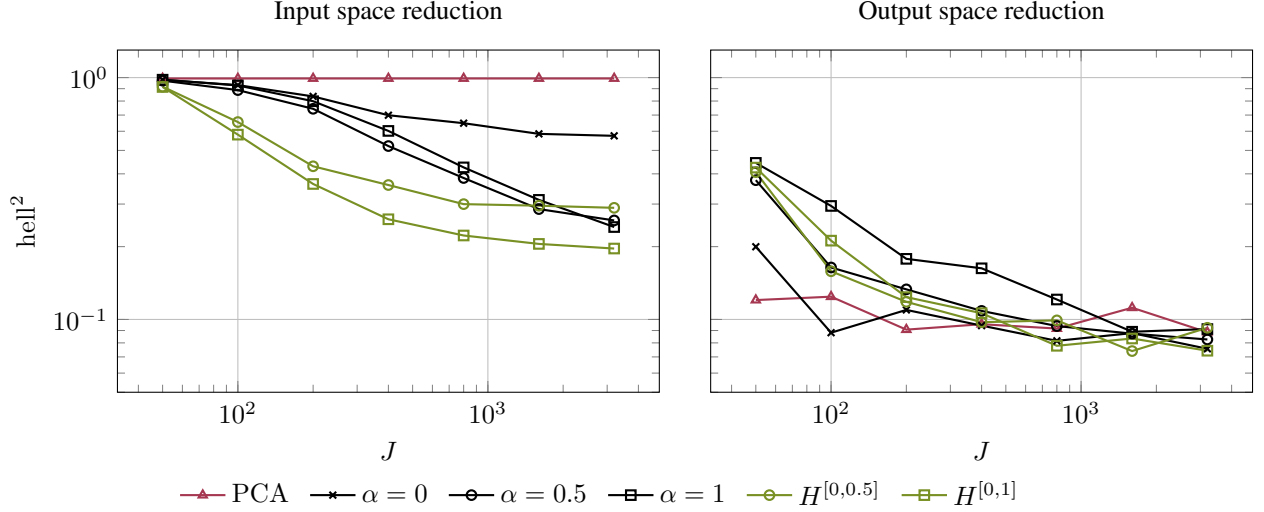

\section{Application to the \emph{calibrate, emulate, sample} framework} \label{sec:ces}
The \emph{calibrate, emulate, sample} (CES) framework~\cite{CleGarLanSchStu21} partitions the task of sampling from $\pi(x\mid y^\dagger)$: \begin{enumerate}
    \item \textbf{Calibrate.} Apply a computationally cheap, low-accuracy sampler such as ensemble Kalman inversion (EKI, \cref{ex:eki}) or the ensemble Kalman sampler (EKS). In this paper, we focus on samplers that produce samples of tempered distributions as iterations, such as EKI.
    \item \textbf{Emulate.} Use the samples from the calibration step to build an emulator $\widetilde{G}$, using function regression---e.g., Gaussian process (GP) regression---for the forward map $G$.
    \item \textbf{Sample.} Apply a computationally expensive, high-accuracy sampler such as Markov chain Monte Carlo (MCMC) to $\widetilde G$.
\end{enumerate}
Samplers based on the ensemble Kalman methodology are appropriate in the calibration step, as they are model-agnostic, robust to noise, and performant in very high-dimensional settings. Emulators, on the other hand, scale poorly in the number of dimensions when trained on a relatively modest number of data samples. These properties, combined with the approximate posterior samples that are available from the calibration stage without any additional cost, make the CES framework well-suited to the dimension reduction methods discussed in \cref{sec:tempered}. We propose the following methodology.
\begin{enumerate}
    \item \textbf{Calibrate.} The calibration step is performed in the original unreduced input and output spaces. Collect the sequence of approximate tempered-posterior samples. 
    \item \textbf{Encode.} Use samples from \textbf{Calibrate} to construct a reduced space with affine encoders $U_r$ and $V_s$, using the derivative-free extension of the methodology.
    \item \textbf{Emulate.} Encode all training data samples from \textbf{Calibrate} and train the emulator in the reduced space on this data.
    \item \textbf{Sample.} Sample from the emulator in the reduced space\footnote{Alternatively, sampling can be done in the full space (e.g., when drawing conditional prior samples is hard).}. The final reduced-posterior samples are projected back into the full space by drawing conditional samples from $\pi_{X_\perp\mid X_r}$.
\end{enumerate}
In practice, even approximate sampling with ensemble Kalman methods in the calibration step can dominate the cost of the algorithm if $G$ is very expensive. In those cases, a solution is to run the ensemble Kalman method for only a few steps. In the case of EKI (\cref{ex:eki}), simulating until time $\alpha$ results in input-output pairs $\{x^{(j)}, G(x^{(j)})\}_{j=1}^J$, where the $x^{(j)}$ are approximate samples from $\pi_{X\mid Y}^{(\alpha)}(\cdot\,\mid y^\dagger)$. It also computes such samples for various $0\le\alpha_k<\alpha$.
\begin{remark}[Emulator regression points]
    Typically in the CES framework, regression points for the emulators are assumed (though not required) to be sampled approximately from $\pi_{X\mid Y}(\cdot\,\mid y^\dagger)$ (so, with $\alpha=1$). However, not only is using a smaller $\alpha$ computationally interesting when $G$ is expensive, it might actually result in better emulators. This is observed in practice and may be related to the result~\cite[Theorem 3]{helin2023introductiongaussianprocessregression}.
\end{remark}

The aim of the emulation stage is then to make $\pi_{Y_s\mid X_r}$ cheaper to evaluate, in order to make sampling from $\pi_{X\mid Y}^*$ in the later sampling stage feasible. In the case of a statistical emulator, such as a GP, we build an emulator distribution $(\widetilde{G}_s, \widetilde\Gamma_s)$, such that
\begin{equation}
    \log\pi_{Y_s\mid X_r}(y_s\mid\,\cdot) \approx -\frac12\norm{y_s - \widetilde G_s(\cdot)}^2_{\widetilde\Gamma_s(\cdot) } -\frac{1}{2} \log\det(\widetilde\Gamma_s(\cdot)) +  \log C \eqqcolon \log\widetilde\pi_{Y_s\mid X_r}(y_s\mid\,\cdot).
\end{equation}
We can define
\begin{equation}
    \widetilde g^{(j)} \coloneqq V_s^\top G(x^{(j)}),
\end{equation}
and train the emulator distribution $(\widetilde G_s, \widetilde{\Gamma}_s)$ by regression on the data $\{x_r^{(j)}, \widetilde g^{(j)}\}_{j=1}^J$. In the case of a deterministic emulator, one simply learns $\widetilde{G}_s$ and defines $\widetilde{\Gamma}_s \coloneqq V_s^\top \Gamma V_s$. Lastly, in the sample stage, we use an algorithm such as MCMC to sample $x$ values distributed according to
\begin{equation}
    \widetilde\pi^*_{X\mid Y}(x\mid y^\dagger) \propto \widetilde\pi_{Y_s\mid X_r}(y_s^\dagger\mid x_r)\pi_{X_r}(x_r)\pi_{X_\perp\mid X_r}(x_\perp\mid x_r)
\end{equation}
by first sampling $x_r$ and then conditionally sampling $x_\perp$. 

\subsection{Example: Lorenz `96 with spatially-dependent forcing} \label{sec:ces:lorenz}
While an exhaustive benchmarking of $\alpha$-LIS reduced spaces within the CES algorithm is beyond the scope of this paper, we will illustrate with an example that prior-based, tempered, and accumulated $\alpha$-LIS show robustness and can outperform PCA, even when working with few samples and large noise on evaluations.

As an example of $\alpha$-LIS applied within the CES framework, we will look at a model problem that is very challenging to sample with MCMC directly. In particular, the forward map $G$ contains observables of a chaotic state, which leads to a highly noisy likelihood in MCMC if applied directly to the system. Without using complex annealing strategies or drastically increasing the simulator cost within $G$, MCMC will be unable to sample this posterior effectively. On the other hand, CES only interfaces with the noisy model during the calibration stage---with EKI, which is well suited to sampling with a noisy likelihood. The regression-based emulator $\widetilde{G}$ will be faster and smoother than the true map $G$ by design, which makes the sampling stage feasible. 

The Lorenz `96 ODE system~\cite{Lor96} is commonly used to represent a toy atmospheric model. Given~$N$ states, we evolve the system
\begin{equation} \label{eq:res:lorenz}
\begin{aligned}
    \frac{\mathrm du_i}{\mathrm dt} &= (u_{i+1} - u_{i-2}) u_{i-1} - u_i + F_i, \qquad i = 1, \dots, N,\\
    &\text{ subject to the cyclic boundary condition }u_{j} = u_{N+j}.
\end{aligned}
\end{equation}
This system contains nonlinear advection and linear damping, and we extend the classical formulation by introducing a forcing term $F_i$ \cite{DunGjiMorSch25_pre} that depends on the state index but not on time. Moderate-magnitude forcings $F_i$ can induce chaotic behavior in this system. 

We denote the full state and forcing vectors by $u$ and $F$, respectively, and define by $u(t; F, u^{(0)})$ the state vector at time $t$ subject to forcing $F$ and with initial state $u^{(0)}$. We are interested in learning the forcing from data in this chaotic regime. Following~\cite{DunGarSchStu21,DunGjiMorSch25_pre}, one can learn the parameters of the chaotic system by fitting certain statistics of the system. For a given $u^{(0)}$ and end time $T$, the empirical mean and standard deviation over the window $[0,T]$ are given by
\begin{subequations}
\begin{align}
\mathrm{mean}(u;F,u^{(0)},T) &\coloneqq \frac{1}{T}\int_0^T u(t;F,u^{(0)}) \,\mathrm{d}t, \\
\mathrm{std}(u;F,u^{(0)},T) &\coloneqq \left(\mathrm{mean}(u;f,u^{(0)},T)^2-\frac{1}{T}\int_0^T u(t;F,u^{(0)})^2 \,\mathrm{d}t\right)^{1/2}.
\end{align}
\end{subequations}
We define a forward map as
\begin{equation}
    G(F; T) = \begin{bmatrix}
        \mathrm{mean}(u;F,u^{(0)},T)\\
        \mathrm{std}(u;F,u^{(0)},T)
    \end{bmatrix}, \qquad \qquad u^{(0)}\sim \rho(F),
\end{equation}
where $\rho(F)$ is the invariant measure of the Lorenz `96 system with the forcing $F$. To simulate $u^{(0)}\sim\rho(F)$, we run a short spinning-up phase before aggregating the statistics. To motivate $G$, consider the observation $y^\dagger\sim G(F^\dagger;T)$. It would appear initially that to estimate $F^\dagger$ from $y^\dagger$ one must infer the true initial condition $(u^{(0)})^\dagger$. However, aggregate statistics (for sufficiently large $T$) of dynamical systems that are sufficiently mixing can exploit central-limit theorems (e.g.,~\cite{BurDun87,You98}). Lorenz `96 is observed to satisfy this mixing property, and so we assume that $G$ is statistically equivalent from any initial condition from the invariant measure $\rho(F)$.

We can write this in the familiar Bayesian framework, where $F$ plays the role of our abstract parameters $X$, as 
\begin{equation}
    Y \sim G(F;T) \approx G_\infty(F) + \eta, \qquad \text{with} \qquad \eta \sim \mathcal N\Bigl(0,\frac{1}{T}\Gamma\Bigr). 
\end{equation}
Here, $G_\infty \coloneqq G(F;\infty)$ is the ergodic average and is deterministic, while $\Gamma$ characterizes the noise induced from having an unknown initial condition on $G$. Since integration up to $T\gg1$ is infeasible\footnote{The magnitude of the Lorenz `96 noise is not small until $T\approx 10^3$.}, and since we must always evaluate $G$ with \emph{some} initial condition, evaluations of $G$ are unavoidably ``noisy'' (as in some experiments in the previous section). Hence, we cannot use standard optimization approaches to learn the parameters.

For Lorenz `96, ergodicity is observed at just $T \approx 10$, making for relatively cheap but rather noisy evaluations of $G$. This motivates using methods such as EKI that are robust to noise and can perform inversion these noisy integrations~\cite{DunGjiMorSch25_pre,DunDunStuWol22, DunGarSchStu21}. 

For our experiment we use $T=20$ (after a spin-up of length $4$), which results in rather noisy evaluations, with timesteps $\Delta t=0.01$. We set dimension $d_x=N=40$, such that $d_y=80$, and choose a true forcing
\begin{equation}
    F_i^\dagger = 8 + 6\sin\left(4\pi\,\frac{i-1}{N-1}\right).
\end{equation}
We estimate $\Gamma$ as the covariance of $Y$ for 200 random initial conditions, plus $10^{-2}I$. We compare \begin{itemize}
    \item an informative prior covariance $(\Gamma_0)_{ij} = 25\exp(-|i-j|)$, nudging $F$ to be smooth, to
    \item an uninformative prior covariance $\Gamma_0=25I$,
\end{itemize}
and set the prior equal to $\mathcal N({\bf8}, \Gamma_0)$, where $\bf8$ is a vector filled with the number 8.

We use the CES pipeline with dimension reduction as introduced before: we run EKI with 200 samples, compute the reduced space with either all 200 or a subset of 50 random samples per $\alpha$, train an emulator using regression over random Fourier features\footnote{Random-feature regression (RFR)~\cite{DunNelMul25} can be thought of as Gaussian-process regression (GPR)~\cite{RasWil06} with a random low-rank kernel representation for computational scaling advantage. Exact relationships link random-feature families and Gaussian-process kernels; here, we leverage the fact that random Fourier features approximate a GP with squared-exponential ARD kernel~\cite{RahRec07}. In all experiments we draw 200 features, add a $10^{-6}$ nugget term, and perform training to tune the effective lengthscale hyperparameters.} on the 200 samples at $\alpha=0.5$, and then sample from this emulator in the reduced space with MCMC. Like before, any accumulated posterior information comes from $\alpha\in\{0, 0.1, \ldots\}$. Based on the results in \cref{fig:res:darcy-num-samples}, we vary the input space reduction method while the output space is reduced with either PCA (when the input is reduced with PCA) or 0-LIS (otherwise). We approximate derivatives with SL, but notably use all samples up to a given $\alpha_k$ to compute the SL approximation at $\alpha_k$. This is important to avoid unreliable gradient approximations at larger $\alpha_k$, where ensembles cluster more than at small $\alpha$s.

As we cannot evaluate the true posterior in this setting, our metric is the distance of the reduced posterior mean to the true forcing parameter $F^\dagger$. Although the mean of even the full posterior cannot be expected to match $F^\dagger$ due to the noise in the evaluation $y^\dagger$ and the nonlinearity of $G$, this metric is sufficient to separate very poor reduced posteriors from reasonably good ones.

We implemented the emulation stage using the \texttt{CalibrateEmulateSample.jl} package~\cite{Dun_etal24}; in fact, this experiment series was modified from the \texttt{Lorenz/*\_spatial\_dep.jl} examples in that package. MCMC sampling and dimension reduction using $\alpha$-LIS are also available in the archive linked in the section on data availability.

\begin{figure}
    \centering
        \begin{tikzpicture}
        \begin{groupplot}[group style={{group size=2 by 2, horizontal sep=2cm, vertical sep=1.8cm}},name=t2, height=0.37\textwidth, width=0.53\textwidth,legend style={transpose legend,legend columns=0,draw=none}]

        \nextgroupplot[xlabel={$r=s$}, ylabel={$\norm{F^\dagger - \mathrm{sample\;mean}}$}, ymin = 0, ymax = 31, grid=major, legend to name=grouplegend2, title={50 samples, informative prior}, every axis plot/.append style={thick, mark options={solid, scale=0.9}}]
            \addplot[winered, mark=triangle] table [x={in_r}, y={err_pca_u_pca_g}, col sep=comma] {data/lorenz/errors-50-1.0-all-1-1.csv};
            \addlegendentry{$\mathrm{PCA}\;\;$}
            \addplot[black, mark=x] table [x={in_r}, y={err_Hu_0.0_ekp_linreg_Hg_0.0_ekp_linreg}, col sep=comma] {data/lorenz/errors-50-1.0-all-1-1.csv};
            \addlegendentry{$\alpha=0\;\;$}
            \addplot[black, mark=o] table [x={in_r}, y={err_Hu_0.5_ekp_linreg_Hg_0.0_ekp_linreg}, col sep=comma] {data/lorenz/errors-50-1.0-all-1-1.csv};
            \addlegendentry{$\alpha=0.5\;\;$}
            \addplot[black, mark=square] table [x={in_r}, y={err_Hu_1.0_ekp_linreg_Hg_0.0_ekp_linreg}, col sep=comma] {data/lorenz/errors-50-1.0-all-1-1.csv};
            \addlegendentry{$\alpha=1\;\;$}
            \addplot[ferngreen, mark=o] table [x={in_r}, y={err_Hu_half_ekp_linreg_Hg_0.0_ekp_linreg}, col sep=comma] {data/lorenz/errors-50-1.0-all-1-1.csv};
            \addlegendentry{$H^{[0,0.5]}\;\;$}
            \addplot[ferngreen, mark=square] table [x={in_r}, y={err_Hu_all_ekp_linreg_Hg_0.0_ekp_linreg}, col sep=comma] {data/lorenz/errors-50-1.0-all-1-1.csv};
            \addlegendentry{$H^{[0,1]}\;\;$}

        \nextgroupplot[xshift=-.1\textwidth, xlabel={$r=s$}, ymin = 0, ymax = 31, grid=major, legend to name=grouplegend2, title={50 samples, uninformative prior}, every axis plot/.append style={thick, mark options={solid, scale=0.9}}, yticklabels={}]
            \addplot[winered, mark=triangle] table [x={in_r}, y={err_pca_u_pca_g}, col sep=comma] {data/lorenz/errors-50-1.0e-6-all-1-1.csv};
            \addlegendentry{$\mathrm{PCA}\;\;$}
            \addplot[black, mark=x] table [x={in_r}, y={err_Hu_0.0_ekp_linreg_Hg_0.0_ekp_linreg}, col sep=comma] {data/lorenz/errors-50-1.0e-6-all-1-1.csv};
            \addlegendentry{$\alpha=0\;\;$}
            \addplot[black, mark=o] table [x={in_r}, y={err_Hu_0.5_ekp_linreg_Hg_0.0_ekp_linreg}, col sep=comma] {data/lorenz/errors-50-1.0e-6-all-1-1.csv};
            \addlegendentry{$\alpha=0.5\;\;$}
            \addplot[black, mark=square] table [x={in_r}, y={err_Hu_1.0_ekp_linreg_Hg_0.0_ekp_linreg}, col sep=comma] {data/lorenz/errors-50-1.0e-6-all-1-1.csv};
            \addlegendentry{$\alpha=1\;\;$}
            \addplot[ferngreen, mark=o] table [x={in_r}, y={err_Hu_half_ekp_linreg_Hg_0.0_ekp_linreg}, col sep=comma] {data/lorenz/errors-50-1.0e-6-all-1-1.csv};
            \addlegendentry{$H^{[0,0.5]}\;\;$}
            \addplot[ferngreen, mark=square] table [x={in_r}, y={err_Hu_all_ekp_linreg_Hg_0.0_ekp_linreg}, col sep=comma] {data/lorenz/errors-50-1.0e-6-all-1-1.csv};
            \addlegendentry{$H^{[0,1]}\;\;$}
        
        \nextgroupplot[xlabel={$r=s$}, ylabel={$\norm{F^\dagger - \mathrm{sample\;mean}}$}, ymin = 0, ymax = 31, grid=major, legend to name=grouplegend2, title={200 samples, informative prior}, every axis plot/.append style={thick, mark options={solid, scale=0.9}}]
            \addplot[winered, mark=triangle] table [x={in_r}, y={err_pca_u_pca_g}, col sep=comma] {data/lorenz/errors-200-1.0-all-1-1.csv};
            \addlegendentry{$\mathrm{PCA}\;\;$}
            \addplot[black, mark=x] table [x={in_r}, y={err_Hu_0.0_ekp_linreg_Hg_0.0_ekp_linreg}, col sep=comma] {data/lorenz/errors-200-1.0-all-1-1.csv};
            \addlegendentry{$\alpha=0\;\;$}
            \addplot[black, mark=o] table [x={in_r}, y={err_Hu_0.5_ekp_linreg_Hg_0.0_ekp_linreg}, col sep=comma] {data/lorenz/errors-200-1.0-all-1-1.csv};
            \addlegendentry{$\alpha=0.5\;\;$}
            \addplot[black, mark=square] table [x={in_r}, y={err_Hu_1.0_ekp_linreg_Hg_0.0_ekp_linreg}, col sep=comma] {data/lorenz/errors-200-1.0-all-1-1.csv};
            \addlegendentry{$\alpha=1\;\;$}
            \addplot[ferngreen, mark=o] table [x={in_r}, y={err_Hu_half_ekp_linreg_Hg_0.0_ekp_linreg}, col sep=comma] {data/lorenz/errors-200-1.0-all-1-1.csv};
            \addlegendentry{$H^{[0,0.5]}\;\;$}
            \addplot[ferngreen, mark=square] table [x={in_r}, y={err_Hu_all_ekp_linreg_Hg_0.0_ekp_linreg}, col sep=comma] {data/lorenz/errors-200-1.0-all-1-1.csv};
            \addlegendentry{$H^{[0,1]}\;\;$}

        \nextgroupplot[xlabel={$r=s$}, ymin = 0, ymax = 31, grid=major, legend to name=grouplegend2, title={200 samples, uninformative prior}, every axis plot/.append style={thick, mark options={solid, scale=0.9}}, yticklabels={}]
            \addplot[winered, mark=triangle] table [x={in_r}, y={err_pca_u_pca_g}, col sep=comma] {data/lorenz/errors-200-1.0e-6-all-1-1.csv};
            \addlegendentry{$\mathrm{PCA}\;\;$}
            \addplot[black, mark=x] table [x={in_r}, y={err_Hu_0.0_ekp_linreg_Hg_0.0_ekp_linreg}, col sep=comma] {data/lorenz/errors-200-1.0e-6-all-1-1.csv};
            \addlegendentry{$\alpha=0\;\;$}
            \addplot[black, mark=o] table [x={in_r}, y={err_Hu_0.5_ekp_linreg_Hg_0.0_ekp_linreg}, col sep=comma] {data/lorenz/errors-200-1.0e-6-all-1-1.csv};
            \addlegendentry{$\alpha=0.5\;\;$}
            \addplot[black, mark=square] table [x={in_r}, y={err_Hu_1.0_ekp_linreg_Hg_0.0_ekp_linreg}, col sep=comma] {data/lorenz/errors-200-1.0e-6-all-1-1.csv};
            \addlegendentry{$\alpha=1\;\;$}
            \addplot[ferngreen, mark=o] table [x={in_r}, y={err_Hu_half_ekp_linreg_Hg_0.0_ekp_linreg}, col sep=comma] {data/lorenz/errors-200-1.0e-6-all-1-1.csv};
            \addlegendentry{$H^{[0,0.5]}\;\;$}
            \addplot[ferngreen, mark=square] table [x={in_r}, y={err_Hu_all_ekp_linreg_Hg_0.0_ekp_linreg}, col sep=comma] {data/lorenz/errors-200-1.0e-6-all-1-1.csv};
            \addlegendentry{$H^{[0,1]}\;\;$}
        \end{groupplot}
    
        \node at (6.5,-5.9)[below, yshift=-2\pgfkeysvalueof{/pgfplots/every axis title shift}]{\ref*{grouplegend2}};
    \end{tikzpicture}
    \vspace{-.1cm}
    \caption{Distance from the sample mean to the true parameter for the Lorenz `96 problem, with  either PCA (for input PCA) or 0-LIS (for other input reduction methods) in the output space---note that, even in the full space, this distance is not zero}
    \vspace{-.2cm}
    \label{fig:res:lorenz}
\end{figure}
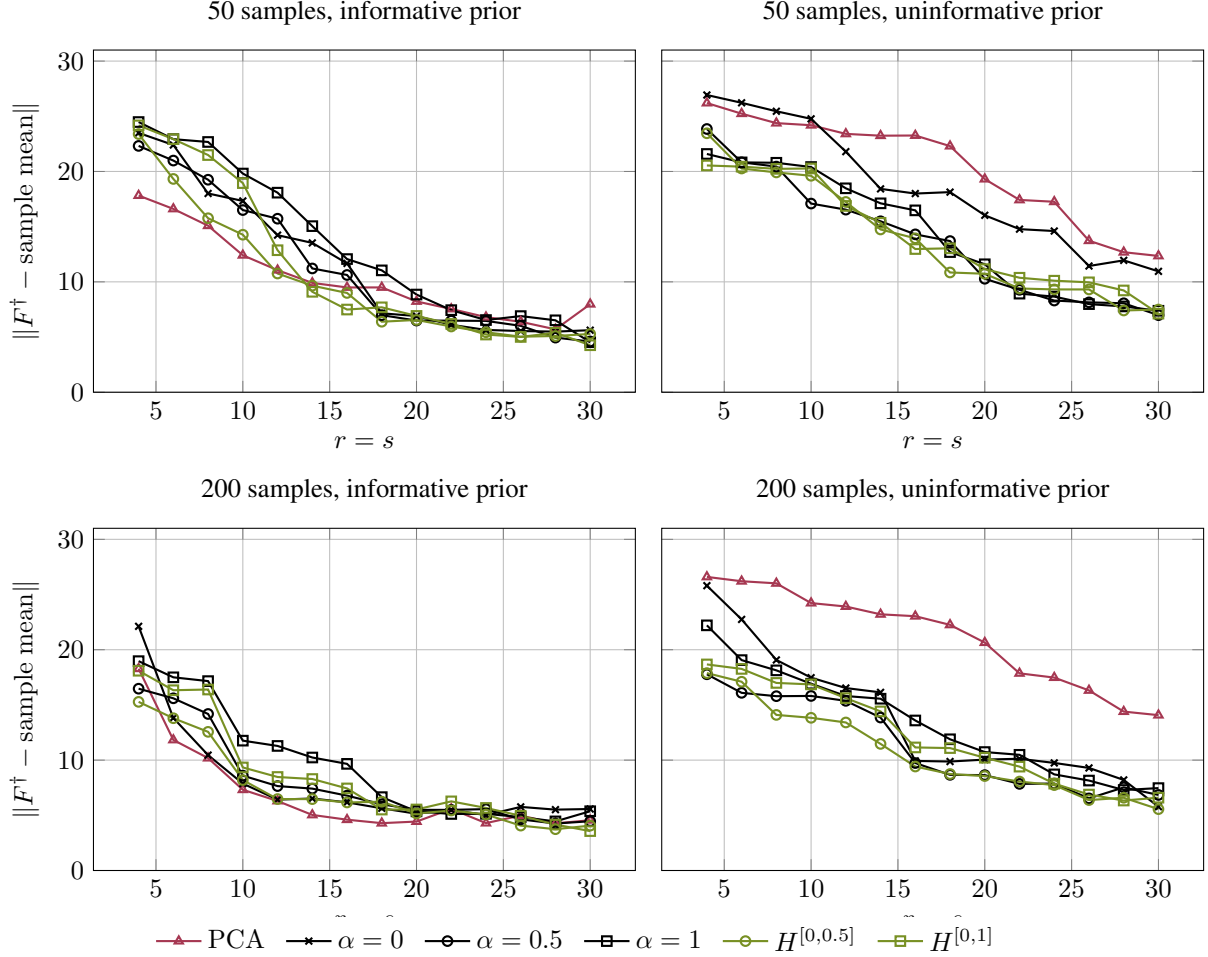

The results of the experiment are shown in \cref{fig:res:lorenz}. We note that PCA performs reasonably well when the prior is informative (where a prior-based approach makes sense), but converges very slowly for uninformative priors. Among the likelihood-informed techniques, either $\alpha=0$ or $\alpha=1$ always performs worst. The strongest performance comes from the accumulated diagnostic matrices, especially when the number of samples per $\alpha$ is very small.

\section{Discussion and recommendations}
For many practical problems, LIS methods have to handle approximate gradients, a low number of samples, and noisy evaluations. Our experiments made clear that theory and results became divided once the algorithms were tested in these increasingly restrictive scenarios. Therefore, we summarize some practical recommendations evidenced by the current experiments.

\subsection{The recommended dimension reduction procedure}
The optimal configuration based on our experiments is as follows, using samples from tempered posteriors $0=\alpha_0 < \dots < \alpha_k$. 
    \begin{enumerate}
        \item \textbf{Whiten.} Whiten the input and output spaces (see \cref{rem:intro:white}). If needed, perform some dimension reduction with PCA to a reasonably sized space where LIS could be applied.
        \item \textbf{Reduce inputs.} Use accumulated $\alpha$-LIS: $H^{[0,\alpha_k]}$.
        \item \textbf{Reduce outputs.} Consider accumulated LIS on all outputs $H^{[0,\alpha_k]}$ if the manifold optimization is computationally feasible. Otherwise, perform DF-LIS ($\alpha=0$).
    \end{enumerate}

\subsection{Other findings}
DB-LIS ($\alpha=1$) is optimal in theory and works well in some practical settings closely approximating the theory (with large sample sizes of the posterior, available derivatives and deterministic model evaluations), such as \cref{fig:res:linear-simple}. However, we see suboptimal performance across multiple experiments outside of these assumptions. It is also the most expensive variant to compute, since samples from the full posterior distribution are needed. Meanwhile, DF-LIS $(\alpha=0)$ is generally robust and well-performing across all scenarios, although it is never the best performer. It is advantageous in that the computation in the output space (after whitening) has a closed-form solution, whereas current algorithms for output space reduction with $\alpha>0$ are expensive.

Our newly proposed $\alpha$-LIS $(0<\alpha<1)$, exemplified by $\alpha=0.5$ in experiments, often approaches the performance of $\alpha=1$ in the theoretically justified settings. In more challenging settings, we find that it even outperforms both extremal values $\alpha=0$ and $\alpha=1$ (see \cref{fig:res:linear-noise,fig:res:darcy})---especially with a low number of samples, as in \cref{fig:res:linear-combined-H}.

The benchmark PCA is scalable to very high dimensions and performance is natively robust under noisy forward maps and lack of derivatives. In the output space, we see good (but usually suboptimal) performance overall. In the input space, performance is (unsurprisingly) strongly tied to how informative the prior is. It can perform well where dimensions correspond to pointwise representations of functions on a numerical grid, which makes for a strong prior (see \cref{sec:ces:lorenz}). 

The derivative-free extension allows the methods to be applicable even without gradients of $G$. A significant cost of accuracy is paid here, as seen in \cref{fig:res:linexp-comp-grads}. Nevertheless, the derivative-free extension to LIS still generally improves significantly upon using PCA. The accumulation extension appears to be very effective for robustness, as it utilizes all data, and does not share the the performance sensitivity observed for the $1$-LIS; it is the best performing method for the most challenging scenarios (noisy evaluations, no derivatives, small sample sizes), as it can effectively utilize all of the available data. Moreover, given enough samples, \cref{fig:res:darcy-num-samples} shows that $\alpha$-LIS, even without accumulation, can largely overcome the issues of noise and approximate gradients.

Our methodology can be transplanted into the CES framework: by reducing dimensions before the emulation stage, the emulator becomes cheaper \emph{and} more accurate. \Cref{sec:ces:lorenz} demonstrated that performance for posterior sampling also translates to performance in training emulators for posterior sampling under the most challenging setting: derivative-free and with noisy forward map.

\section{Conclusions and outlook}
With this investigation we advance the utility of likelihood-informed subspaces in modern Bayesian workflows. In particular, we prove a generalization of current approaches: $\alpha$-LIS, with $\alpha\in[0,1]$, links DF-LIS ($\alpha=0$) and DB-LIS ($\alpha=1$) methods. Our method aims to reduce the KL divergence between true and reduced posteriors with $y$ averaged over $p_\alpha$, a distribution that naturally induces a relationship to tempered posteriors $\pi^{(\alpha)}(x\mid y^\dagger)$. These are a common intermediate stage for Bayesian algorithms. We propose two extensions: derivative-free $\alpha$-LIS and learning a subspace that utilizes information from an entire sequence of tempered distributions with $\alpha_0<\dots<\alpha_k$. 

We demonstrate through extensive numerical experiments that in settings consistent with the theory, $\alpha$-LIS performs almost as well as 1-LIS. In challenging settings outside of the theory (small sample sizes, forward map noise, or a lack of derivatives), $\alpha$-LIS with accumulation shows superior performance, as it makes better use of the available data than alternative methods. We show that these findings translate from direct posterior sampling to the CES framework, which includes emulator training in the subspace and naturally provides the tempered samples needed for $\alpha$-LIS.

There are a few directions that warrant further investigation. \begin{itemize}
    \item It would be interesting to understand why larger $\alpha$s (particularly, $\alpha=1$) are far more susceptible to noise and approximations than $\alpha=0$. Is there an optimal $\alpha$ for a given setting?
    \item The statistical-linearization (SL) approach we currently use to approximate derivatives (see \cref{sec:tempered:gradfree}) is rather rudimentary. Many other techniques should be attempted and compared, such as localized SL \cite{WACKER2026116756} or \emph{ensemble-based gradient inference} \cite{doi:10.1137/22M1533281}. These techniques have the potential to better approximate gradients for nonlinear functions in regimes with larger sample sizes, but carry the risk of a curse of dimensionality.
    \item Output space reduction with $\alpha>0$ triggers a costly optimization procedure. We describe such algorithms in \cref{sec:opt}; more work is needed to make those ideas more efficient. Moreover, our experiments show that $\alpha>0$ in the output space improves performance in settings close to the theory, but loses its edge in more constrained problems. It would be interesting to study whether this is an artifact of the optimization procedure, or a fundamental property.
    \item It is possible to choose the reduced dimensions $r$ and $s$ based on eigenvalues of the diagnostic matrix---see, e.g., \cite{BapMarZah22_pre,CuiTon22}. However, we found these thresholds sometimes unreliable in problems with noise, few samples, and approximate gradients. It would be useful to develop robust (adaptive) techniques to select $r$ and $s$ based on an accuracy parameter.
    \item Current $\alpha$-LIS methods construct the diagnostic matrix, and so are naturally limited in memory. In contrast, PCA has many matrix-free variants that are directly computable from very high-dimensional samples, increasing its applicability. An interesting avenue for future work would be to find $\alpha$-LIS spaces without ever explicitly constructing diagnostic matrices.
    \item Motivated in part by the CES application, we used cheap, approximate ensemble samplers to estimate diagnostic matrices. For other DB-LIS methods, algorithms have been proposed to iteratively draw some samples in the current reduced space and then refine the reduced space with those samples~\cite{CuiTon22,CuiMarMarSolSpa14,zahm2022certified,CUI2016363}. Similar methods could be applied to $\alpha$-LIS.
    \item The problem has thus far been framed ``offline'', reducing dimensions after tempered samples have been collected. Another possibility is to do online dimension reduction within the tempering method itself. We envision that this may aid in convergence of the tempering algorithm, as one can construct proposals or updates for the next stage $\alpha_{k+1}$ that are constrained within the currently most informative space based on $\{\alpha_0, \ldots, \alpha_k\}$. It may provide an adaptive alternative to static optimal subspaces (e.g., DF-LIS~\cite{BapMarZah22_pre}, optimal ensemble subspaces~\cite{HarSch25}), or provide a nonlocal, likelihood-informed annealing schedule.
    \item Given the synergies with the CES framework---both $\alpha$-LIS and CES need essentially the same tempered samples---the next step would be to apply $\alpha$-LIS within a large-scale CES pipeline on a true system. For example, building on the success of CES for uncertainty quantification in an intermediate-complexity climate model \cite{DunGarSchStu21,DunHowSch22}, $\alpha$-LIS could extend application of this approach to full-complexity Earth system models with possibly hundreds of parameters.
\end{itemize}

% ...

\section*{Data availability}
    The data that support the findings of this study are openly available at the Zenodo repository \url{https://doi.org/10.5281/zenodo.20298018}. A user-friendly implementation of our algorithm is available in the \texttt{CalibrateEmulateSample.jl} Julia package \cite{Dun_etal24}.

\section*{Funding}
    The work of AB was funded by the Research Foundation -- Flanders (FWO) (grants~1169725N and V408725N). ORAD is supported by Schmidt Sciences, LLC, the U.S. National Science Foundation (grant AGS-1835860), and ONR (grant N00014-23-16161-2654).

\appendix
\section{Derivation of $\alpha$-LIS} \label{sec:tempered-deriv}

Similarly to~\cite[pp.\ 34--35]{BapMarZah22_pre}, we look to derive diagnostic matrices for the objective \eqref{eq:diag:E}. First, we upper-bound our objective with two terms that each depend exclusively on $U_r$ or on $V_s$. We start with the equality
\begin{align*}
    \KL{\pi(\cdot\mid y)}{\pi^*(\cdot\mid y)} &= \int\pi(x\mid y)\log\frac{\pi(x\mid y)}{\pi^*(x\mid y)}\,\mathrm dx\\
    &= \int\pi(x\mid y)\log\frac{\pi(y\mid x)\pi(x)\pi(y_s)}{ \pi(y_s\mid x_r)\pi(x)\pi(y)}\,\mathrm dx\\
    &= \int\pi(x\mid y)\Big(\log\frac{\pi(y\mid x)}{\pi(y)}-\log\frac{\pi(y_s\mid x_r)}{\pi(y_s)}\Big)\,\mathrm dx\\
     &= \int\pi(x\mid y)\Big(\log\frac{\pi(x, y)}{\pi(x)\pi(y)}-\log\frac{\pi(x_r, y_s)}{\pi(x_r)\pi(y_s)}\Big)\,\mathrm dx\\
     &= \int\pi(x\mid y) \Big(i(x; y) - i(x_r; y_s)\Big)\,\mathrm dx,
\end{align*}
where
\begin{equation}
    i(x; y) \coloneqq \log\frac{\pi(x, y)}{\pi(x)\pi(y)}
\end{equation}
denotes the \emph{pointwise mutual information} (PMI) at a point $(x, y)$. Using the chain rule for PMI  we upper-bound
\begin{align*}
    i(x; y) - i(x_r; y_s) &= \left[i(x_r; y_s) + i(x_r; y_\perp\mid y_s) + i(x_\perp; y_s \mid x_r) + i(x_\perp; y_\perp \mid x_r, y_s)\right] - i(x_r; y_s)\\
    &= i(x_r; y_\perp\mid y_s) + \left[i(x_\perp; y \mid x_r) - i(x_\perp; y_\perp\mid x_r, y_s)\right] + i(x_\perp; y_\perp \mid x_r, y_s)\\
    &= [i(x; y_\perp \mid y_s) - i(x_\perp; y_\perp \mid x_r, y_s)] + i(x_\perp; y\mid x_r)
\end{align*}
with
\begin{equation}
    i(x; y\mid z) \coloneqq \log\frac{\pi(x, y\mid z)}{\pi(x\mid z)\pi(y\mid z)},
\end{equation}
whence
\begin{equation}
\begin{aligned}
    &\KL{\pi(\cdot\mid y)}{\pi^*(\cdot|y)}\\
    &\qquad= \int\pi(x\mid y) i(x_\perp; y\mid x_r)\,\mathrm dx + \int\pi(x\mid y) i(x; y_\perp\mid y_s)\,\mathrm dx - \int \pi(x\mid y)i(x_\perp;y_\perp\mid x_r,y_s)\,\mathrm dx\\
    &\qquad\le \int\pi(x\mid y) i(x_\perp; y\mid x_r)\,\mathrm dx + \int\pi(x\mid y) i(x; y_\perp\mid y_s)\,\mathrm dx,
\end{aligned}
\end{equation}
where the last integral of the first equality can be rewritten as an expected KL divergence and, hence, is positive, justifying the inequality. We can now bound our target expectation \eqref{eq:diag:E} as
\begin{equation} \label{eq:diag:two-terms}
\begin{aligned}
    \mathbb E_{y\sim p_\alpha}\left[\KL{\pi(\cdot\mid y)}{\pi^*(\cdot\mid y)} \right] &\le \iint\pi(x\mid y)i(x_\perp; y \mid x_r)\,\mathrm dx\,p_\alpha(y)\,\mathrm dy\\
    &\quad + \iint\pi(x\mid y)i(x; y_\perp \mid y_s)\,\mathrm dx\,p_\alpha(y)\,\mathrm dy.
\end{aligned}
\end{equation}
The first term depends only on $U_r$ and the second only on $V_s$, so we will bound them separately. By optimizing the bounds on each term, we will then recover our reduced subspaces.

\paragraph{First term.}
We start by bounding
\begin{align*}
    \int\pi(x\mid y) i(x_\perp; y\mid x_r)\,\mathrm dx &= \iint\pi(x_\perp\mid y,x_r)i(x_\perp; y\mid x_r)\,\mathrm dx_\perp \pi(x_r\mid y)\,\mathrm dx_r\\
    &= \iint\frac{\pi(x_\perp, y\mid x_r)}{\pi(x_\perp\mid x_r)\pi(y\mid x_r)}i(x_\perp; y\mid x_r)\pi(x_\perp\mid x_r)\,\mathrm dx_\perp \pi(x_r\mid y)\,\mathrm dx_r\\
    &\eqqcolon \iint h(x_\perp) \log h(x_\perp)\pi(x_\perp\mid x_r)\,\mathrm dx_\perp\pi(x_r\mid y)\,\mathrm dx_r\\
    &\le \int\frac{\bar C}2\int\norm{\nabla_{x_\perp}\log h(x_\perp)}^2 h(x_\perp)\pi(x_\perp\mid x_r)\,\mathrm dx_\perp\pi(x_r\mid y)\,\mathrm dx_r,
\end{align*}
defining $h(x_\perp)\coloneqq \exp i(x_\perp; y\mid x_r)$ and using \cref{ass:prior}, since $\int h(x_\perp)\pi(x_\perp\mid x_r)\,\mathrm dx_\perp = \int\frac{\pi(x_\perp,y\mid x_r)}{\pi(y\mid x_r)}\,\mathrm dx_\perp=1$. We can simplify:
\begin{align*}
    \int\pi(x\mid y) i(x_\perp; y\mid x_r)\,\mathrm dx &\le \frac{\bar C}2 \iint\norm{\nabla_{x_\perp}\log h(x_\perp)}^2 h(x_\perp)\pi(x_\perp\mid x_r)\,\mathrm dx_\perp\pi(x_r\mid y)\,\mathrm dx_r\\
    &= \frac{\bar C}2\iint\norm{\nabla_{x_\perp} \log \frac{\pi(x_\perp, y\mid x_r)}{\pi(x_\perp\mid x_r)\pi(y\mid x_r)}}^2\pi(x_\perp\mid y,x_r)\,\mathrm dx_\perp\pi(x_r\mid y)\,\mathrm dx_r\\
    &= \frac{\bar C}2\iint\norm{\nabla_{x_\perp}\log\frac{\pi(y\mid x)}{\pi(y\mid x_r)}}^2\pi(x_\perp\mid y,x_r)\,\mathrm dx_\perp\pi(x_r\mid y)\,\mathrm dx_r\\
    &= \frac{\bar C}2\iint\norm{\nabla_{x_\perp}\log\pi(y\mid x)}^2\pi(x_\perp\mid y,x_r)\,\mathrm dx_\perp\pi(x_r\mid y)\,\mathrm dx_r\\
    &= \frac{\bar C}2\int\norm{\nabla_{x_\perp}\log\pi(y\mid x)}^2\pi(x\mid y)\,\mathrm dx = \frac{\bar C}2\int\norm{\nabla_x\log\pi(y\mid x)U_\perp}^2\pi(x\mid y)\,\mathrm dx.
\end{align*}
The first term in \eqref{eq:diag:two-terms} is thus bounded, by filling in \eqref{eq:bip:likelihood}, as
\begin{equation} \label{eq:diag:X-iint}
\begin{aligned}
    &\iint\pi(x\mid y)i(x_\perp; y \mid x_r)\,\mathrm dx\,p_\alpha(y)\,\mathrm dy\\
    &\quad \le \frac{\bar C}2\iint\norm{\nabla_x\log\pi(y\mid x)U_\perp}^2\pi(x\mid y)p_\alpha(y)\,\mathrm dx\,\mathrm dy\\
    &\quad = \frac{\bar C}2\Tr\left[{U_\perp^\top \left\{\iint\nabla G(x)^\top \Gamma^{-1}(y-G(x))(y-G(x))^\top \Gamma^{-1}\nabla G(x)\pi(x\mid y)p_\alpha(y)\,\mathrm dx\,\mathrm dy\right\}U_\perp}\right]\\
    &\quad\propto \frac{\bar C}2\Tr\biggl[U_\perp^\top \biggl\{\int\nabla G(x)^\top \Gamma^{-1}\left((1-\alpha)\Gamma + \alpha^2(y^\dagger-G(x))(y^\dagger-G(x))^\top \right)\\
    &\qquad\qquad\qquad\qquad\qquad\qquad\qquad\qquad\qquad\qquad\Gamma^{-1}\nabla G(x)\pi(x)\pi(y^\dagger\mid x)^\alpha\,\mathrm dx\biggr\}U_\perp\biggr].
\end{aligned}
\end{equation}
To derive this, we used the fact that $f_\mathrm{gauss}(x; m, C)^\alpha \propto f_\mathrm{gauss}(x; m, \alpha^{-1}C)$ to compute
\begin{align*}
    &\int (y-G(x))(y-G(x))^\top \pi(x\mid y)p_\alpha(y)\,\mathrm dy\\
    &\quad \propto \int (y-G(x))(y-G(x))^\top \pi(x)f_\mathrm{gauss}\left(y; G(x), \Gamma\right)f_\mathrm{gauss}\left(y; y^\dagger, \frac{1-\alpha}\alpha\Gamma\right)\,\mathrm dy\\
    &\quad = \int (y-G(x))(y-G(x))^\top \pi(x) f_\mathrm{gauss}\left(y^\dagger; G(x), \Gamma + \frac{1-\alpha}\alpha\Gamma\right)f_\mathrm{gauss}\left(y; (1-\alpha)G(x) + \alpha y^\dagger; (1-\alpha)\Gamma\right)\,\mathrm dy\\
    &\quad = \pi(x)f_\mathrm{gauss}\left(y^\dagger; G(x), \alpha^{-1}\Gamma\right)\int(y-G(x))(y-G(x))^\top f_\mathrm{gauss}\left(y; (1-\alpha)G(x)+\alpha y^\dagger; (1-\alpha)\Gamma\right)\,\mathrm dy\\
    &\quad \propto \pi(x)\pi(y^\dagger\mid x)^\alpha\int(y-G(x))(y-G(x))^\top f_\mathrm{gauss}\left(y; (1-\alpha)G(x)+\alpha y^\dagger; (1-\alpha)\Gamma\right)\,\mathrm dy.
    \end{align*}
    Evaluating this integral is taking the expectation of a quadratic function $(y-G(x))(y-G(x))^\top $ over a Gaussian PDF. Decomposing $y - G(x) = (y-\mu)+(\mu-G(x))$, for $\mu = (1-\alpha)G(x) +\alpha y^\dagger$, and taking $\Sigma=(1-\alpha)\Gamma$, then in expectation we have
    \begin{align*}
    &\mathbb{E}_{y\sim\mathcal N(\mu,\Sigma)}((y-G(x))(y-G(x))^\top)\\
    &\qquad= \mathbb{E}_{y\sim\mathcal N(\mu,\Sigma)}((y-\mu)(y-\mu)^\top) +(\mu-G(x))(\mu -G(x))^\top \\
    &\qquad= \Sigma + (\mu-G(x))(\mu -G(x))^\top\\
    &\qquad= (1-\alpha)\Gamma + \left((1-\alpha)G(x) + \alpha y^\dagger - G(x)\right)\left((1-\alpha)G(x) + \alpha y^\dagger) - G(x)\right)^\top\\
    &\qquad= (1-\alpha)\Gamma + \alpha^2(y^\dagger-G(x))(y^\dagger-G(x) )^\top.
\end{align*}
Substituting this in we arrive at
\begin{equation}
    \int (y-G(x))(y-G(x))^\top \pi(x\mid y)p_\alpha(y)\,\mathrm dy \propto \pi(x)\pi(y^\dagger\mid x)^\alpha \left[(1-\alpha)\Gamma + \alpha^2(y^\dagger-G(x))(y^\dagger-G(x))^\top \right].
\end{equation}

Following~\cite[Proposition 2]{BapMarZah22_pre} (see also~\cite[Corollary 4.3.39]{Horn_Johnson_2012}), the right-hand side of \eqref{eq:diag:X-iint} is minimized when the columns of $U_r$ are the $r$ leading eigenvectors of
\begin{equation}
    H_{X\mid y^\dagger}^\alpha \coloneqq \int\nabla G(x)^\top \Gamma^{-1}\left((1-\alpha)\Gamma + \alpha^2(y^\dagger-G(x))(y^\dagger-G(x))^\top \right)\Gamma^{-1}\nabla G(x)\pi(x)\pi(y^\dagger\mid x)^\alpha\,\mathrm dx.
\end{equation}
In the case $\alpha=0$, this matches the diagnostic matrix from~\cite{BapMarZah22_pre} when both are applied after whitening (see \cref{rem:intro:white}): $H_{\bar X\mid y^\dagger}^0 = H_{\bar X}$. For $\alpha=1$, we have $H_{\bar X\mid y^\dagger}^1 = H_{\bar X\mid y^\dagger}$. For $0<\alpha<1$, $H_{\bar X\mid y^\dagger}^\alpha$ interpolates non-linearly between these two diagnostic matrices.

\paragraph{Second term.}
Now define $h(x)\coloneqq \exp i(x; y_\perp \mid y_s)$ and consider the second term in \eqref{eq:diag:two-terms}:
\begin{align*}
    &\iint\pi(x\mid y)p_\alpha(y)i(x; y_\perp \mid y_s)\,\mathrm dx\,\mathrm dy\\
    &\qquad = \iint h(x)\log h(x)\pi(x\mid y_s)\,\mathrm dx\,p_\alpha(y)\,\mathrm dy \\
    &\qquad \le \frac{\bar C}2\iint h(x)\norm{\nabla_x\log h(x)}^2\pi(x\mid y_s)\,\mathrm dx\,p_\alpha(y)\,\mathrm dy  = \frac{\bar C}2\iint \norm{\nabla_x\log h(x)}^2\pi(x\mid y)\,\mathrm dx\,p_\alpha(y)\,\mathrm dy\\
    &\qquad = \frac{\bar C}2\iint\norm[\Big]{\nabla_x\log\frac{\pi(x, y_\perp\mid y_s)}{\pi(x\mid y_s)}}^2\pi(x\mid y)\,\mathrm dx\,p_\alpha(y)\,\mathrm dy = \frac{\bar C}2\iint\norm{\nabla_x\log\pi(y_\perp\mid x, y_s)}^2\pi(x\mid y)\,\mathrm dx\,p_\alpha(y)\,\mathrm dy\\ 
    &\qquad = \frac{\bar C}2\iint\norm{\nabla_x\log \pi(y\mid x) - \nabla_x\log \pi(y_s\mid x)}^2\pi(x\mid y)\,\mathrm dx\,p_\alpha(y)\,\mathrm dy\\
    &\qquad = \frac{\bar C}2\iint \norm{\nabla G(x)^\top\Gamma^{-1}(y-G(x))-\nabla G(x)^\top V_s(V_s^\top\Gamma V_s)^{-1}V_s^\top(y-G(x))}^2\pi(x\mid y)\,\mathrm dx\,p_\alpha(y)\,\mathrm dy,
\end{align*}
where we used \cref{ass:joint} (given that $\int h(x)\pi(x\mid y_s)\,\mathrm dx = \int\pi(x\mid y)\,\mathrm dx = 1$) and \cref{cor:joint}. At this point, we introduce $P_s = \Gamma^{-1} - V_s(V_s^\top \Gamma V_s)^{-1}V_s^\top = \Gamma^{-1} - (V_sV_s^\top  \Gamma V_sV_s^\top )^+$ (with ${}^+$ denoting the Moore--Penrose pseudoinverse). Note that $P_s$ is symmetric and satisfies the idempotent relationship $P_s\Gamma P_s = P_s$, making it a $\Gamma$-orthogonal projection. We continue:
\begin{align*}
    &\iint\pi(x\mid y)p_\alpha(y)i(x; y_\perp \mid y_s)\,\mathrm dx\,\mathrm dy\\
    &\qquad\le \frac{\bar C}2\iint \norm{(y-G(x))^\top P_s\nabla_x G(x)}^2\pi(x\mid y)\,\mathrm dx\,p_\alpha(y)\,\mathrm dy\\
    &\qquad = \frac{\bar C}{2}\Tr\left[\iint\nabla G(x)^\top P_s(y-G(x))(y-G(x))^\top P_s\nabla G(x)\pi(x\mid y)\,\mathrm dx\,p_\alpha(y)\,\mathrm dy\right]\\
    &\qquad = \frac{\bar C}2\Tr\left[\int\nabla G(x)^\top P_s\left((1-\alpha)\Gamma + \alpha^2(y^\dagger-G(x))(y^\dagger-G(x))^\top \right)P_s\nabla G(x)\pi(x)\pi(y^\dagger\mid x)^\alpha\,\mathrm dx\right].
\end{align*}
We then define
\begin{equation}
    J(V_s) \coloneqq \Tr\left[\int\nabla G(x)^\top P_s\left((1-\alpha)\Gamma + \alpha^2(y^\dagger-G(x))(y^\dagger-G(x))^\top \right)P_s\nabla G(x)\pi(x)\pi(y^\dagger\mid x)^\alpha\,\mathrm dx\right]
\end{equation}
and recall that minimizing $J$ is needed for minimizing an upper bound on $\mathbb E_{y\sim p_\alpha}[\KL{\pi(\cdot\,\mid y)}{\pi^*(\cdot\,\mid y)}]$.

\section{Solving the optimization problem for output space reduction} \label{sec:opt}
Since $\alpha$-LIS does not provide a closed-form diagnostic matrix that provides a solution for the reduced output space, we have to find ways to numerically solve the optimization problem \eqref{eq:tempered:J}.

\subsection{Optimizing $J$ using manifold optimization} \label{sec:opt:manifold}
Our main approach for minimizing \eqref{eq:tempered:J} is to use numerical optimization algorithms over the Grassmann manifold. For this, it will be helpful to compute gradients of $J$. As $J$'s only dependence on $V_s$ appears through $P_s$, we start by computing its directional derivative in direction $H$, which we will later use to build the Euclidean and Riemannian gradients for $J$.

Setting $A_s\coloneqq V_s^\top \Gamma V_s$,
\begin{equation}
    D_{V_s}P_s[H]=-HA_s^{-1}V_s^\top  -V_sA_s^{-1}H^\top  + V_s(A_s^{-1}(H^\top \Gamma V_s + V_s^\top \Gamma H)A_s^{-1})V_s^\top .
\end{equation}
Considering a symmetric test matrix $C=C^\top $, we can compute (using the cyclic invariance of the trace) that
\begin{align*}
    \Tr\left[CD_{V_s}P_s[H]\right] &= \Tr\left[-CHA_s^{-1}V_s^\top -CV_sA_s^{-1}H^\top  + CV_s(A_s^{-1}(H^\top \Gamma V_s + V_s^\top \Gamma H)A_s^{-1})V_s^\top \right]\\
    &=\Tr\left[(-2CV_sA_s^{-1}+2\Gamma V_sA_s^{-1}V_s^\top CV_sA_s^{-1})^\top H\right]\\
    &=-2\Tr\left[(\Gamma(\Gamma^{-1} - V_sA_s^{-1}V_s^\top )CV_sA_s^{-1})^\top H\right]\\
    &=-2\Tr\left[(\Gamma P_sC V_s(V_s^\top \Gamma V_s)^{-1})^\top H\right].
\end{align*}
Then we may write, with further application of cyclic invariance of the trace,
\begin{equation}
    D_{V_s}J(V_s)[H]= \Tr\left[ \left(
        \int W_s(x)\,\pi(x)\pi(y^\dagger\mid x)^\alpha\,\mathrm dx
    \right)D_{V_s} P_s[H] + \left(
        \int W_s(x)^\top \,\pi(x)\pi(y^\dagger\mid x)^\alpha\,\mathrm dx 
    \right) D_{V_s} P_s[H]\right],
\end{equation}
where
\begin{equation}
    W_s(x) \coloneqq \nabla G(x)\nabla G(x)^\top \,P_s\,\left((1-\alpha)\Gamma + \alpha^2(y^\dagger-G(x))(y^\dagger-G(x))^\top \right).
\end{equation}
The Euclidean gradient $\nabla_{V_s}(\cdot)$ is defined through the relationship $ D_{V_s}J(V_s)[H]= \Tr\left[(\nabla_{V_s}J(V_s))^\top H\right]$. To this end, we define $\mathrm{Sym}(W) = \frac{1}{2}(W+W^\top )$, then apply the result above with $C=2\int \mathrm{Sym}(W_s(x))\,\pi(x)\pi(y^\dagger\mid x)^\alpha\,\mathrm dx$. We may then read off
\begin{equation} \label{eq:diag:nabla}
\begin{aligned}
    \nabla_{V_s} J(V_s) &= -4\Gamma P_s \left(\int \mathrm{Sym}(W_s(x))\,\pi(x)\pi(y^\dagger\mid x)^\alpha\,\mathrm dx\right)\,V_s(V_s^\top  \Gamma V_s)^{-1}.
\end{aligned}
\end{equation}
The Riemannian gradient is now defined as
\begin{equation} \label{eq:diag:gradGr}
    \mathrm{grad}_\mathrm{Gr} J(V_s) = (I - V_sV_s^\top )\nabla_{V_s}J(V_s).
\end{equation}
The expression \eqref{eq:diag:gradGr} allows us to use gradient-based optimizers to find a good $V_s$. In this paper, we will use the Julia package \texttt{Manopt.jl} \cite{Bergmann2022} to find a $V_s$.

\subsection{Optimizing $J$ incrementally} \label{sec:opt:incr}
Compared to using the leading eigenvectors of $H_{X\mid y^\dagger}^\alpha$ or $H_{\bar Y}$, numerically optimizing $J$ to find $V_s$ has two disadvantages: \begin{itemize}
    \item The optimization procedure tends to be significantly more expensive than simply computing eigenvectors of a diagnostic matrix.
    \item Subspaces are not ``nested''. When working with a diagnostic matrix, $V_s = [V_{s-1}\; v_s]$ by definition. This means that one eigendecomposition of the diagnostic matrix gives you access to \emph{all the reduced bases} without needing to do additional computations. In contrast, the minimizer of $J$ in $s$ dimensions might be completely unrelated to the minimizer in $s-1$ dimensions.
\end{itemize}

To remedy the latter issue, we can sacrifice some optimality and use a greedy approximation: we incrementally build nested subspaces $\widetilde V_s$. We start with $\widetilde V_1 = V_1$, the one-dimensional minimizer of $J$. Then we iteratively build $\widetilde V_s = [\widetilde V_{s-1}\; \widetilde v_s]$ with
\begin{equation} \label{eq:sec:opt:incr:argmin}
    \widetilde v_s \coloneqq \argmin_{\norm{v_s} = 1, \widetilde V_{s-1}^\top v_s = 0} J([\widetilde V_{s-1}\; v_s]).
\end{equation}
While $\widetilde V_s$ is not guaranteed to minimize $J$ when $s>1$, we now have a nested set of reduced bases similar to those for diagnostic matrices.

To solve this minimization problem, we first write
\begin{equation}
    v_s \coloneqq \widetilde V_\perp u_s,
\end{equation}
where\footnote{When $s-1=0$, we define $\widetilde V_\perp=I$.} $\widetilde V_{s-1}^\top \widetilde V_\perp =0$, guaranteeing that the orthogonality constraint in \eqref{eq:sec:opt:incr:argmin} is satisfied. Then, we solve
\begin{equation} \label{eq:sec:opt:incr:argmin-u}
    \widetilde u_s \coloneqq \argmin_{\norm{u_s} = 1} J([\widetilde V_{s-1}\; \widetilde V_\perp u_s]).
\end{equation}
To do this with a numerical solver, we can use the fact that
\begin{equation} \label{eq:sec:opt:incr:nabla}
    \nabla_{u_s}J([\widetilde V_{s-1}\; \widetilde V_\perp u_s]) = \widetilde V_\perp^\top \nabla_{[\widetilde V_{s-1}\; \widetilde V_\perp u_s]}J([\widetilde V_{s-1}\; \widetilde V_\perp u_s])e_s
\end{equation}
and, as before,
\begin{equation} \label{eq:sec:opt:incr:grad}
    \mathrm{grad}_\mathrm{Gr}J([\widetilde V_{s-1}\; \widetilde V_\perp u_s]) = (I-u_su_s^\top )\nabla_{u_s}J([\widetilde V_{s-1}\; \widetilde V_\perp u_s]),
\end{equation}
where the gradient is with respect to $u_s$. This incremental approximation is potentially faster than trying to optimize $V_s$ in its entirety. In addition, the spaces are now nested, which is especially beneficial when we need subspaces for multiple values of $s$.

\subsection{Optimizing $J$ using nonlinear eigenvalue problems} \label{sec:opt:nepv}
There is another way to look at the problem of minimizing $J(V_s)$. Stationary points $V_s$ of $J$ must satisfy $\mathrm{grad}_\mathrm{Gr} J(V_s) = 0$, which by \eqref{eq:diag:gradGr} means that $\nabla_{V_s}J(V_s)$ must lie in the column space of $V_s$. Points with this property satisfy
\begin{equation} \label{eq:sec:opt:nepv:stationary}
    \Gamma P_sA(V_s)V_s(V_s^\top \Gamma V_s)^{-1} = V_s\widehat\Lambda
\end{equation}
for some $\widehat\Lambda$ that may depend on $V_s$, with
\begin{equation} \label{eq:sec:opt:nepv:A}
    A(V_s) = \int \mathrm{Sym}(W_s(x))\,\pi(x)\pi(y^\dagger\mid x)^\alpha\,\mathrm dx.
\end{equation}
Let us focus on the common case where $\Gamma=I$ (which occurs, for instance, when having used the whitening techinque from \cref{rem:intro:white}). Here,
\begin{equation} \label{eq:sec:opt:nepv:J}
    J(V_s) = \Tr[P_sA(V_s)] = \Tr[A(V_s)] - \Tr[V_s^\top A(V_s)V_s]
\end{equation}
and \eqref{eq:sec:opt:nepv:stationary} becomes
\begin{equation}
    P_sA(V_s)V_s = V_s\widehat\Lambda;
\end{equation}
since the right-hand side lies in the column space of $V_s$ and the left-hand side in its complement, this implies that $P_sA(V_s)V_s = 0$ or, equivalently,
\begin{equation} \label{eq:sec:opt:nepv:nepv}
    A(V_s)V_s = V_s\Lambda
\end{equation}
for some $\Lambda$ that can be chosen diagonal, since $A(V_s)$ depends only on the column space of $V_s$. Since \eqref{eq:sec:opt:nepv:nepv} is equivalent to the stationarity condition $\mathrm{grad}_\mathrm{Gr} J(V_s) = 0$, the unconstrained optimization problem of \eqref{eq:sec:opt:nepv:J}, $V_s \coloneqq \argmin_{V_s^\top V_s=I}J(V_s)$, can equivalently be written as
\begin{equation} \label{eq:sec:opt:nepv:opt-nepv}
    V_s \coloneqq \argmax_{V_s^\top V_s=I}\,\Tr[\Lambda] - \Tr[A(V_s)] \qquad \text{s.t.} \qquad A(V_s)V_s = V_s\Lambda.
\end{equation}
This is a nonlinear eigenvalue problem with eigenvector dependency (NEPv). An equivalent derivation is possible for finding $V_\perp$ instead of $V_s$, resulting in a slightly more conventional NEPv formulation\footnote{With a slight abuse of notation, $A(V_\perp)$ is defined as $A(V_s)$ where $V_s^\top V_\perp = 0$.}:
\begin{equation} \label{eq:sec:opt:nepv:opt-nepv-compl}
    V_\perp \coloneqq \argmin_{V_\perp^\top V_\perp=I}\,\Tr[\Lambda] \qquad \text{s.t.} \qquad A(V_\perp)V_\perp = V_\perp\Lambda.
\end{equation}
Since often $s < \frac12d_y$ in practice, we mainly work with $V_s$ instead of $V_\perp$ and hence prefer \eqref{eq:sec:opt:nepv:opt-nepv}.


Something similar can be done for the greedy approach. We continue assuming that whitening has happened such that $\Gamma=I$. In \cref{sec:opt:incr}, we proposed to find $\widetilde v_s$ given $\widetilde V_{s-1}$ by using a gradient-based manifold optimizer that relies on
\begin{equation}
    \mathrm{grad}_\mathrm{Gr}J([\widetilde V_{s-1}\; \widetilde V_\perp u_s]) = -4(I - u_su_s^\top )\widetilde V_\perp^\top P_sA([\widetilde V_{s-1}\; \widetilde V_\perp u_s])\widetilde V_\perp u_s,
\end{equation}
where $P_s = (I - [\widetilde V_{s-1}\; \widetilde V_\perp u_s][\widetilde V_{s-1}\; \widetilde V_\perp u_s]^\top )$, which straightforwardly arises by combining \eqref{eq:diag:nabla}, \eqref{eq:sec:opt:incr:nabla}, \eqref{eq:sec:opt:incr:grad}, and \eqref{eq:sec:opt:nepv:A}. Recall that the columns of $\widetilde V_\perp$ are the orthogonal complement of those of $\widetilde V_{s-1}$. As before, this gradient is zero when
\begin{equation} \label{eq:sec:opt:nepv:incr-step1}
    \widetilde V_\perp^\top P_sA([\widetilde V_{s-1}\; \widetilde V_\perp u_s])\widetilde V_\perp u_s = \hat\lambda u_s
\end{equation}
for some $\hat\lambda$. It can be checked that $\widetilde V_\perp^\top P_s = (I-u_su_s^\top )\widetilde V_\perp^\top $, such that solutions have $\hat\lambda=0$ and \eqref{eq:sec:opt:nepv:incr-step1} is equivalent to, for some $\lambda$,
\begin{equation} \label{eq:sec:opt:nepv:incr-step2}
    \widetilde V_\perp^\top A([\widetilde V_{s-1}\; \widetilde V_\perp u_s])\widetilde V_\perp u_s = \lambda u_s.
\end{equation}
So to find the optimal vector $\widetilde v_s$ to add in iteration $s$, we will minimize $J([\widetilde V_{s-1}\; \widetilde V_\perp u_s])$ under the constraint given by \eqref{eq:sec:opt:nepv:incr-step2} and set $\widetilde v_s \coloneqq \widetilde V_\perp u_s$.

Solving NEPv-type problems is an active field of research and finding solutions to equations like \eqref{eq:sec:opt:nepv:opt-nepv}, \eqref{eq:sec:opt:nepv:opt-nepv-compl}, or \eqref{eq:sec:opt:nepv:incr-step2} based on the eigenproblem constraint in a robust way is challenging. We will give an algorithm only for this greedy, incremental approach. To approximate the minimizer we use a variant of the self-consistent field (SCF) iteration, a common approach for solving nonlinear eigenvalue problems. Defining $M(x)\coloneqq (1-\alpha)\Gamma + \alpha^2(y^\dagger-G(x))(y^\dagger-G(x))^\top $, $C(x)\coloneqq\nabla G(x)\nabla G(x)^\top$, and $\nu(x)\coloneqq\pi(x)\pi(y^\dagger\mid x)^\alpha$, we can decompose
\begin{equation}
    A(V) = \int \mathrm{Sym}(C(x)M(x))\nu(x)\,\mathrm dx - \sum_{j} \int \mathrm{Sym}(C(x)v_jv_j^\top M(x))\nu(x)\,\mathrm dx \eqqcolon \bar A_0 + \bar A(V),
\end{equation}
where $\{v_j\}_j$ are the columns of $V$. Here, $\bar A([V_1\;V_2]) = \bar A(V_1)+\bar A(V_2)$ for any $V_1$ and $V_2$. Now write $\widetilde v_s \coloneqq \widetilde V_\perp u_s$ such that
\begin{equation}
    A([\widetilde V_{s-1}\; \widetilde v_s])\widetilde v_s = \bar A_0 + \bar A(\widetilde V_{s-1})\widetilde v_s + \bar A(\widetilde v_s)\widetilde v_s.
\end{equation}
We now note that
\begin{equation}
\begin{aligned}
    \bar A(v)v &= -\left[\int\mathrm{Sym}(C(x)vv^\top M(x))\nu(x)\,\mathrm dx\right]v\\
    &= -\frac12\left[\int(v^\top M(x)v)C(x) + (v^\top C(x)v)M(x)\nu(x)\,\mathrm dx\right]v \eqqcolon \bar B(v)v
\end{aligned}
\end{equation}
for any $v$ and rewrite \eqref{eq:sec:opt:nepv:incr-step2} as
\begin{equation} \label{eq:sec:opt:nepv:incr-step3}
    \widetilde V_\perp^\top (\bar A_0 + \bar A(\widetilde V_{s-1}) + \bar B(\widetilde V_\perp u_s))\widetilde V_\perp u_s = \lambda u_s.
\end{equation}
We solve \eqref{eq:sec:opt:nepv:incr-step3} with a fixed-point iteration, starting with a random vector $u_s^{(0)}$ and matrix $B^{(0)}\coloneqq \bar B(\widetilde V_\perp u_s^{(0)})$. At iteration $k$, we solve the linear eigenvalue problem
\begin{equation}
    \widetilde V_\perp^\top (\bar A_0 + \bar A(\widetilde V_{s-1}) + B^{(k-1)})\widetilde V_\perp u_s^{(k)} = \lambda u_s^{(k)},
\end{equation}
choosing the eigenvector $u_s^{(k)}$ that minimizes $J([\widetilde V_{s-1}\; \widetilde V_\perp u_s^{(k)}])$. We then update
\begin{equation} \label{eq:sec:opt:nepv:iter}
    B^{(k)} \coloneqq (1-\eta^{(k)})B^{(k-1)} + \eta^{(k)} \bar B(\widetilde V_\perp u_s^{(k)}).
\end{equation}
As a measure of convergence, we define
\begin{equation}
    \varepsilon^{(k)} \coloneqq \norm{B^{(k-1)} - \bar B(\widetilde V_\perp u_s^{(k)})}.
\end{equation}
We then iterate \eqref{eq:sec:opt:nepv:iter} until $\varepsilon^{(k)}$ is below a preset threshold of $10^{-4}$, after which we use $u_s^{(k)}$ as an approximate solution $u_s$. The smoothing parameters $\eta^{(k)}$ are determined adaptively according to the following heuristic:
\begin{equation} \label{eq:sec:opt:nepv:eta}
    \eta^{(k)} \coloneqq \begin{cases}
        1.0 & \text{if reset triggered,}\\
        1.0 & \text{if perturbation triggered,}\\
        \min(\eta_{\max}^{(k)}, 1.1\eta^{(k-1)}) & \text{if $\varepsilon^{(k)}/\varepsilon^{(k-1)} \le 1.01$,}\\
        0.5\eta^{(k-1)} & \text{otherwise}
    \end{cases}
\end{equation}
and
\begin{equation}
    \eta_{\max}^{(k)} \coloneqq \begin{cases}
        \max(0.01, 0.8\eta_{\max}^{(k-1)}) & \text{if reset triggered,}\\
        \eta_{\max}^{(k-1)} & \text{otherwise.}
    \end{cases}
\end{equation}
We choose $\eta^{(0)}=\eta_{\max}^{(0)}\coloneqq0.8$. The condition ``perturbation triggered'' is met when no perturbation or reset was performed during the last $20$ iterations and $\varepsilon^{(j)}\ge\varepsilon^{(j-1)}-10^{-4}$ for all $j\ge k-20$. The condition ``reset triggered'' is met when no reset was performed during the last $40$ iterations and $k-40>\argmin_\text{$j$ since last reset}\varepsilon^{(j)}$. In the case of a reset on iteration $k$, we also set $u_s^{(k)}$ to a random vector instead of the solution to \eqref{eq:sec:opt:nepv:iter}. While \eqref{eq:sec:opt:nepv:eta} is a very ad hoc approach and we in no way claim it is robust, it does seem to cause fast convergence of most SCF iterations in our tests. However, we find it is typically still more expensive than the approach from \cref{sec:opt:incr}. In addition, in rare cases, there is no convergence after many resets; then the unconstrained manifold optimization can be used as a fallback.

There are many options to construct an SCF fixed-point iteration here; some seem to never converge, but we have had a fair amount of success with the one presented above. However, more research is needed to find robust, well-understood, and computationally efficient NEPv solvers for our problems. One could consider adapting techniques from, e.g.,~\cite{doi:10.1137/130910014,janssens2025linearizing}.

\subsection{Comparison} \label{sec:opt:comp}
We now compare various ways to approximately solve the manifold optimization problem for finding $V_s$ when $\alpha>0$. We use the linear test problem from \cref{sec:results} with $d_x=d_y=100$ and reduce only the output dimension, with likelihood-informed dimension reduction using $\alpha\in\{0.5, 1\}$. We compare (i) solving the full optimization problem with \texttt{Manopt.jl} \cite{Bergmann2022}, as in \cref{sec:opt:manifold}, (ii) incrementally building nested subspaces, as in \cref{sec:opt:incr}, with \texttt{Manopt.jl}, and (iii) building those nested subspaces with the NEPv-based technique described in \cref{sec:opt:nepv}.

\Cref{fig:res:linear-incr} shows the $W_2$ error between the true and approximate posteriors as a function of the reduced output dimension $s$. For $\alpha=1$, solving the full optimization problem at once slightly outperforms the other approaches, indicating that the optimal subspace might not be nested. For $\alpha=0.5$, all three methods perform similarly until $s$ becomes rather large, where the full optimization plateaus somewhat.

In general, there is only a moderate difference in performance between the various approximate solvers of the optimization problem. For the experiments in \cref{sec:results}, we therefore use incremental subspaces with \texttt{Manopt.jl} as they show robust performance (as opposed to the NEPv approach, which sometimes fails to converge and then needs \texttt{Manopt.jl} as a fallback) and are convenient due to the nested structure.

\begin{figure}
    \centering
        \begin{tikzpicture}
        \begin{groupplot}[group style={{group size=2 by 1, horizontal sep=2cm, vertical sep=2cm}},name=t2, height=0.37\textwidth, width=0.53\textwidth,legend style={transpose legend,legend columns=0,draw=none}]
        \nextgroupplot[ymode=log, xmode=log,log base x=2, grid=major, xlabel={$s$}, ylabel={$W_2$}, legend to name=grouplegend, unbounded coords=discard, every axis plot/.append style={thick, mark options={solid, scale=0.9}}, title={$\alpha=0.5$}, ymin=0.5e-8, ymax=2e1]
            \addplot[color=black, dotted, mark=o] table [x=dim, y={err_α=0.5_α=0.5 (full)}, col sep=comma]{data/linear-incr/W2-100-100-false-true.csv};

            \addplot[color=ferngreen, solid, mark=o] table [x=dim, y={err_α=0.5_α=0.5 (incr.)}, col sep=comma]{data/linear-incr/W2-100-100-false-true.csv};

            \addplot[color=winered, dashed, mark=o] table [x=dim, y={err_α=0.5_α=0.5 (NEPv)}, col sep=comma]{data/linear-incr/W2-100-100-false-true.csv};

        \nextgroupplot[xshift=-.1\textwidth, ymode=log, xmode=log,log base x=2, grid=major, xlabel={$s$}, legend to name=grouplegend, unbounded coords=discard, every axis plot/.append style={thick, mark options={solid, scale=0.9}}, title={$\alpha=1$}, ymin=0.5e-8, ymax=2e1, yticklabels={}]
            \addplot[color=black, dotted, mark=square, forget plot] table [x=dim, y={err_α=1.0_α=1.0 (full)}, col sep=comma]{data/linear-incr/W2-100-100-false-true.csv};
            \addlegendimage{color=black, dotted, no marks}
            \addlegendentry{$\mathrm{full}\;\;$}

            \addplot[color=ferngreen, solid, mark=square, forget plot] table [x=dim, y={err_α=1.0_α=1.0 (incr.)}, col sep=comma]{data/linear-incr/W2-100-100-false-true.csv};
            \addlegendimage{color=ferngreen, solid, no marks}
            \addlegendentry{$\mathrm{incr.\;(\texttt{Manopt.jl})}\;\;$}

            \addplot[color=winered, dashed, mark=square, forget plot] table [x=dim, y={err_α=1.0_α=1.0 (NEPv)}, col sep=comma]{data/linear-incr/W2-100-100-false-true.csv};
            \addlegendimage{color=winered, dashed, no marks}
            \addlegendentry{$\mathrm{incr.\;(NEPv)}$}
        \end{groupplot}

        \node at (6.2,0)[below, yshift=-2\pgfkeysvalueof{/pgfplots/every axis title shift}]{\ref*{grouplegend}};
    \end{tikzpicture}
    \caption{Wasserstein-2 posterior errors as a function of the reduced output dimension $s$, comparing various ways to solve the manifold optimization problem in the output space}
    \label{fig:res:linear-incr}
\end{figure}
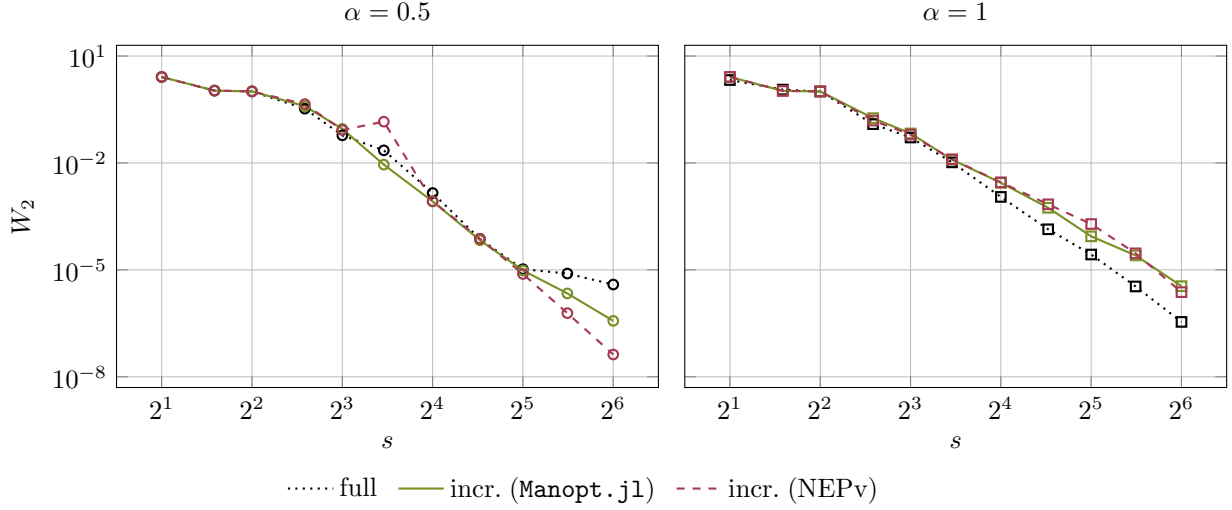

\printbibliography

\end{document}